\RequirePackage{ifpdf}
\documentclass[12pt,a4paper]{article}
\usepackage{amsmath,amssymb,epsfig,jheppubme}
\usepackage{multirow,longtable,enumerate,bm}
\usepackage{booktabs}
\usepackage{xfrac}
\usepackage{lipsum}
\usepackage[table]{xcolor}
\usepackage{amsthm}
\usepackage{dsfont}
\usepackage{mathrsfs}
\usepackage{enumitem}
\usepackage{bbold}
\newcounter{magicrownumbers}
\newcommand\rownumber{\stepcounter{magicrownumbers}\arabic{magicrownumbers}}

\usepackage{tikz-cd}
\usepackage{subcaption}

\newcommand\del{\partial}
\newcommand\bi{\begin{itemize}}
\newcommand\ei{\end{itemize}}

\newcommand\bea{\begin{eqnarray}}
\newcommand\eea{\end{eqnarray}}
\newcommand\be{\begin{equation}}
\newcommand\ee{\end{equation}}

\newcommand\ZZ{\hbox{Z\kern-.4emZ}}
\newcommand\sZZ{\hbox{\sevenfont Z\kern-.4emZ}}

\newcommand{\eref}[1]{Eq.\,(\ref{#1})}
\newcommand{\Comment}[1]{{}}

\newcommand{\mfg}{{\mathfrak g}}
\newcommand{\mfh}{{\mathfrak h}}

\def\IB{\relax{\rm I\kern-.18em B}}
\def\IC{{\mathbb C}}
\def\ID{\relax{\rm I\kern-.18em D}}
\def\IE{\relax{\rm I\kern-.18em E}}
\def\IF{\relax{\rm I\kern-.18em F}}
\def\II{\relax{\rm I\kern-.18em I}}
\def\IZ{\mathbb{Z}}
\def\Id{\relax{1\kern-.32em 1}}
\def\IG{\relax\hbox{$\inbar\kern-.3em{\rm G}$}}
\def\IR{\relax{\rm I\kern-.18em R}}

\usepackage{xparse}

\NewDocumentCommand{\CFT}{m}{${^{(#1)}}$CFT}
\NewDocumentCommand{\A}{o o}{%
  \IfNoValueTF{#1}
    {\mathsf{A}}
    {\IfNoValueTF{#2}
       {\mathsf{A}_{#1}}
       {\mathsf{A}_{#1}^{#2}}
    }
}
\NewDocumentCommand{\B}{o o}{%
  \IfNoValueTF{#1}
    {\mathsf{B}}
    {\IfNoValueTF{#2}
       {\mathsf{B}_{#1}}
       {\mathsf{B}_{#1}^{#2}}
    }
}
\NewDocumentCommand{\C}{o o}{%
  \IfNoValueTF{#1}
    {\mathsf{C}}
    {\IfNoValueTF{#2}
       {\mathsf{C}_{#1}}
       {\mathsf{C}_{#1}^{#2}}
    }
}
\NewDocumentCommand{\D}{o o}{%
  \IfNoValueTF{#1}
    {\mathsf{D}}
    {\IfNoValueTF{#2}
       {\mathsf{D}_{#1}}
       {\mathsf{D}_{#1}^{#2}}
    }
}
\NewDocumentCommand{\E}{o o}{%
  \IfNoValueTF{#1}
    {\mathsf{E}}
    {\IfNoValueTF{#2}
       {\mathsf{E}_{#1}}
       {\mathsf{E}_{#1}^{#2}}
    }
}
\NewDocumentCommand{\F}{o o}{%
  \IfNoValueTF{#1}
    {\mathsf{F}}
    {\IfNoValueTF{#2}
       {\mathsf{F}_{#1}}
       {\mathsf{F}_{#1}^{#2}}
    }
}
\NewDocumentCommand{\G}{o o}{%
  \IfNoValueTF{#1}
    {\mathsf{G}}
    {\IfNoValueTF{#2}
       {\mathsf{G}_{#1}}
       {\mathsf{G}_{#1}^{#2}}
    }
}
\def\SL{\textsl{SL}}
\def\X{X}
\NewDocumentCommand{\SC}{o o}{
    \IfNoValueTF{#1}
        {\mathbf{S}}
        {\IfNoValueTF{#2}
            {\mathbf{S}(#1)}
            {\mathbf{S}(#1)\big/#2}
        }
    }
\NewDocumentCommand{\EC}{o o}{
    \IfNoValueTF{#1}
        {\mathbf{E}}
        {\IfNoValueTF{#2}
            {\mathbf{E}(#1)}
            {\mathbf{E}(#1)\big/#2}
        }
    }
\newtheorem{theorem}{Theorem}
\newtheorem{lemma}[theorem]{Lemma}
\newtheorem{proposition}[theorem]{Proposition}
\newtheorem{definition}[theorem]{Definition}
\newtheorem{example}[theorem]{Example}

\renewcommand{\Im}{{\rm Im\,}}

\title{\boldmath Classification of Unitary RCFTs with Two Primaries and Central Charge Less Than $25$}

\author[a]{Sunil Mukhi}
\author[b,c]{and Brandon C. Rayhaun}

\affiliation[a]{Indian Institute of Science Education and Research,\\ Homi Bhabha Rd, Pashan, Pune 411 008, India}
\affiliation[b]{Stanford Institute for Theoretical Physics, Stanford University\\ 382 Via Pueblo, Stanford, CA 94305, USA}
\affiliation[c]{Yang Institute for Theoretical Physics, Stony Brook University\\ 100 Nicolls Rd, Stony Brook, NY 11794, USA}

\emailAdd{sunil.mukhi@gmail.com}
\emailAdd{brandonrayhaun@gmail.com}

\abstract{We classify all two-dimensional, unitary, rational conformal field theories with two primaries, central charge $c<25$, and arbitrary Wronskian index. In mathematical parlance, we classify all strongly regular vertex operator algebras (VOAs) with central charge $c<25$ and exactly two simple modules.  We find that any such theory is either one of the Mathur--Mukhi--Sen (MMS) theories $\A[1,1]$, $\G[2,1]$, $\F[4,1]$, or $\E[7,1]$, or it is a coset of a chiral algebra with one primary operator (also known as a holomorphic VOA) by such an MMS theory. By leveraging existing results on the classification of holomorphic VOAs, we are able to explicitly enumerate all of the aforementioned cosets and compute their characters. This leads to 123 theories, most of which are new. We emphasize that our work is a bona fide classification of RCFTs, not just of characters. Our techniques are general, and we argue that they offer a promising strategy for classifying chiral algebras with low central charge beyond two primaries.}

\preprint{}


\begin{document}

\maketitle

\section{Introduction}

Conformal field theories in two dimensions --- and the associated mathematical structures of chiral algebras, or vertex operator algebras (VOAs) --- are ubiquitous in physics and mathematics.  Motivations for their study are numerous: they universally describe statistical systems near second order phase transitions, they arise on the world-sheets of string theories, and they encode information about protected sectors in higher-dimensional quantum field theories, to name a few. The present work takes the perspective of the conformal bootstrap, and is concerned with the classification of a special class of such theories: the unitary, rational conformal field theories (RCFTs). 

The classification of RCFTs in its full glory is of course a very infinite problem; one is forced to introduce further qualifiers in order to make progress. It is customary to organize the classification in terms of the central charge $c$, typically starting with low values, and making one's way up. This program is met with striking success  when $0\leq c < 1$, where it is known that every compact unitary CFT (i.e.\ every CFT with a discrete spectrum) is a minimal model \cite{belavin1984infinite,kac1978highest,feigin1984verma,friedan1984conformal,cappelli1987modular,cappelli1987ade}, all of which are rational. At $c=1$, it is believed (but to our knowledge, unproven) that every unitary, compact conformal field theory is\footnote{Both the compact boson and its $\mathbb{Z}_2$ orbifold are irrational CFTs when $R^2$ is an irrational number; every other known $c=1$ theory is rational.} \footnote{Throughout this paper, an ``orbifold of a CFT by a finite group $G$'' refers to a theory obtained by restricting to $G$-invariant states and then adding in twisted-sector states. In other contexts, one might define the orbifold of a chiral algebra/VOA to simply be the $G$-invariant subalgebra.}
\begin{enumerate}[label=\arabic*)]
    \item a compact free boson (parametrized by a radius $R$, up to T-duality),
    \item a $\mathbb{Z}_2$-orbifold thereof, or
    \item one of three isolated theories obtained by orbifolding the compact boson at its self-dual radius by an exceptional finite subgroup \cite{mckay37graphs} of the diagonal of $\textsl{SU}(2)\times \textsl{SU}(2)$  \cite{dijkgraaf1988c,Ginsparg:1987eb,bardakci1988string,ginsparg1988applied}.
\end{enumerate} 
Beyond $c=1$, it is not known even qualitatively what the space of rational conformal field theories looks like.\footnote{There are of course many classes of theories known (lattice VOAs, affine Kac--Moody algebras, W-algebras, etc.), several tools for manipulating them (orbifolds, cosets, tensor products, extensions, etc.), and even lore which claims that applying these tools to the known classes of theories is sufficient to reach all of theory space. Still, we appear to be far from an explicit picture.} 

Another parameter which one can introduce to regain control is the number of primaries\footnote{Throughout this paper, we will use the word ``primary'' as a short-hand in place of ``multiplet of primary operators,'' so that the number of primaries $p$ is the same as the number of irreducible representations of the maximal chiral algebra, which may be larger than Virasoro. We hasten to add that this is a slightly different quantity than the number of (linearly-independent) characters in the theory, which is less than or equal to the number of primaries, because several irreducible representations of the chiral algebra may have the same character, as happens for example with complex conjugate pairs of primaries.\label{primchar}} in the theory, which we call $p$. In particular, one can fix $p$ to be some small value and begin to classify theories with $p$ primaries (which we henceforth denote \CFT{p}s), again from low central charge moving up. Conformal field theories with a finite number of primaries are called rational in the physics literature, and their chiral algebras form a structure which mathematicians call a strongly regular vertex operator algebra, with $p$ the number of simple modules. 

The simplest case is $p=1$, where the sole primary is the identity. In physics, such theories are often called meromorphic CFTs (though see Footnote \ref{footnote:meromorphic} for our slightly differing usage of the term), and in math they are known as holomorphic VOAs. It is known on general grounds that \CFT{1}s can only occur when the central charge is a multiple of 8. The unique chiral algebra at $c=8$ is $\E[8,1]$, the current algebra (or affine VOA) based on $\E[8]$ at level 1. At $c=16$, there are two possibilities: $\E[8,1]\otimes \E[8,1]$ and an extension of $\D[16,1]$ which we will write as $\mathsf{D}_{16,1}^+$. The classification at $c=24$ was first put forward in classic work of Schellekens \cite{schellekens1993meromorphicc}, where he proposed that there are 71 such theories. In particular, in addition to the monster CFT \cite{frenkel1984natural,frenkel1989vertex} and the Leech lattice VOA, there are 69 chiral algebras, each of which can be obtained as a conformal extension of the current algebra generated by its space of dimension one operators, which form a semi-simple Lie algebra. These theories have since been elucidated and the classification put on firmer mathematical footing in a number of papers \cite{van2020construction,van2020dimension,van2021schellekens,hohn2020systematic,moller2019dimension,betsumiya2022automorphism,lam201971,Moller:2021clp,lam2019inertia,Hohn:2017dsm}, with the notable exception of proving that the monster CFT is the unique $c=24$ theory without any continuous symmetry, which remains an open problem.\footnote{Fortunately, none of our results rely on this conjecture being true.} Beyond $c=24$, the landscape of \CFT{1}s becomes unwieldy quite quickly. For example, every positive-definite, even, unimodular lattice of rank $c$ defines a \CFT{1} with central charge $c$, and so there are at least as many theories as there are lattices. The number of such lattices of rank 32, for instance, has been bounded below using (a refinement \cite{king2003mass} of) the Smith--Minkowski--Siegel mass formula, which shows that there are more than a billion. It therefore seems unlikely that we will be witnessing anything resembling a complete classification of \CFT{1}s at $c=32$ or higher any time soon. 

Instead, it is fruitful to ask how far one can get in the next simplest case, $p=2$. Thus, there is one primary besides the identity, whose conformal dimension we  denote by $h$. This program was initiated in \cite{mathur1988classification} and taken further in \cite{mathur1989reconstruction,naculich1989differential,hampapura20162d,gaberdiel2016cosets,tener2017classification,harvey2018hecke,chandra2019towards,chandra2019curiosities,harvey2020galois,grady2020classification}.\footnote{Actually, the cited references often operated in the more general setting of two-character theories, as opposed to two-primary theories. The difference between these two notions has been explained in Footnote \ref{primchar}.} For reasons which will be reviewed below, the classification of \CFT{p}s is often organized by a non-negative integer $\ell$ known as the Wronskian index, rather than by the central charge $c$. It is defined roughly as the number of zeros of the Wronskian determinant of the characters of the \CFT{p} (see \S\ref{subsec:MLDE} for more details), and in the special case that $p=2$, it satisfies the relation
\begin{align}
    \ell = \frac{c}{2}-6h+1.
\end{align}
With two primaries, it is known \cite{naculich1989differential} that $\ell$ is even. If a theory has $\ell=0,$ $2,$ or $4$, then it is said to be \emph{extremal} \cite{tener2017classification} in the sense that its non-identity primary has as large a conformal dimension as is permitted by modularity given its central charge. The (finitely many) physically consistent extremal characters have been completely classified in the above references, and, as our results will confirm,  nearly all \CFT{2}s with these characters have been successfully identified, save for two. It is natural to ask whether one can say anything about the space of \CFT{2}s with $\ell\geq 6$, i.e.\ about \CFT{2}s which are not extremal. 

Leaving the safety of low Wronskian index presents qualitatively new challenges. For example, when $\ell\geq 6$, there can be infinitely many vector-valued modular forms with positive coefficients at a fixed central charge  \cite{harvey2018hecke,chandra2019towards}, and most of them will not correspond to physical theories. The situation is rather analogous to the one encountered in the classification of meromorphic conformal field theories with $c=24$ \cite{schellekens1993meromorphicc}. In that setting, modularity alone fixes the torus partition function only up to a free parameter $\mathcal{N}$, which can be thought of as the number of dimension-one currents,
\begin{align}
    Z(\tau) = j(\tau)-744+\mathcal{N} = q^{-1}+\mathcal{N} + 196884q + \cdots
\end{align}
where $j(\tau)$ is the Klein $j$-invariant. Additional insight was needed to determine which values of $\mathcal{N}$ support conformal field theories, and which do not. Moreover, in many cases it was found that there are multiple theories for a given value of $\mathcal{N}$.

In the present work, we initiate the perilous excursion into the territory of non-extremal \CFT{2}s with $\ell\geq 6$, and therefore confront similar challenges to the ones described in the previous paragraph in the context of two-primary theories. To define a classification question whose answer yields a finite number of theories, we revert back to organizing \CFT{2}s by their central charge instead of by their Wronskian index; in particular, we will see that our methods allow us to obtain a complete classification of two-primary theories with $c<25$, most of which will turn out to have $\ell\geq 6$.\footnote{In fact, inspection of Appendix \ref{app:theories} reveals that the only value of $\ell\geq 6$ which is realized for theories with $c<25$ is $\ell=8$.}

Our approach is to relate the classification of \CFT{p}s with $p>1$ to the classification of \CFT{1}s which, as we have reviewed above, has been completed up to central charge $24$. Specifically, we make heavy use of the following idea \cite{hohn1996selbstduale,Moore:1989vd,schellekens1990simple,creutzig2022gluing,booker2012commutative}, which we state for simplicity in the context of $p=2$ theories, though a version holds when $p>2$ as well (see \S\ref{subsec:coset}).\footnote{The statement of the idea in general involves appealing to modular tensor categories, however $p=2$ theories are simpler in that their modular tensor categories are completely characterized by their associated modular representation, so we can avoid category theory and state things more simply in terms of modular representations. This simplification is not limited to $p=2$, but we restrict to this case anyway because it is the focus of this work.} Let $\mathcal{V}$ be the chiral algebra of a \CFT{2} with central charge $c$, whose characters transform under a modular representation $\varrho:\SL_2(\IZ)\to \textsl{GL}_2(\IC)$, i.e.\ whose modular data is $\mathcal{S}=\varrho(S)$ and $\mathcal{T}=\varrho(T)$, where $S=(\begin{smallmatrix}0&-1\\ 1&0\end{smallmatrix})$ and $T=(\begin{smallmatrix} 1 & 1 \\ 0 & 1\end{smallmatrix})$. Instead of repeatedly saying these words, we will simply say that $\mathcal{V}$ belongs to the genus $(c,\varrho)$, which we sometimes abbreviate to $\mathcal{V}\in(c,\varrho)$. Now, consider a representation of the form 
\begin{align}
    \tilde{\varrho}= \omega^n\varrho^\ast
\end{align} 
for some $n\in \{0,1,2\}$, where $\omega:\SL_2(\IZ)\to\mathbb{C}^\ast$ is defined by the assignments $\omega(S)=1$ and $\omega(T) = e^{-2\pi i \frac{n}{3}}$;  we refer to $\tilde{\varrho}$ as a \emph{conjugate twist} of $\varrho$. Then, we have the following claim, which characterizes all theories belonging to the genus $(\tilde{c},\tilde{\varrho})$.\\

\noindent \textbf{The gluing principle.} \emph{Any} theory $\widetilde{\mathcal{V}}$ in the genus $(\tilde{c},\tilde{\varrho})$ can be glued to $\mathcal{V}$, which belongs to the genus $(c,\varrho)$, to produce a $p=1$ chiral algebra $\mathcal{A}$ in the genus $(C,\omega^n)$, with 
\begin{align}
   \mathcal{A} = ( \mathcal{V}\otimes \widetilde{\mathcal{V}})\oplus( \mathcal{V}_1 \otimes \widetilde{\mathcal{V}}_1), \ \ \ \ \  C=c+\tilde{c}, \ \ \ \ \ \omega^n = \varrho^T\tilde{\varrho},
\end{align}
where $\mathcal{V}_1$ and $\widetilde{\mathcal{V}}_1$ are the non-vacuum modules of $\mathcal{V}$ and $\widetilde{\mathcal{V}}$ respectively.\footnote{A version of this statement holds for $p>2$ as well. However, it is not sufficient to just check that the modular data of $\mathcal{V}$ and $\widetilde{\mathcal{V}}$ are conjugate twists of one another, one must further check that their full modular tensor categories are dual \cite{frohlich2006correspondences}. If this is the case, then the tensor product $\mathcal{V}\otimes \widetilde{\mathcal{V}}$ can be extended to a chiral algebra $\mathcal{A}$ with one primary.} Conversely,  any  $\widetilde{\mathcal{V}}$ in the genus $(\tilde{c},\tilde\varrho)$ can be obtained as a coset \cite{goddard1985virasoro,goddard1986unitary,frenkel1992vertex} (see \cite{gaberdiel2016cosets} for directly relevant computations) of some $\mathcal{A}$ in $(C,\omega^n)$ by $\mathcal{V}$ in $(c,\varrho)$, i.e.
\begin{align}\label{eqn:coset}
\begin{split}
    \widetilde{\mathcal{V}} \cong \mathcal{A}\big/\mathcal{V}, \ \ \ \ \  \tilde{c} = C-c, \ \ \ \ \ \tilde{\varrho}=\omega^n\varrho^\ast.
\end{split}
\end{align}
Thus, knowledge of even a single theory $\mathcal{V}\in (c,\varrho)$, along with knowledge about how it embeds into every $\mathcal{A}\in (C,\omega^n)$, is sufficient to classify all theories  $\widetilde{\mathcal{V}}\in (\tilde{c},\tilde\varrho)$: one simply enumerates all inequivalent cosets of the form Eq.\ \eqref{eqn:coset}.\\

\noindent The power of this idea when $p$ is small derives from the fact that the physically consistent low-dimensional modular representations $\varrho$ have been classified \cite{rowell2009classification,ng2022reconstruction} (see also \cite{tuba2001representations,mason20082}). In light of this, an effective strategy for the classification of \CFT{p}s with $p$ small and fixed, and $c<24$, is as follows.
\begin{enumerate}[label=\arabic*)]
    \item Find a minimal set of $p$-dimensional seed representations $\{\varrho\}$ with the property that all other physically consistent $p$-dimensional representations can be obtained by repeatedly taking conjugate twists of the seeds. In the case that $p=2$, we will use $\varrho_{\A}$ and $\varrho_{\G}$ as the seeds, which correspond to the theories $\A[1,1]$ and $\G[2,1]$, respectively. (See Eq.\ \eqref{eqn:modularrepAG}.)
    \item Using independent methods, classify all theories which have the seeds as their modular data, in the range of central charge desired. Fortunately, it has already been established \cite{mason2018vertex} that $\A[1,1]$ and $\G[2,1]$ are the unique theories with $c<24$ and modular representations $\varrho_{\A}$ and $\varrho_{\G}$, respectively.
    \item Iteratively classify theories corresponding to the conjugate twists of the seeds using the gluing principle, supplemented with the classification of  \CFT{1}s with central charge $C\leq 24$.
\end{enumerate}
In the case of two-primary theories, it turns out that the classification of modular representations alone is enough to rule out the existence of \CFT{2}s with central charge in the range $24\leq c < 25$. This is why we are able to claim that we have classified theories slightly beyond the $c=24$ threshold attainable using the gluing principle. Our $p=2$ analysis is greatly simplified by the fact that the only cosets that the gluing principle requires us to take feature an affine Kac--Moody algebra in the denominator. Because this is so, the enumeration of the requisite cosets essentially reduces to a problem of computing equivalence classes of embeddings of certain simple Lie algebras, which we are able to carry out successfully.

The full list of theories appears in the tables of Appendix \ref{app:theories}. Mathematicians should understand this as a list of strongly regular VOAs with exactly two simple modules. In total, we find 123 \CFT{2}s with $c<25$. In each case, we are able to compute the basic data of the CFT: the central charge $c$, the conformal dimension $h$ of the non-identity primary, its degeneracy $d$, the Wronskian index $\ell$, the characters $\chi_i(\tau)$, and the Kac--Moody symmetry algebra.\footnote{One could attempt to complete this data with additional information such as primary correlation functions on the sphere and torus, following the methods of \cite{mathur1989differential,mathur1989reconstruction,mukhi2018universal}, though we will not do so here.} Every theory we construct has a non-zero number $\mathcal{N}$ of dimension one currents, but in contrast with the $c\leq 24$ \CFT{1}s with $\mathcal{N}\neq 0$, not every theory we obtain is pure Sugawara:\footnote{A chiral algebra with Kac--Moody subalgebra $(\mathfrak{g}_1)_{k_1}\otimes\cdots\otimes (\mathfrak{g}_n)_{k_n}$ is said to be pure Sugawara if its central charge $c$ is equal to $\sum_i \frac{k_i \mathrm{dim}(\mathfrak{g}_i)}{k_i+g_i}$ where $g_i$ is the dual Coxeter number. That is, $\mathcal{V}$ is pure Sugawara if it is a conformal extension of its Kac--Moody subalgebra.} in particular, we will see that some of them have a stress tensor which receives contributions from one or more $c<1$ minimal models. Thus, contrary to expectations, one could not have hoped to have achieved our results by imitating the techniques of \cite{schellekens1993meromorphicc}.  \\

\subsection{Organization}

The organization of this article is as follows.

In \S\ref{subsec:notation}, we provide a glossary of the various notations that we use throughout the paper.

In \S\ref{sec:background}, we review background material. We start in \S\ref{subsec:RCFTbasics} by giving a lightning quick summary of the basics of rational conformal field theory. In \S\ref{subsec:meromorphic}, we summarize some relevant facts about theories whose only primary operator is the identity. In \S\ref{subsec:coset}, we describe how to glue chiral algebras together, as well as how to take cosets. In \S\ref{subsec:MLDE}, we establish some character-theoretic methods.

We then move on to the classification in \S\ref{sec:classification}. We start by classifying physically admissible two-dimensional modular representations in \S\ref{subsec:modulardata}. We then classify physically admissible characters which transform under these modular representations in \S\ref{subsec:admissiblecharacters}. Finally, we describe how to enumerate the full list of two-primary theories with $c<25$ in \S\ref{subsec:enumeration}.

We conclude by explaining a few directions for future research in \S\ref{sec:future}.

Appendix \ref{app:theories} contains tables which list the theories one obtains using the arguments of \S\ref{sec:classification}. These constitute our principal results. Appendix \ref{app:reps} contains data related to two-dimensional representations of the modular group. Appendix \ref{app:liecurrentalgebras} describes technical facts about Lie algebras and current algebras. Finally, Appendix \ref{app:penumbral} explains the relationship between characters of two-primary theories and holomorphic/skew-holomorphic Jacobi forms, with an eye towards penumbral moonshine.

\subsection{Notation guide}\label{subsec:notation}

\begin{footnotesize}

\begin{list}{}{
	\itemsep -1pt
	\labelwidth 23ex
	\leftmargin 13ex	
	}

\item[$\mathds{1}$] The identity primary of an RCFT.
\item[$\Phi$] The non-identity primary of an RCFT with two primaries.

\item[$\mathcal{A}$] A chiral algebra, usually of a \CFT{1}.

\item[${^{(p)}}$CFT] A unitary RCFT whose maximal chiral algebra has $p$ irreducible representations. 

\item[$c$] The central charge of a chiral algebra.

\item[$\mathcal{C}$] The charge conjugation matrix, $\mathcal{C}=\mathcal{S}^2$.

\item[$\mathscr{C}$] A unitary modular tensor category.

\item[$\mathscr{C}^{(p)}$] A unitary modular tensor category with $p$ isomorphism classes of simple objects.

\item[$(c,\mathscr{C})$] The genus of chiral algebras with central charge $c$ and representation category given by $\mathscr{C}$.

\item[$(c,\varrho)$] The set chiral algebras with central charge $c$ whose modular data is given by the representation $\varrho:\SL_2(\IZ)\to\textsl{GL}_p(\IC)$. 

\item[$\mathsf{D}_{16,1}^+$] The unique extension of $\D[16,1]$ with one irreducible module.

\item[$\mathbb{H}$] The upper half-plane, $\mathbb{H}=\{\tau\in\mathbb{C}\mid\Im(\tau)>0\}$.

\item[$h$] The conformal dimension of the non-identity primary $\Phi$ of an RCFT, or the ``conformal dimension" of a vector-valued modular form which does not necessarily have a CFT interpretation. 

\item[$\hat h$] The conformal dimension modulo 1 of the non-identity primary $\Phi$, or the ``conformal dimension'' modulo 1 of a vector-valued modular form which does not necessarily have a CFT interpretation. 
	
\item[$h_c^{(\ell)}$] The ``conformal dimension'' of $X^{(\ell)}_{c,N}(\tau)$, determined by $X^{(\ell)}_{c,N,1}(\tau) = N q^{-\frac{c}{24}+h_c^{(\ell)}}+O(q^{-\frac{c}{24}+h_c^{(\ell)}+1}).$

\item[$\mathsf{L}_c$] A simple Virasoro VOA with central charge $c$.

\item[$\ell$] The Wronskian index of a vector-valued modular form $\chi$. 

\item[$M_{i\bar\jmath}$] The multiplicity matrix of an RCFT, see Eq.\ \eqref{eqn:multiplicitymatrix}.

\item[$m$] The conductor of a chiral algebra $\mathcal{V}$. 

\item[$\mathcal{N}$] The number of currents (i.e.\ dimension one operators) in a chiral algebra $\mathcal{V}$.
	
\item[$N^{(\ell)}_c$] The two values of $N$ which make the matrix $\mathcal{S}_{c,N}^{(\ell)}$ symmetric are $N=\pm N_c^{(\ell)}$. 

\item[$\mathscr{N}_{ij}^k$] The fusion coefficients of an RCFT, or more generally the numbers computed from a modular $\mathcal{S}$-matrix through the Verlinde formula, Eq.\ \eqref{eqn:verlindeformula}, regardless of whether $\mathcal{S}$ is realized in a physical RCFT.

\item[$p$] The number of irreducible representations of a chiral algebra $\mathcal{V}$. We mostly consider the cases $p=1,2$ in this work.

\item 
[$\mathbf{S}(\mathsf{X})$] The unique meromorphic $c=24$ conformal field theory with $\mathsf{X}$ as its current algebra (see \cite{schellekens1993meromorphicc}). ``S'' is for Schellekens.

\item[$\varrho$] A representation of $\SL_2(\IZ)$, typically of dimension 2. 

\item
[$\mathcal{S},\mathcal{T}$] Generic matrices which generate a representation of $\SL_2(\IZ)$ through the assignment $\varrho(S)=\mathcal{S}$ and $\varrho(T)=\mathcal{T}$.

\item[$\varrho_{\mathrm{U}}$, $\varrho_{\mathrm{V}}$, $\varrho_{\mathrm{W}}$] Given a two-dimensional representation $\varrho$, we define $\varrho_{\mathrm{U}} = U\varrho U^\dagger$, $\varrho_{\mathrm{V}}=V\varrho V^\dagger$, and $ \varrho_{\mathrm{W}} = W\varrho W^\dagger$, where $U,V,W$ are defined in Eq.\ \eqref{eqn:UVW}.

\item
[$\mathcal{S}_{\mathrm{U}},\mathcal{T}_{\mathrm{U}},\dots$] The values of $\varrho_{\mathrm{U}},\dots$ on the generators. 

\item[$\varrho_I$] A specially chosen representative of the $I$th equivalence class of two-dimensional $\SL_2(\IZ)$ representations with finite image, where $I=1,\dots, 54$. It is defined on the generators $S$ and $T$ of $\SL_2(\IZ)$ through Eq.\ \eqref{eqn:SMatrix} and Eq.\ \eqref{eqn:TMatrix} using the data from Table \ref{tab:oddweight} and Table \ref{tab:evenweight}.

\item
[$\mathcal{S}_I,\mathcal{T}_I$] The values of $\varrho_I$ on the generators. 

\item[$\varrho^{(\ell)}_{c,N}$] The representation defined by the matrices $\mathcal{S}^{(\ell)}_{c,N}$, $\mathcal{T}_{c,N}^{(\ell)}$. 

\item 
[$\mathcal{S}^{(\ell)}_{c,N}, \mathcal{T}^{(\ell)}_{c,N}$] The matrices defined in Eq.\ \eqref{eqn:S0}, Eq.\ \eqref{eqn:S2}, and Eq.\ \eqref{eqn:S4} which describe how the functions $X_{c,N}^{(\ell)}(\tau)$ transform under $\SL_2(\IZ)$.

\item
[$\mathcal{S}^{(\ell)}_{c,\pm},\mathcal{T}_{c,\pm}^{(\ell)}$] The matrices defined in Eq.\ \eqref{eqn:S0pm}, Eq.\ \eqref{eqn:S2pm}, and Eq.\ \eqref{eqn:S4pm} which describe how the functions $X_{c,\pm}^{(\ell)}(\tau)$ transform under $\SL_2(\IZ)$.

\item[$U,V,W$] The unitary matrices $U=\left(\begin{smallmatrix}-1 & 0\\0&1 \end{smallmatrix}\right)$, $V=\left(\begin{smallmatrix}0 & 1 \\ 1 & 0  \end{smallmatrix}\right)$, and $W=UV$.

\item[$\mathcal{V}$] A chiral algebra/VOA, typically strongly regular.

\item
[$\X^{(\ell)}_{c,N,i}(\tau)$] The two solutions $\X^{(\ell)}_{c,N,0}(\tau)$ and $\X^{(\ell)}_{c,N,1}(\tau)$ of an MLDE with Wronskian index $\ell$, defined in Eq.\ \eqref{eqn:Xl0}, Eq.\ \eqref{eqn:Xl2}, and Eq.\ \eqref{eqn:Xl4}.

\item
[$\X^{(\ell)}_{c,N}(\tau)$] The vector-valued function $\X^{(\ell)}_{c,N}(\tau) = \big(\X^{(\ell)}_{c,N,0}(\tau),\X^{(\ell)}_{c,N,1}(\tau)\big)^T.$

\item[$\X_{c,\pm,i}^{(\ell)}(\tau)$] The function $X_{c,N,i}^{(\ell)}(\tau)$ evaluated at $N=\pm N_c^{(\ell)}$, where $N_c^{(\ell)}$ is defined in Eq.\ \eqref{eqn:N0}, Eq.\ \eqref{eqn:N2}, and Eq.\ \eqref{eqn:N4}. 

\item[$\X_{c,\pm}^{(\ell)}(\tau)$]
The vector-valued function $\X_{c,\pm}^{(\ell)}(\tau) = \big( \X_{c,\pm,0}^{(\ell)}(\tau),\X_{c,\pm,1}^{(\ell)}(\tau)  \big)^T$.

\item[$\mathsf{X}_{r}^{(x)}$] A notation which is meant to emphasize that the simple Lie algebra $\mathsf{X}_r$ is embedded into some simple Lie algebra $\mathfrak{g}$ with an embedding index $x$. The simple Lie algebra $\mathfrak{g}$ is typically clear from context.
\item[$\mathsf{X}_{r,k}$] The current algebra of level $k$ based on the simple Lie algebra $\mathsf{X}_r$ of rank $r$ and type $\mathsf{X}\in\{\mathsf{A},\mathsf{B},\mathsf{C},\mathsf{D},\mathsf{E},\mathsf{F},\mathsf{G}\}$.

\item[$\omega$] The one-dimensional representation of $\SL_2(\IZ)$ defined as $\omega=\zeta^8$.

\item[$\zeta$] The one-dimensional representation of $\SL_2(\IZ)$ which assigns $\zeta(T) = e^{2\pi i/12}$ and $\zeta(S) = -i$. It generates the group of one-dimensional representations of $\SL_2(\IZ)$, which is cylic of order 12.

\end{list}

\end{footnotesize}

\clearpage

\section{Preliminaries and Background}\label{sec:background}

We begin in \S\ref{subsec:RCFTbasics} by summarizing some of the basic properties of RCFTs that we will use throughout the rest of the paper. We then move on in \S\ref{subsec:meromorphic} to summarizing the classification of \CFT{1}s with $c\leq 24$. Next, we briefly review the method of coset conformal field theory in \S\ref{subsec:coset}. We conclude in \S\ref{subsec:MLDE} by writing down formulae for vector-valued modular forms coming from modular linear differential equations, which we will use to construct the characters of the \CFT{2}s that we classify. 

\subsection{RCFT basics}\label{subsec:RCFTbasics}

In a 2d conformal field theory (CFT) $\mathcal{H}$, a distinguished role is played by the space of meromorphic (anti-meromorphic) operators $\mathcal{V}$ ($\overline{\mathcal{V}}$). It forms a consistent truncation of the full operator algebra which is known as the \emph{left-moving (right-moving) chiral algebra}. The mathematical axiomatization of a chiral algebra is known as a \emph{vertex operator algebra (VOA)}. We will use these two terms interchangeably. 

In a non-trivial 2d CFT, at least one of the left-moving or right-moving chiral algebras should be non-empty. Indeed, any non-trivial local quantum field theory is expected to have a conserved stress tensor, which is further traceless when the theory is conformal. These two conditions imply in complex coordinates that $T(z)\equiv T_{zz}(z,\bar{z})$ is meromorphic and $\overline{T}(\bar{z})\equiv T_{\bar{z}\bar{z}}(z,\bar{z})$ is anti-meromorphic, so these operators (if they are non-zero) as well as the Virasoro algebras they generate populate $\mathcal{V}$ and $\overline{\mathcal{V}}$, respectively. We can associate a central charge $c$ to $\mathcal{V}$ (respectively $\bar{c}$ to $\overline{\mathcal{V}}$) through the operator product expansion (OPE) of the stress tensors with themselves, 
\begin{align}
\begin{split}
    T(z)T(w) &\sim \frac{c/2}{(z-w)^4} + \frac{2T(w)}{(z-w)^2}+\frac{\partial T(w)}{z-w} \\
    \overline{T}(\bar z) \overline{T}(\bar{w}) &\sim \frac{\bar c/2}{(\bar z - \bar w)^4}+ \frac{2\overline T(\bar w)}{(\bar z-\bar w)^2}+\frac{\bar\partial \overline T(\bar w)}{\bar z-\bar w}
\end{split}
\end{align}
where $\sim$ indicates that we have dropped the regular terms of the OPE. The central charge is always non-negative, $c,\bar c\geq 0$, in a unitary theory.

In this work, we are interested in the situation that $\mathcal{V}$ forms what a mathematician would call a \emph{unitary, strongly regular VOA} (see e.g.\ \cite{Mason:2014gba,dong1997regularity,abe2004rationality} for the definition) and these qualifiers should always be assumed unless it is explicitly stated otherwise. This loosely coincides with $\mathcal{H}$ being what a physicist would call a \emph{unitary, rational conformal field theory (RCFT)}.\footnote{The adjective ``unitary'' here means that the state space $\mathcal{H}$ should be consistent with a positive-definite inner-product. See e.g.\ \cite{fuchs2002tft} for a rigorous formulation of RCFT.}  In particular, such a VOA $\mathcal{V}$ $(\overline{\mathcal{V}})$ has a finite number $p$ of irreducible representations, and on general grounds, the Hilbert space $\mathcal{H}$ can be decomposed into finitely many representations of $\mathcal{V}\otimes\overline{\mathcal{V}}$,
\begin{align}\label{eqn:multiplicitymatrix}
    \mathcal{H} = \bigoplus_{i,\bar{\jmath}=0}^{p-1} M_{i,\bar{\jmath}} \mathcal{V}_i\otimes \overline{\mathcal{V}}_{\bar{\jmath}}
\end{align}
where $i$ $(\bar{\jmath})$ is an index which runs over irreducible modules of $\mathcal{V}$ ($\overline{\mathcal{V}}$), and $M_{i,\bar\jmath}$ is a matrix of non-negative integer multiplicities, with $M_{0,\bar 0} = 1$.   In our conventions, $i=0$ ($\bar\jmath=\bar 0$) is \emph{always} the index of the vacuum representation, i.e.\ the module obtained by considering $\mathcal{V}$ $(\overline{\mathcal{V}})$ as a representation over itself.

Let us simplify the discussion by restricting attention to the left-moving chiral algebra; everything we say will apply (with the appropriate modifications) to the right-moving chiral algebra as well. Define $L_0$ as usual to be the zero-mode of the stress tensor through
\begin{align}
    T(z) = \sum_{n\in\mathbb{Z}} L_n z^{-n-2}.
\end{align}
It is known that the characters of a strongly regular VOA, 
\begin{align}\label{eqn:qexpansioncharacter}
    \chi_i(\tau) = \mathrm{Tr}_{\mathcal{V}_i} q^{L_0-\frac{c}{24}}=q^{-\frac{c}{24}+h_i}\sum_{n\geq 0} a_i(n)q^n,
\end{align}
form a weight-zero weakly-holomorphic vector-valued modular form \cite{zhu1996modular,miyamoto2004modular}. That is, the functions $\chi_i$ transform under modular transformations as
\begin{align}\label{eqn:charactertransform}
    \chi_i(-1/\tau) = \sum_j\mathcal{S}_{ij}\chi_j(\tau), \ \ \ \ \chi_i(\tau+1) = \mathcal{T}_{ii} \chi_i(\tau)
\end{align}
and $\chi_i(\tau)$ is holomorphic as a function of the upper half-plane, with the exception of possibly having singularities as $\tau$ approaches $i\infty$ or a rational number on the real line. The number $h_i$ is called the \emph{conformal dimension} of (the multiplet of primary operators in) $\mathcal{V}_i$, and $h_0 = 0$. Also, $\mathcal{S}$ and $\mathcal{T}$ are both $p\times p$ matrices subject to the following necessary (though not necessarily sufficient) physical consistency conditions in a unitary RCFT.\footnote{There are further physical consistency conditions one may consider imposing on modular data, such as constraints coming from the Frobenius--Schur indicators of primary fields \cite{bantay1997frobenius}. However, in the case of rank-2 modular data, which we study in \S\ref{subsec:modulardata}, we find that the conditions we incorporate into our definition of ``admissible'' precisely weed out all unphysical modular representations.}

\begin{enumerate}[label=\arabic*)]
    \item Both $\mathcal{S}$ and $\mathcal{T}$ are unitary matrices. Further, $\mathcal{T}$ is diagonal, $\mathcal{S}$ is symmetric, and $\mathcal{S}_{0i}>0$ \cite{dijkgraaf1988modular}.
    \item They generate a $p$-dimensional representation $\varrho:\SL_2(\IZ)\to \textsl{GL}_p(\IC)$ of the modular group through the assignment $\varrho(S) = \mathcal{S}$ and $\varrho(T) = \mathcal{T}$, where $S = \left(\begin{smallmatrix} 0 & -1 \\ 1 & 0 \end{smallmatrix}\right)$ and $T = \left(\begin{smallmatrix} 1 & 1 \\ 0 & 1 \end{smallmatrix}\right)$.
    \item The matrix $\mathcal{S}^2 \equiv \mathcal{C}$ is a permutation matrix, known as the \emph{charge-conjugation} matrix.
    \item The kernel of the representation $\varrho$ contains the principal congruence subgroup 
    \begin{align}
        \Gamma(m) = \left\{ \left(\begin{smallmatrix} a & b \\ c & d \end{smallmatrix}\right) \in\SL_2(\IZ) \mid a, d\equiv 1~\mathrm{mod}~m; \ b,c\equiv 0 ~\mathrm{mod}~m \right\}
    \end{align}
    for some positive integer $m$ \cite{bantay2003kernel,dong2015congruence}. The smallest such $m$ is known as the \emph{conductor} of $\mathcal{V}$.
    \item The numbers
    \begin{align}\label{eqn:verlindeformula}
        \mathscr{N}_{ij}^k \equiv \sum_l \frac{\mathcal{S}_{il}\mathcal{S}_{jl}(\mathcal{S}^{-1})_{kl} }{\mathcal{S}_{0l}}
    \end{align}
    computed by the Verlinde formula \cite{verlinde1988fusion,huang2008vertex} are non-negative integers.
\end{enumerate}
\begin{definition}
A pair of matrices $(\mathcal{S},\mathcal{T})$, or a modular representation $\varrho:\SL_2(\IZ)\to\textsl{GL}_n(\IC)$, is said to be admissible if it satisfies the 5 conditions above.
\end{definition}
We remark that our definition of admissibility does \emph{not} require that there is a chiral algebra which realizes $(\mathcal{S},\mathcal{T})$ through the transformation properties of its characters, though it will turn out that all two-dimensional admissible representations are indeed realized by chiral algebras.  We also emphasize that admissibility is an extremely basis-dependent notion: in particular, conjugating an admissible modular representation $(\mathcal{S},\mathcal{T})$ by a unitary matrix generically spoils its admissibility.

It is also possible to study vector-valued modular forms more generally, regardless of whether or not they furnish the characters of an actual RCFT (see e.g.\ \cite{bantay2007vector,gannon2013theory,bantay2006conformal} for a beautiful general theory). The following definitions, which are variations on similar concepts appearing in \cite{chandra2019towards}, will be useful in the sequel.

\begin{definition}
A quasi-character is a vector-valued function $\chi$ with the following properties. 
\begin{enumerate}[label=\arabic*)]
    \item The function $\chi$ is holomorphic in the interior of the upper half-plane, though possibly with singularities at the cusps.
    \item It transforms according to Eq.\ \eqref{eqn:charactertransform}, with $(\mathcal{S},\mathcal{T})$ generating an admissible modular representation.
    \item It admits a $q$-expansion of the form Eq.\ \eqref{eqn:qexpansioncharacter} with $c\geq 0$, $h_i\geq 0$, and $h_0=0$. The numbers $c$ and $h_i$ are referred to as the central charge and conformal dimensions of $\chi$.
    \item The coefficients $a_i(n)$ in Eq.\ \eqref{eqn:qexpansioncharacter} are (not necessarily non-negative) integers and $a_0(0) = 1$. 
\end{enumerate}
If the coefficients $a_i(n)$ of a quasi-character are further non-negative then the quasi-character is said to be an admissible character. 
\end{definition}
\noindent\emph{Remark:}  Our definitions are stronger than their analogs appearing in \cite{chandra2019towards}. For example, we require that a quasi-character transforms under an admissible modular representation, and we also require that the central charge and conformal dimensions are non-negative, whereas op.\ cit.\ did not. \\

\noindent We emphasize that not every admissible character is realized by a unitary, strongly regular VOA, as we will see explicitly in the next subsection.

In this work, our goal is to classify RCFTs. In general, this is a slightly more involved problem than just classifying chiral algebras, because the data of an RCFT $\mathcal{H}$ involves the choice of a chiral algebra $\mathcal{V}$ \emph{and} a multiplicity matrix $M_{i,\bar\jmath}$ consistent with modular invariance.\footnote{Modular invariance is necessary but in general not sufficient \cite{cappelli1987ade,cappelli1987modular, davydov2016unphysical}.} In other words, in addition to classifying chiral algebras, one must also classify multiplicity matrices $M_{i,\bar\jmath}$ for which the torus partition function
\begin{align}
    Z(\tau,\bar\tau) = \sum_{i,\bar\jmath}M_{i\bar\jmath} \chi_i(\tau) \overline{\chi}_{\bar\jmath}(\bar\tau),
\end{align}
is invariant under modular transformations,
\begin{align}
    Z(\tau+1,\bar\tau +1) = Z(\tau,\bar\tau)=Z(-1/\tau,-1/\bar\tau).
\end{align}
However, in the case of $p=1$ and $p=2$ theories, it turns out that there is always a unique multiplicity matrix, $M_{i,\bar\jmath}=\delta_{i,\bar\jmath}$, consistent with modular invariance. Therefore, in this situation, classifying RCFTs is equivalent to classifying chiral algebras, and so we will often say that we are doing the former, even though at a technical level we are almost always working with the latter objects.

\subsection{Meromorphic theories}\label{subsec:meromorphic}

The classification of \CFT{1}s was once thought to be quite a distinct program from that of classifying \CFT{p}s with $p>1$. However in recent years it has emerged that the two are very closely linked \cite{gaberdiel2016cosets,tener2017classification,chandra2019curiosities,harvey2020galois,Das:2021uvd,bae2021conformal,Duan:2022ltz}. This linkage is important for our goal of classifying \CFT{2}s, and will be discussed below at length. In this subsection, we review some of what is known about \CFT{1}s, also known as holomorphic VOAs to mathematicians.

Such theories have a unique character $\chi(\tau)$ which transforms under a one-dimensional representation of $\SL_2(\IZ)$. It is known that the one-dimensional representations of $\SL_2(\IZ)$ form the cyclic group $\mathbb{Z}_{12}=\langle \zeta \rangle$ with generator 
\begin{align}\label{eqn:generatoronedimensional}
    \zeta(T) = e^{\pi i/6}, \ \ \ \ \zeta(S) = -i.
\end{align}
However admissibility in the one-dimensional case requires that $\mathcal{S}>0$, leaving only the representations $\omega\equiv \zeta^8$, $\omega^2$, and the trivial representation, which we analyze in turn.

Because $\omega(T) = e^{-2\pi i/3}$, it follows that \CFT{1}s whose single character transforms covariantly with respect to $\omega$ have central charge $c\equiv 8~\mathrm{mod}~24$. At $c=8$, the unique admissible character is 
\begin{align}
    \chi(\tau) = j(\tau)^{\frac13} = q^{-\frac13}(1+248q+\cdots)
\end{align}
with $j(\tau)$ the Klein $j$-invariant, and the unique chiral algebra with this character is $\E[8,1]$, the current algebra at level 1 based on the exceptional simple Lie algebra $\E[8]$. The uniqueness of this theory follows from the fact that the dimension 1 operators in a unitary RCFT form a reductive Lie algebra whose rank is less than or equal to the central charge \cite{goddard1989meromorphic,dong2004rational}: the unique reductive Lie algebra with dimension $248$ and rank less than or equal to $8$ is $\E[8]$, and it can have level at most 1 because any higher level leads to a central charge larger than $8$. 

\CFT{1}s with modular representation $\omega^2$ have central charge $c\equiv 16~\mathrm{mod}~24$. At $c=16$, the unique admissible character is 
\begin{align}
    \chi(\tau) = j(\tau)^{\frac23}=q^{-\frac23}(1+496q+\cdots),
\end{align}
however there are now two chiral algebras with this character, $\E[8,1]\otimes \E[8,1]$ and the unique extension of $\D[16,1]$ which we will label $\D[16,1]^+$. Again, uniqueness is proved by noting that $\E[8]$ and $\D[16]$ are the two unique reductive Lie algebras with dimension 496 and rank less than or equal to 16.

When the modular representation is the trivial one, which requires that $c\equiv 0~\mathrm{mod}~24$, the character of a \CFT{1} is modular invariant on its own, and one can contemplate forming a completely chiral CFT, which we will refer to as a \emph{meromorphic conformal field theory}.\footnote{Some authors refer to \emph{any} chiral algebra with one irreducible representation as a meromorphic conformal field theory, however we reserve the term for chiral algebras with one irreducible representation and $c\equiv 0 ~\mathrm{mod}~24$ because otherwise the character is only modular invariant up to a phase.\label{footnote:meromorphic}} At $c=24$, there is a one-parameter family of putative modular-invariant partition functions,
\begin{align}
    \chi(\tau)=j(\tau)-744+{\mathcal N} = q^{-1}+\mathcal{N} + 196884+\cdots
\end{align} 
where $\mathcal{N}$ is an integer, interpreted as the number of Noether currents in the theory, which is restricted by admissibility to satisfy $\mathcal{N}\ge 0$. However, among these infinitely many possibilities only a finite number correspond to an actual meromorphic CFT.\footnote{This exemplifies our earlier claim that not all admissible characters are realized by a physical theory. The present work furnishes further examples of $p=2$ admissible characters that do not correspond to any \CFT{2}. Moreover from the methods we use, it is clear that this will be a generic feature for all $p>2$ as well.} 

Indeed, in a classic paper, Schellekens \cite{schellekens1993meromorphicc} conjectured that there are just 70 distinct CFTs with $\mathcal{N}\neq0$ (see \cite{goddard1989meromorphic} for earlier related work), and that they correspond to a still smaller number of values of ${\mathcal N}$, because different CFTs can have global symmetry groups of the same dimension. The unique theory whose global symmetry algebra contains Abelian factors is the Leech lattice VOA with $\mathcal{N}=24$. Schellekens then proved that every other holomorphic VOA $\mathcal{A}$ with $\mathcal{N}\neq 0$ is pure Sugawara (in the sense that the Kac--Moody subalgebra $\mathcal{K}\subset \mathcal{A}$ of each theory has central charge equal to 24), and moreover produced a complete list of 69 possibilities for $\mathcal{K}$. His conjecture was then that no two holomorphic VOAs with $c=24$ share the same Kac--Moody subalgebra, which he supported by showing that any two meromorphic CFTs containing the same $\mathcal{K}$ are identical as $\mathcal{K}$-modules. 

The picture Schellekens painted for the $\mathcal{N}\neq 0$ theories has since been confirmed and clarified in a number of relatively recent papers \cite{van2020construction,van2020dimension,van2021schellekens,hohn2020systematic,moller2019dimension,betsumiya2022automorphism,lam201971,Moller:2021clp,lam2019inertia,Hohn:2017dsm}. The complete list of theories can be found in the table at the end of \cite{schellekens1993meromorphicc}. We will use the notation $\mathbf{S}(\mathsf{X})$ to denote the unique Schellekens theory with Kac--Moody subalgebra given by $\mathsf{X}$. The main gap in the classification of meromorphic CFTs with $c=24$ is the proof of the widely-believed conjecture that the monster CFT \cite{frenkel1984natural,frenkel1989vertex} is the unique such theory with $\mathcal{N}=0$.

As an aside, we comment that it is now understood that these 69 theories can be obtained in a uniform manner by orbifolding the Leech lattice VOA using a ``generalized deep hole construction'' \cite{moller2019dimension} which extends the correspondence between the 23 Niemeier lattices with non-vanishing root system \cite{niemeier1973definite} and the deep holes of the Leech lattice \cite{conway2013sphere}. 

Crucially for our purposes, the automorphism groups of the $c=24$ theories have been completely determined \cite{dong1999automorphism,SHIMAKURA:2020zbb,betsumiya2022automorphism,lamshimakura}. The following proposition summarizes the main property of these automorphism groups which is relevant for our calculations.\footnote{We are grateful to Ching Hung Lam and Hiroki Shimakura for helping us to distill this proposition from their work.}

\begin{proposition}\label{prop:outerauts}
Let $\mathcal{A}$ be any chiral algebra with $c=24$ and one simple module, except for $\mathbf{S}(\A[7,4]\A[1,1]^3)$ and  $\mathbf{S}(\D[6,5]\A[1,1]^2)$. The automorphism group of $\mathcal{A}$ acts transitively on the level 1 simple factors of the Kac--Moody subalgebra of $\mathcal{A}$.
\end{proposition}

Before moving on we remark that, as we mentioned in the introduction, the number of theories quickly explodes when $c>24$. For example, the next non-trivial case is $c=32$, and we can obtain a lower-bound on the number of such theories by lower-bounding the number of even, unimodular lattices of dimension 32, each of which defines a distinct \CFT{1} with $c=32$. The latter bound can be achieved using the Smith--Minkowski--Siegel mass formula, which says that 
\begin{align}
    \sum_{\Lambda} \frac{1}{|\mathrm{Aut}(\Lambda)|} = \frac{|B_{c/2}|}{c}\prod_{1\leq j<c/2}\frac{|B_{2j}|}{4j}
\end{align}
where $B_n$ is a Bernoulli number, and the sum is over even, unimodular lattices $\Lambda$ of dimension $c$. In the case that $c=32$, the right-hand side evaluates to
\begin{align}
\begin{split}
     &\sum_{\Lambda} \frac{1}{|\mathrm{Aut}(\Lambda)|}=\\
     & \ \ \ \ \ \  \frac{4890529010450384254108570593011950899382291953107314413193123}{121325280941552041649762780685623131486814208000000000}
\end{split}
\end{align}
which is on the order of $4\times 10^7$. The number of lattices is at least twice as many, because every lattice has an automorphism group with at least two elements (the identity transformation and the  canonical $v\mapsto -v$ involution). Refinements of these ideas produce a stronger lower bound of a billion \cite{king2003mass}. Thus, $c=24$ is a sort of upper-critical central charge beyond which implementation of our program of reducing the classification of $p>1$ RCFTs to the classification of \CFT{1}s becomes impractical, at least without additional inputs.

\subsection{Genus, cosets, and gluing}\label{subsec:coset}

A very natural invariant of a chiral algebra $\mathcal{V}$ is its \emph{genus} $(c,\mathscr{C})$. Here, $c$ is the central charge, and $\mathscr{C}=\mathrm{Rep}(\mathcal{V})$ is the equivalence class of the modular tensor category (MTC) formed by the representations of $\mathcal{V}$ \cite{moore1989classical,huang2008rigidity}. The genus is named after an analogous invariant of lattices, and it is useful for classification purposes because it is conjectured that the number of chiral algebras in any genus is finite \cite{Hohn:2002dm}. Hence, classifying all chiral algebras in a fixed genus might be an attainable goal in simple cases.

We will sometimes write $\mathscr{C}^{(p)}$ to emphasize that $\mathscr{C}$ has $p$ simple objects, i.e.\ that $\mathcal{V}$ has $p$ irreducible representations. For example, in this paper, our goal is to classify all chiral algebras which belong to genera of the form $(c,\mathscr{C}^{(2)})$ with $c<25$. Every MTC $\mathscr{C}^{(p)}$ gives rise to a modular representation $\varrho:\SL_2(\IZ)\to \textsl{GL}_p(\mathbb{C})$ which is ambiguous up to multiplication by the character $\omega^n$ for some $n$ (though the choice of a chiral algebra, or even just a central charge $c$, fixes this ambiguity). In many cases, there is a unique MTC $\mathscr{C}^{(p)}$ with $\varrho$ as its modular representation (for example, this is true whenever $p=2$), and in these cases we are free to write $(c,\varrho)$ in place of $(c,\mathscr{C}^{(p)})$. To any MTC $\mathscr{C}$ with modular representation $\varrho$, it is possible to associate a dual MTC $\overline{\mathscr{C}}$ with modular representation $\varrho^\ast$ (see e.g.\ \S 6 of \cite{frohlich2006correspondences} for precise definitions).

Consider two chiral algebras $\mathcal{V}$ and $\widetilde{\mathcal{V}}$ in the genera $(c,\mathscr{C})$ and $(\tilde{c},\tilde{\mathscr{C}})$, respectively. In special circumstances to be described below, their tensor product can be extended to a larger chiral algebra,
\begin{align}
    \mathcal{V}\otimes\widetilde{\mathcal{V}}\subset \mathcal{A},
\end{align}
in which case one can decompose the extension into $\mathcal{V}\otimes\widetilde{\mathcal{V}}$-modules as
\begin{align}\label{eqn:gluing}
    \mathcal{A} \cong  \bigoplus_{i,j} d_{i,j} \mathcal{V}_i\otimes \widetilde{\mathcal{V}}_j,
\end{align}
where $\mathcal{V}_i$ and $\widetilde{\mathcal{V}}_j$ are the irreducible representations of $\mathcal{V}$ and $\widetilde{\mathcal{V}}$, respectively. We say that $\mathcal{A}$ is obtained by \emph{gluing} $\mathcal{V}$ to $\widetilde{\mathcal{V}}$.\footnote{This generalizes an analogous notion of gluing for lattices \cite{conway2013sphere} to the chiral algebra setting.} The characters of $\mathcal{A},\mathcal{V}$, and $\widetilde{\mathcal{V}}$ then enjoy a bilinear relation of the form

\begin{align}\label{eqn:bilinearrelation}
    \chi_{\mathcal{A}}(\tau) = \sum_{i,j} d_{i,j} \chi_i(\tau) \widetilde{\chi}_j(\tau).
\end{align}
Mathematically, an extension of the form Eq.\ \eqref{eqn:gluing} is possible when $\mathcal{A}$ is an algebra object of a special kind  \cite{Huang:2014ixa,fuchs2002tft,frohlich2006correspondences} in the Deligne tensor product $\mathscr{C}\boxtimes\tilde{\mathscr{C}}$ of the MTCs associated to $\mathcal{V}$ and $\widetilde{\mathcal{V}}$.

When $\mathcal{V}$ and $\widetilde{\mathcal{V}}$ are dual in the sense that $\tilde{\mathscr{C}}\cong\overline{\mathscr{C}}$, one can label the irreducible modules of $\mathcal{V}$ and $\widetilde{\mathcal{V}}$ in such a way that taking $d_{i,j}=\delta_{i,j}$ in Eq.\ \eqref{eqn:gluing} leads to a chiral algebra $\mathcal{A}$ with only one irreducible module, i.e.\ a chiral algebra in the genus $(C,\textsl{Vect}_{\mathbb{C}})$, where $C=c+\tilde{c}$ and $\textsl{Vect}_{\mathbb{C}}$ is the trivial MTC. For example, when $\mathcal{V}$ and $\widetilde{\mathcal{V}}$ are both $p=2$ theories, all we must check to ensure that $\mathcal{V}$ and $\widetilde{\mathcal{V}}$ are dual is that
\begin{align}\label{eqn:conjugatetwist}
    \tilde{\varrho} = \omega^n\varrho^\ast
\end{align}
where $\varrho$ and $\tilde{\varrho}$ are the modular representations associated to $\mathcal{V}$ and $\widetilde{\mathcal{V}}$, and $n\in\{0,1,2\}$. In this $p=2$ case, the gluing then takes the form 
\begin{align}
    \mathcal{A} \cong  (\mathcal{V}\otimes \widetilde{\mathcal{V}})\oplus (\mathcal{V}_1\otimes \widetilde{\mathcal{V}}_1).
\end{align}
When $\varrho$ and $\tilde{\varrho}$ are related by Eq.\ \eqref{eqn:conjugatetwist}, we will say that $\tilde{\varrho}$ is a \emph{conjugate twist} of $\varrho$. In general, Eq.\ \eqref{eqn:conjugatetwist} holding is not sufficient to determine that $\mathcal{V}\otimes \widetilde{\mathcal{V}}$ can be extended to a $p=1$ chiral algebra $\mathcal{A}$, though it is when $p=2$.

There is a general construction, known as the coset/commutant construction \cite{goddard1986unitary,frenkel1992vertex}, which will be useful for us because it provides a way to recover $\widetilde{\mathcal{V}}$ from knowledge of $\mathcal{A}$ and  $\mathcal{V}\subset\mathcal{A}$.

\begin{definition}
Let $\mathcal{A}$ be a chiral algebra with stress tensor 
\begin{align}
    T(z)=\sum_{n\in\mathbb{Z}} L_nz^{-n-2},
\end{align} 
and $\mathcal{V}\subset\mathcal{A}$ a subalgebra with stress tensor 
\begin{align}
    t(z) = \sum_{n\in\mathbb{Z}} l_n z^{-n-2}
\end{align} 
satisfying $(L_1t)(z)=0$. The coset of $\mathcal{A}$ by $\mathcal{V}$, or the commutant of $\mathcal{A}$ in $\mathcal{V}$, is defined to be the space 
\begin{align}
    \mathcal{A}\big/\mathcal{V} \equiv \textsl{Com}_{\mathcal{A}}(\mathcal{V}) \equiv \{\varphi\in\mathcal{V} \mid l_{-1}\varphi=0 \}.
\end{align}
It is a chiral algebra with stress tensor equal to $T(z)-t(z)$.
\end{definition}
Importantly for this work, if $\mathcal{V},\mathcal{V}'\subset \mathcal{A}$ are two subalgebras of $\mathcal{A}$ which are related by an automorphism $X:\mathcal{A}\to\mathcal{A}$ in the sense that $X(\mathcal{V}) = \mathcal{V}'$, then $X$ also induces an isomorphism of the two cosets $\mathcal{A}\big/\mathcal{V}$ and $\mathcal{A}\big/\mathcal{V}'$ in the obvious way. Thus, in order to enumerate inequivalent cosets of $\mathcal{A}$, a necessary step is to understand equivalence classes of subalgebras of $\mathcal{A}$.  

If one applies the coset construction to a chiral algebra $\mathcal{A}$ that was obtained by gluing $\mathcal{V}$ to $\widetilde{\mathcal{V}}$ as in Eq.\ \eqref{eqn:gluing}, then the coset $\mathcal{A}\big/\mathcal{V}$ will be an extension of $\widetilde{\mathcal{V}}$, namely 
\begin{align}
    \mathcal{A}\big/\mathcal{V} \cong \bigoplus_j M_{0,j} \widetilde{\mathcal{V}}_j.
\end{align}
In particular, if we are given that $\mathcal{A}$ was obtained by gluing two dual theories $\mathcal{V}$ and $\widetilde{\mathcal{V}}$ with a multiplicity matrix $d_{i,j}=\delta_{i,j}$, then we are guaranteed that $\mathcal{A}\big/\mathcal{V}=\widetilde{\mathcal{V}}$. When $\mathcal{A}\big/ \mathcal{V}=\widetilde{\mathcal{V}}$ and $\mathcal{A}\big/\widetilde{\mathcal{V}} = \mathcal{V}$, we say that $(\mathcal{V},\widetilde{\mathcal{V}})$ is a dual pair, or a coset pair, inside $\mathcal{A}$.

In this work, we will mostly specialize to the case that $\mathcal{A}$ is a chiral algebra with only a single primary operator, and $\mathcal{V}$ is an affine Kac--Moody algebra, in which case the condition $(L_1t)(z)=0$ is always satisfied. We note that  \cite{gaberdiel2016cosets} computed  several examples of coset pairs $(\mathcal{V},\widetilde{\mathcal{V}})$ of this kind, with $\mathcal{V}$ and $\widetilde{\mathcal{V}}$ having $p\geq 2$. For the 2 character cases (which included some examples with $p>2$), the characters were explicitly computed in terms of hypergeometric functions (following \cite{naculich1989differential,mathur1989reconstruction}) and the relation \eref{eqn:bilinearrelation} explicitly verified.

\subsection{Modular linear differential equations and quasi-characters}\label{subsec:MLDE}

Our main tool for computing spaces of vector-valued modular forms is homogeneous modular-invariant linear differential equations (MLDEs). The basic strategy is to write down a general MLDE of second-order,  solve it recursively by the Frobenius method, and explicitly determine the modular properties of the solution. 

MLDEs were first proposed as a viable method for classifying  RCFTs in \cite{mathur1988classification}, where the authors studied a simple class of MLDEs (those with Wronskian index $\ell=0$), demanded that the solutions have completely positive $q$-expansions, and then associated physical theories with the resulting characters.\footnote{We will compare our classification to the results of MMS in more detail in \S\ref{subsec:enumeration}.} Our present approach differs slightly in that we will use MLDEs with small Wronskian index ($\ell=0,2,4$) to generate a \emph{basis} of vector-valued modular forms transforming under a given modular representation $\varrho$, but we do not impose positivity on the individual members of that basis; instead, we only impose positivity on  \emph{linear combinations} of the basis elements. Our work constitutes a special application of the techniques developed in \cite{chandra2019towards}.

We begin by reviewing some generalities. The most general second-order MLDE takes the form
\be
\Big(D^2+\phi_{2}(\tau)D+\phi_{4}(\tau)\Big)\,\chi(\tau)=0
\ee
where $D$ is the covariant derivative which acts as
\be
D\equiv \frac{1}{2\pi i}\del_\tau -\frac{k}{12}E_2(\tau)
\ee
on modular forms of weight $k$. This equation generically has two linearly independent solutions of the form of \eref{eqn:qexpansioncharacter}. As we discuss below, these are potentially identified with characters of an RCFT. Therefore we take them to be holomorphic in the interior of the moduli curve $\SL_2(\IZ)\backslash\mathbb{H}$. The coefficient functions $\phi_r(\tau)$ are  modular of weight $r$ with possible poles in the interior of the upper half-plane. This can be seen by defining three Wronskian determinants from the solutions,
\begin{align}
\begin{split}
&\hspace{.9in} W_2(\tau)\equiv\left|\begin{matrix} \chi_0(\tau) &\chi_1(\tau) \\ \del_\tau \chi_0(\tau) &\del_\tau\chi_1(\tau)  \end{matrix}\right|, \\
&W_4(\tau)\equiv\left|\begin{matrix} \chi_0(\tau) &\chi_1(\tau) \\ \del^2_\tau \chi_0(\tau) &\del^2_\tau\chi_1(\tau)  \end{matrix}\right|,~
W_6(\tau)\equiv\left|\begin{matrix} \del_\tau\chi_0(\tau) &\del_\tau\chi_1(\tau) \\ \del^2_\tau \chi_0(\tau) &\del^2_\tau\chi_1(\tau)  \end{matrix}\right|.
\end{split}
\end{align}
These are weakly-holomorphic modular forms of weights 2,4,6 respectively with singular behaviour as $\tau\to i\infty$. It is easy to see that
\be
\begin{split}
\phi_2(\tau)=-\frac{W_4(\tau)}{W_2(\tau)},\quad   \phi_4(\tau)=\frac{W_6(\tau)}{W_2(\tau)}  .
\end{split}
\ee
Thus the maximum number of poles of $\phi_2,\phi_4$ is equal to the number of zeroes of $W_2(\tau)$ in $\SL_2(\mathbb{Z})\backslash\mathbb{H}$. This number is denoted by $\frac{\ell}{6}$, the fractional value being allowed due to the presence of orbifold singularities at the special points $\tau=e^{i\pi/3},i$, where a pole counts with weight $\frac13$ and $\frac12$ respectively. In recent times $\ell$ has come to be known as the Wronskian index of the MLDE. From its definition, we see that $\ell\ge 0$.

The Wronskian $W_2(\tau)$ has modular weight 2, leading behavior 
\begin{align}
    W_2(\tau)\sim q^{-\frac{c}{12}+h}+\cdots
\end{align} 
in its $q$-expansion, and $\frac\ell6$ zeroes.  Together these facts imply a relation between the Wronskian index $\ell$, the central charge $c$, and the single non-trivial conformal dimension $h$,
\be
-\frac{c}{12}+h=\frac{1-\ell}{6},
\ee
or, after re-arranging,
\be
\ell=\frac{c}{2}-6h+1.
\label{RRtwo}
\ee
From this relation, we obtain the bound
\be
\ell < \frac{c}{2}+1.
\label{ellbound}
\ee

Without specifying $\ell$, the number of independent meromorphic modular forms of weight 2 and 4 is unbounded and the MLDE is intractable. Thus, one works case-by-case for each allowed value of $\ell$. It will turn out that taking linear combinations of solutions to MLDEs with $\ell=0,2,4$ is sufficient to obtain all the vector-valued modular forms relevant for our classification, as we will see explicitly in \S\ref{subsec:admissiblecharacters}.

\subsubsection*{\boldmath Quasi-characters with Wronskian index $\ell=0$}
Ref.\ \cite{mathur1988classification} studied the most general MLDE with $\ell=0$, which takes the form
\be
\Big(D^2+\mu E_4 \Big)\chi=0
\label{mmseq}
\ee
where $\mu$ is a free parameter. The solution $X^{(0)}_{c,N}(\tau)$ to this equation can be expressed in terms of hypergeometric functions as
\begin{align}\label{eqn:Xl0}
\begin{split}
    X_{c,N,0}^{(0)}(\tau) &= j(\tau)^{\frac{c}{24}}{_2}F_1\left(-\frac{c}{24},\frac{1}{3}-\frac{c}{24}; \frac{5}{6}-\frac{c}{12};\frac{1728}{j(\tau)}    \right) \\
    X_{c,N,1}^{(0)}(\tau) &= N j(\tau)^{-\frac{1}{6}-\frac{c}{24}}{_2}F_1\left( \frac{1}{6}+\frac{c}{24},\frac{1}{2}+\frac{c}{24}; \frac{7}{6}+\frac{c}{12};\frac{1728}{j(\tau)}  \right)
\end{split}
\end{align}
where $\mu = -\frac{c(c+4)}{576}$. The quasi-character so-produced has central charge $c$, conformal dimension,
\begin{align}\label{eqn:h0}
    h^{(0)}_c=\frac{c+2}{12},
\end{align}
and modular $\mathcal{S}$-matrix
\begin{align}\label{eqn:S0}
\begin{split}
   \mathcal{S}^{(0)}_{c,N}= \left(
\begin{array}{cc}
 \frac12 \csc (\pi \frac{c+2}{12}) & \frac{1}{N}\frac{2^{c/2} \Gamma \left(\frac{4-c}{6}\right) \Gamma
   \left(-\frac{c+2}{12}\right)}{\sqrt{\pi }  \Gamma \left(-\frac{c}{4}\right)} \\
 N\frac{ \Gamma \left(\frac{c+2}{12}\right) \Gamma \left(\frac{c+8}{6}\right)}{2^{c/2}\sqrt{\pi }  c \Gamma \left(\frac{c}{4}\right)} & -\frac12 \csc (\pi \frac{c+2}{12}) \\
\end{array}
\right).
\end{split}
\end{align}
Demanding that this be symmetric (as required by admissibility) implies that $N=\pm N_c^{(0)}$, where
\begin{align}\label{eqn:N0}
    N_c^{(0)}=\left(\frac{2^c c \Gamma \left(\frac{4-c}{6}\right) \Gamma \left(-\frac{c+2}{12}\right) \Gamma \left(\frac{c}{4}\right)}{\Gamma \left(\frac{c+8}{6}\right) \Gamma \left(\frac{c+2}{12}\right) \Gamma
   \left(-\frac{c}{4}\right)}\right)^{\frac12}.
\end{align}
Plugging $N=\pm N_c^{(0)}$ into $\mathcal{S}^{(0)}_{c,N}$, we find
\begin{align}\label{eqn:S0pm}
\begin{split}
   \mathcal{S}^{(0)}_{c,\pm}\equiv\left(
\begin{array}{cc}
 \frac12 \csc (\pi \frac{c+2}{12}) & \pm \sqrt{1-\frac{1}{4}\csc^2(\pi\frac{c+2}{12})}   \\
 \pm\sqrt{1-\frac{1}{4}\csc^2(\pi\frac{c+2}{12})}  & -\frac12 \csc (\pi \frac{c+2}{12}) \\
\end{array}
\right).
\end{split}
\end{align}

\subsubsection*{\boldmath Quasi-characters with Wronskian index $\ell=2$}
Turning now to $\ell=2$ \cite{naculich1989differential,mathur1989reconstruction,gaberdiel2016cosets}, the MLDE again has just one parameter and takes the form
\be
\Big(D^2+\frac13 \frac{E_6}{E_4}+\mu E_4\Big)\chi=0.
\label{elltwoeq}
\ee
The solution can again be expressed in terms of hyper-geometric functions as
\begin{align}\label{eqn:Xl2}
\begin{split}
    X_{c,N,0}^{(2)}(\tau) &= j(\tau)^{\frac{c}{24}}{_2}F_1\left(-\frac{c}{24},\frac{2}{3}-\frac{c}{24}; \frac{7}{6}-\frac{c}{12};\frac{1728}{j(\tau)}    \right) \\
    X_{c,N,1}^{(2)}(\tau) &= N j(\tau)^{\frac{1}{6}-\frac{c}{24}}{_2}F_1\left( -\frac{1}{6}+\frac{c}{24},\frac{1}{2}+\frac{c}{24}; \frac{5}{6}+\frac{c}{12};\frac{1728}{j(\tau)}  \right)
\end{split}
\end{align}
where $\mu = -\frac{c(c-4)}{576}$.
The central charge of $X^{(2)}_{c,N}(\tau)$ is $c$, the conformal dimension is
\begin{align}\label{eqn:h2}
    h_c^{(2)}=\frac{c-2}{12},
\end{align}
and the $\mathcal{S}$-matrix is 
\begin{align}\label{eqn:S2}
\begin{split}
    &\mathcal{S}^{(2)}_{c,N}= \left(
\begin{array}{cc}
 \frac12\csc(\pi\frac{2-c}{12}) & \frac{1}{N}\frac{27\ 2^{\frac{c}{2}-9} (c-16) (c-12) c \Gamma
   \left(\frac{26-c}{6}\right) \Gamma \left(\frac{2-c}{12}\right)}{\sqrt{\pi } (c-14)  \Gamma \left(6-\frac{c}{4}\right)} \\
 N\frac{2^{8-\frac{c}{2}} (c-14)  \Gamma \left(\frac{c-20}{6}\right) \Gamma \left(\frac{c+10}{12}\right)}{27 \sqrt{\pi } (c-16) (c-12) c \Gamma
   \left(\frac{c}{4}-5\right)} & -\frac12\csc(\pi\frac{2-c}{12}) \\
\end{array}
\right).
\end{split}
\end{align}
Demanding that this be symmetric implies that $N=\pm N_c^{(2)}$, where
\begin{align}\label{eqn:N2}
    N_c^{(2)}= 27\cdot 2^{\frac{c-17}{2}}c\frac{(c-16)(c-12)}{c-14}\left(\frac{    \Gamma \left(\frac{26-c}{6}\right) \Gamma \left(\frac{2-c}{12}\right) \Gamma
   \left(\frac{c}{4}-5\right)}{ \Gamma \left(6-\frac{c}{4}\right) \Gamma \left(\frac{c-20}{6}\right) \Gamma \left(\frac{c+10}{12}\right)}\right)^{\frac12}.
\end{align}
Plugging $N=\pm N_c^{(2)}$ back into $\mathcal{S}^{(2)}_{c,N}$, we find 
\begin{align}\label{eqn:S2pm}
    \mathcal{S}^{(2)}_{c,\pm} \equiv \left( \begin{array}{cc} \frac12\csc(\pi\frac{2-c}{12})  & \pm \sqrt{1-\frac14 \csc^2(\pi\frac{2-c}{12})  }\\ \pm \sqrt{1-\frac14 \csc^2(\pi\frac{2-c}{12})  } & -\frac12\csc(\pi\frac{2-c}{12})\end{array}   \right).
\end{align}

\subsubsection*{\boldmath Quasi-characters with Wronskian index $\ell=4$}
Finally, we define functions with $\ell=4$ in terms of the $\ell=0$ solutions as
\begin{align}\label{eqn:Xl4}
    X_{c,N}^{(4)}(\tau) = j(\tau)^{\frac13}X_{c-8,N}^{(0)}(\tau).
\end{align}
The central charge is $c$, the conformal dimension is
\begin{align}\label{eqn:h4}
    h^{(4)}_c=h_{c-8}^{(0)} = \frac{c-6}{12},
\end{align}
and the $\mathcal{S}$-matrix is 
\begin{align}\label{eqn:S4}
    \mathcal{S}^{(4)}_{c,N}=\mathcal{S}^{(0)}_{c-8,N} .
\end{align}
Demanding that this is symmetric gives $N=\pm N_c^{(4)}$, defined by
\begin{align}\label{eqn:N4}
    N^{(4)}_c = N^{(0)}_{c-8},
\end{align}
and plugging this back into the $\mathcal{S}$-matrix, we find 
\begin{align}\label{eqn:S4pm}
    \mathcal{S}^{(4)}_{c,\pm} = \mathcal{S}^{(0)}_{c-8,\pm}.
\end{align}

In each case,  the $\mathcal{T}$-matrix can be determined in terms of the central charge and conformal dimension as
\begin{align}\label{eqn:Tl}
    \mathcal{T}^{(\ell)}_{c,N}=\mathcal{T}^{(\ell)}_c = \mathrm{diag}\left[ \exp\left(-2\pi i \frac{c}{24}\right),  \exp\left(2\pi i\left(- \frac{c}{24}+h^{(\ell)}_c\right)\right) \right].
\end{align}
Since $N$ is uniquely determined by $c$ up to a sign by demanding that $\mathcal{S}^{(\ell)}_{c,N}$ be symmetric, we simplify our notation henceforth by writing
\begin{align}
    \X_{c,\pm}^{(\ell)}(\tau) \equiv \X^{(\ell)}_{c,N}(\tau) \hspace{.5in} (N=\pm N_c^{(\ell)}). \ \ 
\end{align}
In the form that we have expressed the solutions $X_{c,\pm}^{(\ell)}(\tau)$, their modular properties are manifest, though the $q$-expansion and its integrality are not. It is worth emphasizing that the $X_{c,\pm}^{(\ell)}(\tau)$ are not quasi-characters for arbitrary values of $c$ (in particular, their $q$-series are not integral), but rather only for a specific set of allowed values. We identify these values by classifying admissible modular data in the sequel.

\clearpage

\section{Classification}\label{sec:classification}

In this section, we turn to the classification of \CFT{2}s with $c<25$.  In \S\ref{subsec:modulardata}, we classify all two-dimensional admissible modular representations (see \S\ref{subsec:RCFTbasics} for the definition); every such representation is realized as the modular data of a \CFT{2}, and conversely any \CFT{2} has one of these representations as its modular data. Using these representations, in \S\ref{subsec:admissiblecharacters}, we compute all admissible two-component characters with $c<25$ (again, see \S\ref{subsec:RCFTbasics} for the definition). Finally, in \S\ref{subsec:enumeration}, we complete the classification, using the gluing principle explained in \S\ref{subsec:coset}.

\subsection{Modular data}\label{subsec:modulardata}

We begin our investigation of unitary RCFTs with two primaries by classifying all two-dimensional admissible modular representations (see \S\ref{subsec:RCFTbasics} for our definition of admissible representation). This data can be obtained by going through the classification of unitary modular tensor categories with 2 simple objects \cite{rowell2009classification}, however we have opted for a more direct, pedestrian approach. 

We will restrict our attention in the beginning to \emph{irreducible} representations, and then explain at the end why this is justified. As we have reviewed earlier, an admissible modular representation has a non-trivial kernel which contains a principal congruence subgroup, $\ker\varrho \supset \Gamma(m)$, for some integer $m$. In particular, since $\Gamma(m)$ has finite index in $\SL_2(\IZ)$ for every $m$, it follows that $\varrho$ must have finite image in a unitary RCFT. Happily, irreducible two-dimensional representations of $\SL_2(\IZ)$ with finite image have been classified \cite{mason20082}. Each takes the following form (up to equivalence). 

Without loss of generality, we may work in a basis in which $T$ is represented by a diagonal matrix,
\begin{align}\label{eqn:TMatrix}
    \mathcal{T} = \mathrm{diag}(e^{2\pi i m_0},e^{2\pi im_1}).
\end{align}
Then, letting
\begin{align}\label{eqn:kappar}
\begin{split}
    \kappa &=e^{2\pi i (2m_0+m_1)}-e^{2\pi i (m_0+2m_1)}, \\
    r &= -e^{6\pi i (m_0+m_1)}\kappa^2-1,
\end{split}
\end{align}
we define the matrix $\mathcal{S}$ in terms of the eigenvalues of $\mathcal{T}$ as
\begin{align}\label{eqn:SMatrix}
    \mathcal{S} = \kappa^{-1} \left(\begin{array}{cc}1 & \sqrt{r} \\ \sqrt{r} & -1 \end{array}\right).
\end{align}
We note that in all the cases of interest to us, $r$ will be a non-negative real number, and so there is no ambiguity in what we mean by $\sqrt{r}$.

\begin{theorem}[Theorem 3.7 of \cite{mason20082}]
There are 54 equivalence classes of irreducible, two-dimensional representations of $\SL_2(\IZ)$ with finite image. Each has kernel containing $\Gamma(m)$ for some integer $m$, and each is obtained by setting $\varrho(T)=\mathcal{T}$  and $\varrho(S) = \mathcal{S}$ for some choice of $m_0,m_1\in\mathbb{Q}$, where $\mathcal{T}$ and $\mathcal{S}$ are defined in Eq.\ \eqref{eqn:TMatrix} and Eq.\ \eqref{eqn:SMatrix}, respectively. 
\end{theorem}
The 54 choices of $(m_0,m_1)$ which lead to inequivalent representations are summarized in Table \ref{tab:oddweight} and Table \ref{tab:evenweight}. They are numbered according to the order in which they appear in Tables 1--4 of \cite{mason20082}. We will define $\varrho_I$ for $I=1,\dots,54$ to be the $I$th such representation, and also set $\mathcal{S}_I = \varrho_I(S)$ and $\mathcal{T}_I=\varrho_I(T)$. Our goal is to determine which of these representations satisfy the rest of the conditions required by admissibility (possibly after a change of basis).

By explicit computation, one can confirm that in each case, $\mathcal{C}_I=\mathcal{S}_I^2=\pm 1$, which is a basis-independent statement. We can then immediately eliminate the representations with $\mathcal{C}_I=-1$ from consideration because the charge conjugation matrix is not a permutation matrix for these representations. Alternatively, one can rule out such representations as follows. First, note that $S^2$, as a transformation of the upper half-plane, acts trivially on $\tau$. Therefore, modular covariance of the characters of an RCFT dictates that they obey the equation 
\begin{align}
    \chi(\tau) = \mathcal{S}^2\chi(\tau).
\end{align}
Representations with $\mathcal{C}=-1$ therefore will clearly not support (non-vanishing) vector-valued modular forms of weight zero, and hence will not support functions which can serve as characters of a putative RCFT.\footnote{Such representations however will support vector-valued modular forms with odd weight.} The representations with $\mathcal{C}_I=-1$ are the ones appearing in Table \ref{tab:oddweight}; they are precisely the $\varrho_I$ with $I$ an odd integer. Therefore, we restrict our attention in the sequel to the representations $\varrho_I$ with $I$ even, which appear in Table \ref{tab:evenweight} and all satisfy $\mathcal{C}_I=+1$.

Now, in RCFT, it is not sufficient to specify a modular representation up to equivalence. One must also pick a distinguished choice of basis in which $\mathcal{T}$ is diagonal and the representation is unitary. Various physical quantities, such as the fusion rules as computed by the Verlinde formula Eq.\ \eqref{eqn:verlindeformula}, depend on this choice of basis. We remind the reader that we will always interpret the $0$-component as corresponding to the identity primary $\mathds{1}$, and the $1$-component as corresponding to the non-identity primary $\Phi$. 

As it stands, we have already expressed each representation in a unitary basis in which $\mathcal{T}$ is diagonal. However there are still other choices of bases in which the diagonality of $\mathcal{T}$ and the unitarity of the representation are maintained. The following lemma serves to enumerate these bases. 

\begin{lemma}\label{lemma:unitarybases}
Let $\mathcal{T}$ be a diagonal $2\times 2$ matrix with distinct eigenvalues. Any unitary matrix $M$ for which $M\mathcal{T}M^\dagger$ is diagonal takes the form
\begin{align}\label{eqn:generalUnitary}
    M = e^{i\theta}\left(\begin{array}{cc} e^{i\varphi} & 0 \\ 0 & 1  \end{array}\right) \text{ or } M=  e^{i\theta}\left(\begin{array}{rr} 0 & e^{i\varphi} \\ 1 & 0   \end{array}\right).
\end{align}
\end{lemma}

\begin{proof}
Parametrize an arbitrary $2\times 2$ unitary matrix as
\begin{align}
    M=\left(\begin{array}{cc}w & z \\ -e^{i\zeta}z^\ast & e^{i\zeta} w^\ast    \end{array}\right)
\end{align}
where $|z|^2+|w|^2=1$ and $\zeta$ is a real number. If we take $\mathcal{T} = \mathrm{diag}(\lambda_1,\lambda_2)$, then 
\begin{align}
    M\mathcal{T}M^\dagger = \left(\begin{array}{cc} \lambda_1|w|^2+\lambda_2|z|^2 & e^{-i\zeta}(\lambda_2-\lambda_1)wz \\ e^{i\zeta}(\lambda_2-\lambda_1)(wz)^\ast & \lambda_2 |w|^2+\lambda_1|z|^2  \end{array}\right).
\end{align}
In order for this to be diagonal, we require that $e^{-i\zeta}(\lambda_2-\lambda_1)wz=0$. Since the eigenvalues of $\mathcal{T}$ are distinct by assumption, $\lambda_2-\lambda_1\neq 0$ and we must have that either $w=0$ or $z=0$. This then guarantees that the other off-diagonal element vanishes as well. 

If $w=0$, then $z$ must be a phase, and after redefining variables we find that 
\begin{align}
    M = e^{i\theta}\left(\begin{array}{rr} 0 & e^{i\varphi} \\ 1 & 0 \end{array}\right).
\end{align}
On the other hand, if $z=0$, then $w$ must be a phase, and after redefining variables we find that 
\begin{align}
    M =e^{i\theta}\left(\begin{array}{cc} e^{i\varphi} & 0 \\ 0 & 1  \end{array}\right).
\end{align}
This completes the proof.
\end{proof}
Noting that each of the representations $\varrho_I$ have $m_0\neq m_1$, we have shown that the admissible irreducible modular representations must belong to the following set,
\begin{gather}\label{eqn:modularDataSpace}
\begin{split}
    (\mathcal{S},\ \mathcal{T}) &\in  \left\{\Big(U_\varphi\mathcal{S}_IU_\varphi^\dagger,U_\varphi\mathcal{T}_IU_\varphi^\dagger\Big)\mid I=2,4,\dots,54,~\varphi\in\mathbb{R} \right\}  \\
    & \ \ \ \ \ \ \ \ \  \cup  \left\{\Big(V_\varphi\mathcal{S}_IV_\varphi^\dagger,V_\varphi\mathcal{T}_IV_\varphi^\dagger\Big)\mid I=2,4,\dots,54,~\varphi\in\mathbb{R} \right\}
\end{split}
\end{gather}
where we have defined the unitary matrices\footnote{We have used the fact that $M_1AM_1^\dagger = M_2AM_2^\dagger$ whenever $M_1=e^{i\theta}M_2$ in order to set $\theta=0$ in Eq.\ \eqref{eqn:generalUnitary} without loss of generality.}
\begin{align}
    U_\varphi =\left(\begin{array}{cc}e^{i\varphi} & 0 \\ 0 & 1 \end{array}\right), \ \ \ \ V_\varphi = \left(\begin{array}{cc} 0 & e^{i\varphi} \\ 1 & 0 \end{array}\right).
\end{align}
We can further cut down on the number of modular $\mathcal{S}$ and $\mathcal{T}$ matrices we must consider by demanding that the coefficients $\mathscr{N}_{ij}^k$ computed from $\mathcal{S}$ using Eq.\ \eqref{eqn:verlindeformula} be non-negative integers. Note that if one computes putative fusion rules from the modular $\mathcal{S}$-matrix using the Verlinde formula, one finds that every $(\mathcal{S},\mathcal{T})$ in the set Eq.\ \eqref{eqn:modularDataSpace} leads to the following consistent fusion channels,
\begin{align}
   \mathds{1}\times\mathds{1} = \mathds{1},\ \ \ \ \  \mathds{1}\times\Phi = \Phi\times\mathds{1} = \Phi.
\end{align}
Non-trivial constraints can be obtained by considering the fusion of $\Phi$ with itself,
\begin{align}
\begin{split}
   U_\varphi \mathcal{S}_I U^\dagger_\varphi:& \ \ \    \Phi\times\Phi = e^{-2i\varphi}\mathds{1} + e^{-i\varphi}\frac{r_I-1}{\sqrt{r_I}}\Phi \\
   V_\varphi \mathcal{S}_I V^\dagger_\varphi:&  \ \ \   \Phi\times\Phi = e^{-2i\varphi}\mathds{1} - e^{-i\varphi}\frac{r_I-1}{\sqrt{r_I}}\Phi
\end{split}
\end{align}
where $r_I$ is the non-negative real number associated to $\varrho_I$ through Eq.\ \eqref{eqn:kappar}. Demanding just that the fusion coefficient $\mathscr{N}_{\Phi\Phi}^{\mathds{1}}$ be a non-negative integer then requires that $\varphi = 0,\pi$ in both cases. Therefore, imposing this on Eq.\ \eqref{eqn:modularDataSpace} reduces it to a finite set, and we can check which of these finitely many modular representations is admissible with a case-by-case computation.  

We now introduce some notation which will help us describe the result of this computation. Define 
\begin{gather}\label{eqn:UVW}
    U\equiv U_\pi, \ \ \ \  V\equiv V_0, \ \ \ \  W \equiv V_\pi.
\end{gather}
Given any two-dimensional representation $\varrho$, we can define equivalent representations
\begin{align}
    \varrho_{\mathrm{U}}(\gamma)\equiv U\varrho(\gamma)U^\dagger, \ \ \ \varrho_{\mathrm{V}}(\gamma) \equiv V\varrho(\gamma)V^\dagger, \ \ \  \varrho_{\mathrm{W}}(\gamma) \equiv W \varrho(\gamma) W^\dagger.
\end{align}
It will also be useful in what follows to abbreviate the representation $\varrho_I$ to simply $I$, and the representations $(\varrho_I)_{\mathrm{X}}$ to $I$X for $\mathrm{X} = \mathrm{U}, \mathrm{V}, \mathrm{W}$. For example, we will write $(\varrho_{20})_{\mathrm{V}}$ as 20V. Finally, we define 
\begin{align}
    \varrho_{\A}\equiv \varrho_{20}, \ \ \ \varrho_{\G}\equiv \varrho_{32}
\end{align}
so-labeled because $\varrho_{20}$ is the representation with respect to which the characters of $\A[1,1]$ transform, and $\varrho_{32}$ is the representation with respect to which the characters of $\G[2,1]$ transform. For the convenience of the reader, we explicitly write out the modular data in these two cases, 
\begin{align}\label{eqn:modularrepAG}
    \mathcal{S}_{\A} &= \frac{1}{\sqrt{2}}\left(\begin{array}{cr} 1& 1 \\ 1 & -1  \end{array}\right), & \mathcal{T}_{\A} &= \exp 2\pi i ~\mathrm{diag}(\sfrac{23}{24},\sfrac{5}{24}) \\
    \mathcal{S}_{\G} &= \frac{1}{\sqrt{2+\varphi}}\left(\begin{array}{rr} 1 & \phi \\ \phi & -1\end{array}\right), & \mathcal{T}_{\G} &= \exp 2\pi i ~\mathrm{diag}(\sfrac{53}{60},\sfrac{17}{60})
\end{align}
where $\phi = \frac{1+\sqrt{5}}{2}$.

With these conventions in place, we are ready to state the main result of this subsection. By demanding that $\mathscr{N}_{\Phi\Phi}^\Phi$ be a non-negative integer, and that $\mathcal{S}_{0i}>0$, we find the following. 

\begin{theorem}
The full set of two-dimensional irreducible, admissible modular representations is 
\begin{align}\label{eqn:admissiblereps}
\begin{split}
\mathrm{Semion}:&~20,~28\mathrm{W},~24\mathrm{W} \\
\overline{\mathrm{Semion}}:&~30\mathrm{W},~26,~22 \\
\mathrm{Fibonacci}:& ~32,~40\mathrm{W},~36\mathrm{W} \\
\overline{\mathrm{Fibonacci}}:&~42\mathrm{W},~38,~34
\end{split}
\end{align}
which we have organized by their associated modular tensor categories (see \cite{rowell2009classification}).
Alternatively, this list can be re-written as 
\begin{align}\label{eqn:admissibleAG}
\begin{split}
\mathrm{Semion}:&~ \varrho_{\A},~ \omega\varrho_{\A},~\omega^2\varrho_{\A} \\
\overline{\mathrm{Semion}}:&~\varrho_{\A}^\ast,~\omega\varrho_{\A}^\ast,~\omega^2\varrho_{\A}^\ast \\
\mathrm{Fibonacci}:&~\varrho_{\G},~\omega\varrho_{\G},~\omega^2\varrho_{\G} \\
\overline{\mathrm{Fibonacci}}:&~\varrho_{\G}^\ast,~\omega\varrho_{\G}^\ast,~\omega^2\varrho_{\G}^\ast
\end{split}
\end{align}
where $\omega:\SL_2(\IZ)\to\mathbb{C}^\ast$ is the representation generated by the assignments $\omega(S)=1$ and $\omega(T)=e^{-2\pi i/3}$.
\end{theorem}

\noindent\emph{Remark:}
The admissible $\A$-type representations in Eq.\ \eqref{eqn:admissibleAG} give rise to the fusion rule $\Phi\times \Phi = \mathds{1}$ (corresponding to the so-called ``Semion MTC'' in \cite{rowell2009classification} and its complex conjugate), while the admissible $\G$-type representations give rise to the fusion rule $\Phi\times \Phi = \mathds{1}+\Phi$ (corresponding to the ``Fibonacci MTC'' in \cite{rowell2009classification} and its complex conjugate).\\

\noindent We now explain why the two-dimensional irreducible admissible modular representations exhaust \emph{all} of the two-dimensional admissible modular representations. By assumption, an admissible representation is unitary, and hence if it is not irreducible, it decomposes into a direct sum of two one-dimensional representations of $\SL_2(\IZ)$. In \S\ref{subsec:meromorphic}, we explained that the one-dimensional representations of $\SL_2(\IZ)$ form the group $\mathbb{Z}_{12}$ with generator $\zeta$ given in Eq.\ \eqref{eqn:generatoronedimensional}. Thus, an admissible reducible modular representation $\varrho$ is equivalent to $ \zeta^a\oplus \zeta^b$ for some integers $a,b=0,\dots, 11$. It cannot be the case that $a=b$, because then $\varrho(S)$ is diagonal in every basis, which violates the requirement that $\mathcal{S}_{0i}>0$. On the other hand, if $a\neq b$, then the eigenvalues of $\varrho(T)$ are distinct, and Lemma \ref{lemma:unitarybases} tells us that any unitary change of basis which keeps $\varrho(T)$ diagonal will also keep $\varrho(S)$ diagonal. Hence, we again run into a violation of $\mathcal{S}_{0i}>0$. We therefore have the following.
\begin{proposition}
There are no two-dimensional modular representations which are both admissible and reducible.
\end{proposition}
We conclude this section by observing that the list of representations in \eref{eqn:admissibleAG} can be seen to be in one-to-one correspondence with sequences of quasi-characters from \cite{chandra2019towards} that remain after discarding the following sets: (i) those corresponding to more than 2 primaries, denoted the $A_2$ and $D_4$ series, (ii) those labeled Type II, which have infinitely many negative coefficients in their $q$-expansions, and (iii) those giving rise to negative fusion rules. 

The correspondence is as follows. We start with the series called ``Lee-Yang'' in Section 4.3 of \cite{chandra2019towards}. These fall into three classes corresponding to Wronskian index $\ell=0,2,4$ mod 6. 

In the first class one has central charges $c=\frac{2(6n+1)}{5}$ for $n=0,1,2,3,5,6,7,8$ mod 10. The last four of these correspond to Type II quasi-characters and are discarded. The first four have central charges $c=\sfrac25, \sfrac{14}{5},\sfrac{26}{5},\sfrac{38}{5}$ mod 24. Generalizing the discussion in \cite{mathur1989reconstruction}, the first and last of these can be shown to have negative fusion rules. That leaves central charges $c=\sfrac{14}{5},\sfrac{26}{5}$ mod 24, which correspond to the representations $\varrho_{\G},~\omega\varrho_{\G}^*$ respectively. 

Moving on to $\ell=2$, we have $c=\frac{2(6n-1)}{5}$ and here $n=7,8,9,10$ mod 10 give rise to Type I characters, of which $n=7,10$ are discarded as they have negative fusion rules. The remaining ones, with $c=\sfrac{94}{5},\sfrac{106}{5}$, correspond to the representations $\omega^2\varrho_{\G},~\varrho_{\G}^*$ respectively. 

Finally, the $\ell=4$ quasi-characters are obtained by tensoring the $\E[8,1]$ character $j(\tau)^{\frac13}$ with the $\ell=0$ quasi-characters. This process simply multiplies the modular representation by $\omega$ and therefore gives us $\omega\varrho_{\G},~\omega^2\varrho_{\G}^*$, completing the list of $\G$-type (Fibonacci) representations in \eref{eqn:admissibleAG}. 

A similar exercise maps type I quasi-characters in the $A_1$ series, Section 4.4 of \cite{chandra2019towards}, to the six $\A$-type (Semion) representations in \eref{eqn:admissibleAG}. The $\ell=0$ quasi-characters correspond to the representations $\varrho_{\A},~\omega\varrho_{\A}^*$, the $\ell=2$ quasi-characters correspond to the representations $\omega^2\varrho_{\A},~\varrho_{\A}^*$, and the $\ell=4$ quasi-characters correspond to the representations $\omega\varrho_{\A},~\omega^2\varrho_{\A}^*$. It is gratifying that the same modular representations arise from two somewhat different mathematical starting points: the direct classification of admissible representations vs.\ the study of MLDEs as in \cite{kaneko1998supersingular,kanekokoike2003,kaneko2006onmodular}.

\subsection{Admissible characters}\label{subsec:admissiblecharacters}

In the previous subsection, we classified two-dimensional admissible modular representations. To proceed further, we must ask which of these representations admit admissible characters with central charge in the range $0\leq c \leq 24$. We will see that \emph{every} admissible representation we found in the previous subsection supports at least one admissible character.

\subsubsection*{Rigid representations}
To begin, we consider vector-valued modular forms for one of the following admissible representations,
\begin{align}\label{eqn:rigidreps}
    \varrho = 20,~22,~26,~32,~34,~38,~ 28\mathrm{W},~40\mathrm{W}
\end{align}
or in terms of $\varrho_{\A}$ and $\varrho_{\G}$, the representations 
\begin{align}
   \varrho= \varrho_{\A},~\omega^2\varrho_{\A}^\ast,~ \omega\varrho_{\A}^\ast,~\varrho_{\G},~\omega^2\varrho_{\G}^\ast,~\omega\varrho_{\G}^\ast,~\omega\varrho_{\A},~\omega\varrho_{\G}.
\end{align}
We call these representations ``rigid'' because, as we will show, they admit a unique admissible character with central charge $0\leq c\leq 24$. 

Let us start by assuming that $\varrho\neq \omega\varrho_{\A},\omega\varrho_{\G}$. The $\mathcal{T}$-matrix of such a representation is $\mathcal{T}=\mathrm{diag}(e^{2\pi i m_0},e^{2\pi i m_1})$, where the pair $(m_0,m_1)$ can be read off from the second column of Table \ref{tab:evenweight} and satisfies $0\leq m_1<m_0<1$. If there exists a unitary RCFT with $0<c\leq 24$ whose characters transform covariantly with respect to such a representation, then its central charge must be $c=24(1-m_0)$, and its vacuum character must admit a $q$-expansion of the form 
\begin{align}\label{eqn:leadingvacuum}
    \chi_0(\tau) = q^{m_0-1}+O(q^{m_0}).
\end{align}
Furthermore, we claim that the character of the non-identity primary would need to start as 
\begin{align}\label{eqn:leadingprimary}
    \chi_1(\tau) = aq^{m_1}+O(q^{m_1+1})
\end{align}
with $a\neq 0$. Indeed, since $m_1<m_0$, if the non-identity character started as $aq^{m_1-n}+\cdots$ with $a\neq0$ for a positive integer $n$, then the conformal dimension of the non-identity primary would be negative,
\begin{align}
    h = (m_1-n)-(m_0-1)<-n+1\leq 0,
\end{align}
in violation of unitarity. On the other hand, should $n$ be a negative integer, then the Wronskian index would be negative,
\begin{align}
    \ell =-6(m_0-1+m_1-n)+1\leq -6(m_0+m_1)+1<0
\end{align}
and the characters singular, where the right-most inequality can easily be confirmed by computing $-6(m_0+m_1)+1$ for each of the representations under consideration, using the data of Table \ref{tab:evenweight}.

For each $\varrho\neq \omega\varrho_{\A},\omega\varrho_{\G}$ in Eq.\ \eqref{eqn:rigidreps}, an admissible character of the form Eq.\ \eqref{eqn:leadingvacuum} and Eq.\ \eqref{eqn:leadingprimary} exists, and it is unique. To prove uniqueness, assume there exists another character $\chi'$ satisfying Eq.\ \eqref{eqn:leadingvacuum} and Eq.\ \eqref{eqn:leadingprimary}. Then the difference of these two characters, $\chi-\chi'$, does not have any polar terms in its $q$-expansion, and hence is a \emph{holomorphic} vector-valued modular form of weight zero for $\SL_2(\IZ)$. Because the modular representations we are working with have $\Gamma(N)$ as part of their kernels, it follows that the components of $\chi-\chi'$ are ordinary holomorphic modular forms for $\Gamma(N)$. It is known that the only such forms are the constant functions, but constant functions do not transform correctly under the rest of $\SL_2(\IZ)$ for these representations unless they are zero, in which case we have that $\chi=\chi'$.

To prove existence, we can explicitly construct the character in each case.
\begin{proposition}\label{prop:rigid1}
Let $\varrho$ be one of the admissible modular representations in Eq.\ \eqref{eqn:rigidreps}, except for $\omega\varrho_{\A}$ and $\omega\varrho_{\G}$. The unique quasi-character for $\varrho$ with central charge $0<c\leq 24$ is
\begin{align}
    \chi(\tau) = X^{(\ell)}_{c,+}(\tau)
\end{align}
where $\ell = -6(m_0+m_1)+7$ and $c=24(1-m_0)$. Furthermore, it is admissible.
\end{proposition}
\noindent The remaining two rigid cases can be treated similarly, and one finds the following.
\begin{proposition}\label{prop:rigid2}
Let $\varrho=\omega\varrho_{\A}$ or $\omega\varrho_{\G}$. There is a unique quasi-character for $\varrho$ with central charge $0<c\leq 24$, and it is 
\begin{align}
    \chi(\tau) = X^{(\ell)}_{c,+}(\tau)
\end{align}
where $\ell = -6(m_0+m_1)+13$ and $c=24(1-m_1)$. Furthermore, it is admissible.
\end{proposition}
That $X_{c,+}^{(\ell)}(\tau)$ has all the right properties (except perhaps for positivity) for Proposition \ref{prop:rigid1} and Proposition \ref{prop:rigid2} to be true follows from the results of \S\ref{subsec:MLDE}. The $q$-expansions of these characters are reported in Table \ref{tab:healthycharacters};  their positivity to finite order in the $q$-expansion can be seen by inspection, or more rigorously by appealing to the RCFTs we will assign to them in the next subsection.

\subsubsection*{One-parameter representations}

Let us now consider the remaining admissible representations, 
\begin{align}\label{eqn:oneparameterreps}
    \varrho = \omega^2\varrho_{\A},~\varrho_{\A}^\ast,~\omega^2\varrho_{\G},~\varrho_{\G}^\ast.
\end{align}
We will call these \emph{one-parameter representations} because, as we claim below, they support a one-parameter family of admissible characters. We will interpret the parameter (when it is non-zero) as being the degeneracy of the non-identity primary.

The $\mathcal{T}$-matrix of such a representation is $\mathcal{T}=\mathrm{diag}(e^{2\pi i m_1},e^{2\pi i m_0})$, and hence a unitary RCFT with such a representation as its modular data and $0<c\leq 24$ must have $c=24(1-m_1)$. Using arguments which are nearly identical to the ones employed for rigid representations, one can convince oneself of the following.
\begin{proposition} 
The most general quasi-character for any $\varrho$ in Eq.\ \eqref{eqn:oneparameterreps} with $0<c\leq 24$ takes the form 
\begin{align}\label{eqn:leading2}
\begin{split}
    \chi_0(\tau)&=q^{m_1-1}+O(q^{m_1}) \\
    \chi_1(\tau) &= aq^{m_0-1}+O(q^{m_0})
\end{split}
\end{align}
where $a$ is an integer which completely fixes $\chi$ to be
\begin{align}
    \chi(\tau) &= X_{c,+}^{(\ell)}(\tau) + aV X^{(\ell)}_{c',+}(\tau) .
\end{align}
Here, $\ell=-6(m_0+m_1)+7$, $c=24(1-m_1)$, $c'=24(1-m_0)$, and $V$ is the matrix defined in Eq.\ \eqref{eqn:UVW}.
\end{proposition}
The $q$-expansions are again summarized in Table \ref{tab:healthycharacters}. The theories we match to them in the next subsections rigorously show that they can be made completely positive for suitable choices of the integer $a$.

\begingroup
\setlength{\tabcolsep}{3pt}
\begin{table}[]
\begin{scriptsize}
    \centering
    \begin{tabular}{c|c|c|c|l}
    \toprule
        \multirow{2}{*}{$\varrho$} & \multirow{2}{*}{$(c,\hat h)$} & \multirow{2}{*}{Theory} & \multirow{2}{*}{$\chi(\tau)$} & $q^{\frac{c}{24}}\chi_0(\tau)$\\
        & & & & $q^{\frac{c}{24}-h}\chi_1(\tau)$\\\midrule
        \multirow{2}{*}{$\varrho_{\A}$} & \multirow{2}{*}{$(1,\sfrac14)$} & \multirow{2}{*}{$\A[1,1]$} & \multirow{2}{*}{$[0,1,2]$} & \begin{tiny}$1+3 q+4 q^2+7 q^3+13 q^4+19 q^5+29 q^6+43 q^7+62 q^8+O(q^{9})$\end{tiny} \\
        & & & & \begin{tiny}$2+2 q+6 q^2+8 q^3+14 q^4+20 q^5+34 q^6+46 q^7+70 q^8+O(q^9)$\end{tiny} \\\hline
        \multirow{2}{*}{$\omega\varrho_{\A}$} & \multirow{2}{*}{$(9,\sfrac14)$} & \multirow{2}{*}{$\mathbf{M}^{(16)}\big/\E[7,1]$} & \multirow{2}{*}{$[4,9,2]$}&\begin{tiny}$1+251 q+4872 q^2+48123 q^3+335627 q^4+1868001 q^5+O(q^6) $\end{tiny} \\
        & & & & \begin{tiny}$2+498 q+8750 q^2+79248 q^3+522498 q^4+2786256 q^5+O(q^6) $\end{tiny}\\\hline
        \multirow{2}{*}{$\omega^2\varrho_{\A}$} & \multirow{2}{*}{$(17,\sfrac14)$} & \multirow{2}{*}{$\mathbf{M}^{(24)}\big/\E[7,1]$} & $\hspace{-.2in}[2,17,1632]$&\begin{tiny}$1+(88 a+323) q+(5192 a+60860) q^2+O(q^3) $\end{tiny} \\
        & & & $~+aV[2,11,88] $& \begin{tiny}$a+(1632-319 a) q+(162656-11077 a) q^2+O(q^3) $\end{tiny}\\\hline
        \multirow{2}{*}{$\varrho_{\A}^\ast$} & \multirow{2}{*}{$(23,\sfrac34)$} & \multirow{2}{*}{$\mathbf{M}^{(24)}\big/\A[1,1]$} & $[2,23,32384]$&\begin{tiny}$1+(10 a+69) q+(98 a+131905) q^2+(690 a+12195106) q^3+O(q^4) $\end{tiny} \\
        & & & $~+aV[2,5,10]$& \begin{tiny}$a+(32384-65 a) q+(4418944-450 a) q^2+O(q^3) $\end{tiny}\\\hline
        \multirow{2}{*}{$\omega\varrho_{\A}^\ast$} & \multirow{2}{*}{$(7,\sfrac34)$} & $\mathbf{M}^{(8)}\big/\A[1,1]$ &  \multirow{2}{*}{$[0,7,56]$}&\begin{tiny}$1+133 q+1673 q^2+11914 q^3+63252 q^4+278313 q^5+O(q^6) $\end{tiny} \\
        & & $\cong \E[7,1]$& & \begin{tiny}$56+968 q+7504 q^2+42616 q^3+194768 q^4+772576 q^5+O(q^6) $\end{tiny}\\\hline
        \multirow{2}{*}{$\omega^2\varrho_{\A}^\ast$} & \multirow{2}{*}{$(15,\sfrac34)$} & \multirow{2}{*}{$\mathbf{M}^{(16)}\big/\A[1,1]$} & \multirow{2}{*}{$[4,15,56]$} & \begin{tiny}$1+381 q+38781 q^2+1010062 q^3+14752518 q^5+O(q^6)$\end{tiny}\\
        & & & & \begin{tiny}$56+14856 q+478512 q^2+7841752 q^3+87285024 q^4+O(q^5)$\end{tiny}\\\hline
        \multirow{2}{*}{$\varrho_{\G}$} & \multirow{2}{*}{$(\sfrac{14}5,\sfrac25)$} & \multirow{2}{*}{$\G[2,1]$} & \multirow{2}{*}{$[0,\sfrac{14}5,7]$} &\begin{tiny}$ 1+14 q+42 q^2+140 q^3+350 q^4+840 q^5+1827 q^6+3858 q^7+O(q^8)$\end{tiny} \\
        & & & & \begin{tiny}$7+34 q+119 q^2+322 q^3+819 q^4+1862 q^5+4025 q^6+8218 q^7+O(q^8) $\end{tiny}\\\hline
        \multirow{2}{*}{$\omega\varrho_{\G}$} & \multirow{2}{*}{$(\sfrac{54}5,\sfrac25)$} & \multirow{2}{*}{$\mathbf{M}^{(16)}\big/\F[4,1]$} & \multirow{2}{*}{$[4,\sfrac{54}5,7]$} &\begin{tiny}$1+262 q+7638 q^2+103044 q^3+907932 q^4+6165852 q^5+O(q^6) $\end{tiny} \\
        & & & & \begin{tiny}$7+1770 q+37419 q^2+413314 q^3+3244881 q^4+20317202 q^5+O(q^6) $\end{tiny}\\\hline
        \multirow{2}{*}{$\omega^2\varrho_{\G}$} & \multirow{2}{*}{$(\sfrac{94}5,\sfrac25)$} & \multirow{2}{*}{$\mathbf{M}^{(24)}\big/\F[4,1]$} & $\hspace{-.2in}[2,\sfrac{94}5,4794]$ &\begin{tiny}$1+(46 a+188) q+(2093 a+62087) q^2+(27002 a+2923494) q^3+O(q^4)$\end{tiny} \\
        & & & $~+aV[2,\sfrac{46}5,46]$ & \begin{tiny}$a+(4794-184 a) q+(532134-3841 a) q^2+O(q^3)$\end{tiny}\\\hline
        \multirow{2}{*}{$\varrho_{\G}^\ast$} & \multirow{2}{*}{$(\sfrac{106}5,\sfrac35)$} & \multirow{2}{*}{$\mathbf{M}^{(24)}\big/\G[2,1]$} & $[2,\sfrac{106}5,15847]$ &\begin{tiny}$1+(17 a+106) q+(442 a+84429) q^2+(4063 a+5825442) q^3+O(q^4)  $\end{tiny} \\
        & & &$+aV[2,\sfrac{34}5,17]$& \begin{tiny}$a+(15847-102 a) q+(1991846-1088 a) q^2+O(q^3) $\end{tiny}\\\hline
        \multirow{2}{*}{$\omega\varrho_{\G}^\ast$} & \multirow{2}{*}{$(\sfrac{26}5,\sfrac35)$} & $\mathbf{M}^{(8)}\big/\G[2,1]$ & \multirow{2}{*}{$[0,\sfrac{26}5,26]$} &\begin{tiny}$1+52 q+377 q^2+1976 q^3+7852 q^4+27404 q^5+84981 q^6+O(q^7)$\end{tiny} \\
        & & $\cong\F[4,1]$ & & \begin{tiny}$26+299 q+1702 q^2+7475 q^3+27300 q^4+88452 q^5+260650 q^6+O(q^7) $\end{tiny}\\\hline
        \multirow{2}{*}{$\omega^2\varrho_{\G}^\ast$} & \multirow{2}{*}{$(\sfrac{66}5,\sfrac35)$} & \multirow{2}{*}{$\mathbf{M}^{(16)}\big/\G[2,1]$} & \multirow{2}{*}{$[4,\sfrac{66}5,26]$} &\begin{tiny}$1+300 q+17397 q^2+344672 q^3+4072878 q^4+35365284 q^5+O(q^6) $\end{tiny} \\
        & & & & \begin{tiny}$26+6747 q+183078 q^2+2566199 q^3+24832272 q^4+O(q^5) $\end{tiny}\\
        \bottomrule
    \end{tabular}
    \caption{Admissible characters for the admissible two-dimensional modular representations. Here, $[\ell,c,N]$ stands for the function $\X_{c,N}^{(\ell)}(\tau)$, and $\mathbf{M}^{(c)}$ denotes a chiral algebra with one primary and central charge $c$.}
    \label{tab:healthycharacters}
\end{scriptsize}
\end{table}
\endgroup

\subsection{Enumeration of theories}\label{subsec:enumeration}

In the previous subsection, we showed that there are 12 admissible representations which support admissible characters. Their data is summarized in Table \ref{tab:healthycharacters}. We now turn to classifying the unitary rational conformal field theories which have these forms as their characters.

\subsubsection[Theories with $0\leq c <8$]{{\boldmath Theories with $0\leq c <8$}}

We start with the vector-valued forms whose central charge lies in the range $0\leq c<8$. These correspond to the rigid representations
\begin{align}
    \varrho=\varrho_{\A},~\varrho_{\G},~\omega\varrho_{\G}^\ast,~\omega\varrho_{\A}^\ast
\end{align} 
at central charge 
\begin{align}
    c=1,~\sfrac{14}5,~\sfrac{26}5,~7
\end{align} 
respectively. It is known that the MMS theories
\begin{align}
    \A[1,1],~ \G[2,1], ~\F[4,1],~\E[7,1]
\end{align} 
have characters which match these cases. Moreover, it is possible to show that these are the \emph{unique} RCFTs with these characters \cite{mason2018vertex}. 

The proof of this is straightforward. The space of dimension one operators in a unitary RCFT must form a reductive Lie algebra with rank less than or equal to the central charge \cite{goddard1989meromorphic,dong2004rational}. For example, in the case of the representation $\varrho_{\A}$, inspection of the characters reveals that there are 3 currents, so the global symmetry algebra must be a reductive Lie algebra of dimension 3 and rank at most 1. Of course $\A[1]$ is the unique such Lie algebra, and the only possible level we can place it at without exceeding $c=1$ is level 1. In fact, the central charge of $\A[1,1]$ is equal to 1 and the algebra already has two primaries, so we know this must be the correct theory on the nose. The same argument goes through mutatis mutandis in the other three cases.

An alternative proof for the uniqueness of $\F[4,1]$ and $\E[7,1]$ which is more in the spirit of this paper is the following. Let us consider the genus $(c,\varrho)=(\sfrac{26}5,\omega\varrho_{\G}^\ast)$ (see \S\ref{subsec:coset} for the definition of genus). Unitarity of the representation $\varrho_{\G}$ implies that
\begin{align}
   \varrho_{\G}^T \cdot (\omega\varrho_{\G}^\ast) =\omega 
\end{align}
What this means is that \emph{any} theory supported at the genus $(\sfrac{26}{5},\omega\varrho_{\G}^\ast)$ can be glued to a theory supported at $(\sfrac{14}{5},\varrho_{\G})$ (the unique choice being $\G[2,1]$) to produce a chiral algebra with one primary operator in the genus $(8,\omega)$ (the unique choice being $\E[8,1]$). Conversely, if any theory in the genus $(\sfrac{26}{5},\omega\varrho_{\G}^\ast)$ is to exist, then it must be expressible as a coset of the form $\E[8,1]\big/\G[2,1]$. Any two $\G[2,1]$ subalgebras of $\E[8,1]$ are related by an automorphism of $\E[8,1]$, and so there is a unique theory of the form $\E[8,1]\big/\G[2,1]$ up to isomorphim: it is none other than $\F[4,1]$, so $\F[4,1]$ is the unique theory with $c=\sfrac{26}5$ and modular representation $\omega\varrho_{\G}^\ast$. A similar argument proves that $\E[7,1]$ is the unique theory at genus $(7,\omega\varrho_{\A}^\ast)$.

A reader familiar with \cite{mathur1988classification} may be surprised that we are claiming that there are only 4 theories in this range of central charge, so we pause for a moment to compare to older results. By studying the most general $\ell=0$ MLDE, namely Eq.\ \eqref{mmseq}, op.\ cit.\ found that the values of $c$ which render the solution $X^{(0)}_{c,+}(\tau)$ positive are 
\begin{align}
    c = \sfrac25,~1,~2,~\sfrac{14}5,~4,~\sfrac{26}5,~6,~7,~\sfrac{38}5,~8
\end{align}
corresponding to the MMS theories 
\begin{align}
    \mathsf{LY},~\A[1,1],~\A[2,1],~\G[2,1],~\D[4,1],~\F[4,1],~\E[6,1],~\E[7,1],~\E[7\frac12,1],~\E[8,1].
\end{align}
This identification is rigorous in the sense that the characters of the above theories can be independently computed and shown to be equal to those produced by the MLDE.\footnote{Several years after the above discovery and working from a very different point of view, Deligne \cite{deligne1996laserie} discovered the same series as a set of finite Lie algebras sharing some remarkable properties. This series of Lie algebras was also found independently by Cvitanovič and is described in \cite{cvitanovic2008group}.} Why has our classification not recovered these additional cases?

First, as is well-known, $\E[8,1]$ is a \CFT{1}, and its unique character has re-appeared as a solution to a rank-2 MLDE that happens to be modular invariant by itself. Thus it can be discarded since our goal in the present work is to classify only theories with precisely two primary fields. Next, $\A[2,1]$ and $\E[6,1]$ are actually theories with three primaries, two of which have the same character because they are complex conjugates of each other. Similarly $\D[4,1]$ describes a \CFT{4}, three of whose primaries are related by triality and therefore have the same character. We must discard these as well.

This leaves the quasi-characters having central charges $\sfrac25$ and $\sfrac{38}{5}$. In \cite{mathur1989reconstruction} it was shown that both have at least one negative fusion-rule coefficient (as calculated from the modular $\mathcal{S}$-matrix using the Verlinde formula \cite{verlinde1988fusion,huang2008vertex}). These were therefore discarded.\footnote{Today both of these are known as Intermediate Vertex Operator Algebras (IVOA), a generalisation of VOAs where negative fusion rules are permitted \cite{kawasetsu2014intermediate}. We will not include IVOAs in our classification in the present work.} The above works also noted that if one switches which character they consider to correspond to the vacuum and which they consider to correspond to the non-identity primary, then the $c=\sfrac25,\sfrac{38}5$ characters can be reinterpreted as characters with $c=-\sfrac{22}5,-\sfrac{58}5$. In the  former case, this leads to a {\em non-unitary} theory with consistent fusion rules. The  conformal dimension in this case is $-\sfrac15$ which permits us to identify this character with the non-unitary Lee-Yang minimal model. The $c=\sfrac{38}{5}$ case is different: it also acquires consistent fusion rules after exchanging characters, but  this exchange leads to a 57-fold degenerate vacuum state. Thus its characters  cannot be written in the form \eref{eqn:qexpansioncharacter} with unit coefficient for the first term in the first line and the exchanged theory also has to be rejected. 

Since we are restricting attention to unitary theories, both the $\sfrac25, \sfrac{38}{5}$ cases as well as their exchanged versions are unacceptable. Hence the short summary of \cite{mathur1988classification} for us is that there are precisely four unitary \CFT{2}s with $\ell=0$, namely the WZW models $\A[1,1],\G[2,1],\F[4,1],\E[7,1]$.

\subsubsection[Theories with $8\leq c <16$]{{\boldmath Theories with $8\leq c <16$}}

Next, we turn to theories with central charge in the range $8\leq c \leq 16$. Such theories must be associated to the rigid modular representations 
\begin{align}\label{eqn:repstobeclassified16}
   \varrho= \omega\varrho_{\A}, ~ \omega\varrho_{\G},~\omega^2\varrho_{\G}^\ast,~\omega^2\varrho_{\A}^\ast,
\end{align}
which occur at central charge 
\begin{align}
    c=9,~\sfrac{54}5,~\sfrac{66}5,~15
\end{align}
and support admissible characters of Wronskian index $\ell=4$. These can be paired with the representations 
\begin{align}\label{eqn:knownreps}
    \tilde{\varrho}=\omega\varrho_{\A}^\ast,~\omega\varrho_{\G}^\ast, \varrho_{\G},~\varrho_{\A}
\end{align}
respectively, which occur at central charges 
\begin{align}
    \tilde{c} = 7,~\sfrac{26}5,~\sfrac{14}5,~1.
\end{align}
Their pairing leads to the relations
\begin{align}
    \varrho^T\cdot \tilde{\varrho} = \omega^2, \ \ \ \ c+\tilde{c}=16.
\end{align}
At the level of theories, what this means is that \emph{any} theory in the genus $(c,\varrho)$ can be glued to an MMS theory in the dual genus $(\tilde{c},\tilde\varrho)$ to produce a chiral algebra with one primary operator in the genus $(16,\omega^2)$, the two choices being $\E[8,1]\otimes \E[8,1]$ and $\D[16,1]^+$. Conversely, any \CFT{2} with $8\leq c < 16$ can be obtained as a coset of $\E[8,1]\otimes\E[8,1]$ or $\D[16,1]^+$ by an MMS theory. Such cosets can be explicitly enumerated (using the results of Appendix \ref{app:liecurrentalgebras}), and one finds 
\begin{align}
\begin{split}
    {(9,\omega\varrho_{\A})}&= \{\E[8,1]\otimes \A[1,1]\cong (\E[8,1]\otimes\E[8,1])\big/\E[7,1]\} \\
    {(\sfrac{54}5,\omega\varrho_{\G})}&= \{\E[8,1]\otimes \G[2,1]\cong (\E[8,1]\otimes \E[8,1])\big/\F[4,1]\} \\
    (\sfrac{66}5,\omega^2\varrho_{\G}^\ast)&= \{\E[8,1]\otimes \F[4,1]\cong(\E[8,1]\otimes\E[8,1])\big/\G[2,1], ~ \D[16,1]^+\big/\G[2,1]\} \\
    (15,\omega^2\varrho_{\A}^\ast)&= \{\E[8,1]\otimes \E[7,1]\cong(\E[8,1]\otimes\E[8,1])\big/\A[1,1], ~ \D[16,1]^+\big/\A[1,1]\}.
\end{split}
\end{align}

We again compare to previous work. Admissible characters with $\ell=4$ were studied and classified in \cite{tener2017classification,chandra2019towards,grady2020classification}. Furthermore, these references identified at least one example of a chiral algebra for each admissible character. In particular, the theories  $\E[8,1]\otimes\mathrm{MMS}$ were identified as furnishing admissible $\ell=4$ characters, but the two cosets  $\D[16,1]^+\big/\G[2,1]$ and $\D[16,1]^+\big/\A[1,1]$ (which to our knowledge, have not been considered previously in the literature) did not appear because they are isospectral with $\E[8,1]\otimes \F[4,1]$ and $\E[8,1]\otimes \E[7,1]$, respectively. The rest of the admissible characters they found in the range $8\leq c <16$ don't arise in our classification for similar reasons as in the previous subsection.

In addition, the $\ell=4$ MLDE can be shown to permit three new admissible characters with $c=\sfrac{162}{5},33,34$. None of these values of $c$ are less than $25$, the case of interest here, so we can ignore them henceforth. However for the interested reader we mention that the $c=33$ case has been identified with a CFT in
\cite{grady2020classification}, and the $c=\sfrac{162}{5}$ case was found to be of IVOA type in \cite{chandra2019towards}. To our knowledge, the status of the $c=34$ case remains unclear; in any case, it has fusion rules of $\A[2,1]$-type and therefore can at best describe a theory with three primaries.

\subsubsection[Theories with $16\leq c <25$]{{\boldmath Theories with $16\leq c <25$}}

Finally, we enumerate the $p=2$ theories with central charge satisfying $16\leq c < 25$. 

First, we note that there are no theories in the range $24\leq c < 25$. Indeed, by examining the matrix element $\mathcal{T}_{00}$ in every two-dimensional admissible modular representation, we see that no chiral algebra can have $c~\mathrm{mod}~24$ in the interval $[0,1)$.

Thus, we restrict our attention from now on to chiral algebras with $16\leq c < 24$. These must belong to the genera 
\begin{align}
    (c,\varrho) = (17,\omega^2\varrho_{\A}),~(\sfrac{94}{5},\omega^2\varrho_{\G}),~(\sfrac{106}{5},\varrho_{\G}^\ast),~(23,\varrho_{\A}^\ast).
\end{align}
By the gluing principle, any theory in these genera must be obtained as a coset of a $c=24$ chiral algebra $\mathcal{A}$ with one irreducible representation (i.e.\ a Schellekens theory \cite{schellekens1993meromorphicc}) by one of the algebras $\E[7,1]$, $\F[4,1]$, $\G[2,1]$, or $\A[1,1]$, respectively (cf.\ Table \ref{tab:healthycharacters}). 

Fixing $\mathcal{A}$ to be a Schellekens theory, and $\mathcal{V}$ to be one of $\E[7,1],\F[4,1],\G[2,1]$, or $\A[1,1]$, there may be multiple cosets of $\mathcal{A}$ by $\mathcal{V}$. This has a chance of happening if there are multiple inequivalent ways of realizing $\mathcal{V}$ inside of $\mathcal{A}$ as a subalgebra, i.e.\ if there are multiple subalgebras isomorphic to $\mathcal{V}$ which are not related by an automorphism of $\mathcal{A}$ (cf.\ the discussion in \S\ref{subsec:coset}). 

As explained in Appendix \ref{app:liecurrentalgebras}, any embedding  $\mathcal{V}\hookrightarrow \mathcal{A}$ must factor through one of the simple Kac--Moody factors of $\mathcal{A}$ at level 1, 
\begin{align}
    \mathcal{V}\hookrightarrow \mathfrak{g}_1 \hookrightarrow \mathcal{A}.
\end{align}
Let us imagine embedding $\mathcal{V}$ into $\mathcal{A}$ in two different ways, and asking whether they lead to isomorphic cosets or not. There are three different scenarios. \\

\noindent\textbf{Scenario 1.} In the first scenario, we imagine two embeddings $\iota,\iota':\mathcal{V}\to\mathcal{A}$ which factor through two non-isomorphic simple factors $\mathfrak{g}_1,\mathfrak{g}_1'$, respectively, as in part (a) of Figure \ref{fig:equivcosets}. One generically expects that the two cosets so-obtained will be \emph{non-isomorphic}: indeed, in all the cases we consider, this can be seen by using the techniques of Appendix \ref{app:liecurrentalgebras} to compute the Lie algebras of the automorphism groups of the two cosets $\mathcal{A}\big/\iota(\mathcal{V})$ and $\mathcal{A}\big/\iota'(\mathcal{V})$, and verifying that they are not isomorphic. \\ 

\noindent\textbf{Scenario 2.} Another possibility is that $\mathcal{V}$ is embedded in two different ways into the \emph{same} simple factor $\mathfrak{g}_1$ of $\mathcal{A}$. By Proposition \ref{prop:subalgebras}, the images of any two such embeddings can be rotated into one another using the Lie group symmetries of $\mathcal{A}$, and hence lead to isomorphic coset theories. This situation is summarized in part (b) of Figure \ref{fig:equivcosets}. \\

\noindent \textbf{Scenario 3.} The most subtle possibility is that $\mathcal{V}$ is embedded into two different simple factors of $\mathcal{A}$ which are isomorphic to one another, as described in part (c) of Figure \ref{fig:equivcosets}. In this case, we must search the symmetry group of $\mathcal{A}$ for outer automorphisms which map one simple factor to the other. If such an automorphism can be found, then the cosets are isomorphic. Otherwise, further consideration is needed. Proposition \ref{prop:outerauts} guarantees that such an automorphism can be found in all but two cases.

\begin{figure}
     \centering
     \begin{subfigure}[b]{0.3\textwidth}
         \centering
\begin{tikzcd}[column sep=tiny,row sep=small]
 & \mathcal{V} \arrow[dl]\arrow[dr] \\
\mathfrak{g}_1 \arrow[dr] & & \mathfrak{g}_1'\arrow[dl] \\
& \mathfrak{g}_1\oplus\mathfrak{g}_1'\arrow[d] \\
& \mathcal{A}
\end{tikzcd}
         \caption{}
         \label{subfig:a}
     \end{subfigure}
     \hfill
     \begin{subfigure}[b]{0.3\textwidth}
         \centering
         \begin{tikzcd}
\mathcal{V}\arrow[d,bend left]\arrow[d,bend right] \\
\mathfrak{g}_1 \arrow[d]\\
\mathcal{A}
\end{tikzcd}
         \caption{}
         \label{subfig:b}
     \end{subfigure}
     \hfill
     \begin{subfigure}[b]{0.3\textwidth}
         \centering
         \begin{tikzcd}[column sep=tiny,row sep=small]
 & \mathcal{V} \arrow[dl]\arrow[dr] \\
\mathfrak{g}_1 \arrow[dr] & & \mathfrak{g}_1\arrow[dl] \\
& \mathfrak{g}_1\oplus\mathfrak{g}_1\arrow[d] \\
& \mathcal{A}
\end{tikzcd}
         \caption{}
         \label{subfig:c}
     \end{subfigure}
        \caption{Different pairs of embeddings of  $\mathcal{V}=\A[1,1],\G[2,1],\F[4,1],$ or $\E[7,1]$ into a $c=24$ chiral algebra  $\mathcal{A}$ with one simple module. }
        \label{fig:equivcosets}
\end{figure}
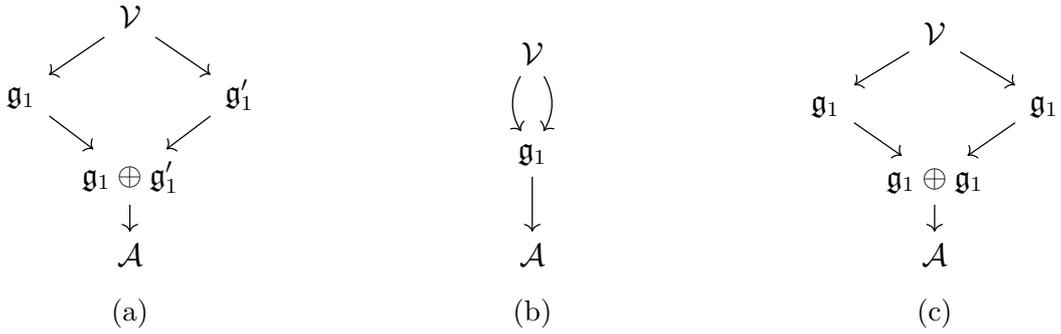

\begin{example}
To make the preceding discussion more concrete, consider the theory $\mathbf{S}(\A[7,1]^2\D[5,1]^2)$, numbered 49 in \cite{schellekens1993meromorphicc}.  The claim is that there are two inequivalent cosets of the form $\mathbf{S}(\A[7,1]^2\D[5,1]^2)\big/\A[1,1]$, one obtained by embedding $\A[1,1]$ into either of the $\A[7,1]$ factors, and the other obtained by embedding $\A[1,1]$ into either of the $\D[5,1]$ factors. We label these two cosets $\mathbf{S}(\A[7,1]^2\D[5,1]^2)\big/(\A[1,1]\hookrightarrow \A[7,1])$ and $\mathbf{S}(\A[7,1]^2\D[5,1]^2)\big/(\A[1,1]\hookrightarrow \D[5,1])$, respectively. 

For the purposes of this discussion, let us distinguish the two copies of $\A[7,1]$ by giving one of them a prime, and likewise for $\D[5,1]$; hence, we write the Schellekens theory as $\mathbf{S}(\A[7,1]\A[7,1]'\D[5,1]\D[5,1]')$. To see why the claim of the preceding paragraph is true, note that by Proposition \ref{prop:subalgebras}, any two $\A[1,1]$ subalgebras which reside in the same $\A[7,1]$ factor, as in Scenario 2, can be rotated into one another by the Lie group symmetries of $\mathbf{S}(\A[7,1]^2\D[5,1]^2)$, and hence lead to isomorphic cosets (likewise for $\A[7,1]$ replaced by $\D[5,1]$).

Furthermore, by Proposition \ref{prop:outerauts}, the following two cosets, which follow Scenario 3, are isomorphic
\begin{align}
   \mathbf{S}(\A[7,1]\A[7,1]'\D[5,1]\D[5,1]')\big/ (\A[1,1]\hookrightarrow \A[7,1])\cong \mathbf{S}(\A[7,1]\A[7,1]'\D[5,1]\D[5,1]')\big/ (\A[1,1]\hookrightarrow \A[7,1]')
\end{align}
because there is an outer automorphism of $\mathbf{S}(\A[7,1]\A[7,1]'\D[5,1]\D[5,1]')$ which maps $\A[7,1]$ to $\A[7,1]'$ (likewise for $\A[7,1],\A[7,1]'$ replaced by $\D[5,1],\D[5,1]'$). 

On the other hand, the two cosets 
\begin{align}
    \mathbf{S}(\A[7,1]\A[7,1]'\D[5,1]\D[5,1]')\big/(\A[1,1]\hookrightarrow \A[7,1]) , ~ \mathbf{S}(\A[7,1]\A[7,1]'\D[5,1]\D[5,1]')\big/ (\A[1,1]\hookrightarrow\D[5,1])
\end{align}
are inequivalent, as in Scenario 1. Indeed, by appealing to Table \ref{tab:centralizers}, we see that the former inherits an $\A[7,1]\A[5,1]\D[5,1]^2\mathsf{U}_1$ Kac--Moody algebra, whereas the latter inherits an $\A[7,1]^2\D[5,1]\A[3,1]\A[1,1]$ Kac--Moody algebra; in particular, their global symmetry groups have different Lie algebras, and hence they are distinct theories.
\end{example}

Using arguments analogous to the ones employed above, in conjunction with the techniques and data provided in Appendix \ref{app:liecurrentalgebras}, one can enumerate all the inequivalent cosets $\mathcal{A}\big/\mathcal{V}$, with $\mathcal{A}$ a Schellekens theory and $\mathcal{V}$ one of the four, two-primary MMS theories. However, there are two exceptions we must contend with (which occur when $\mathcal{A}$ is one of the exceptions to Proposition \ref{prop:outerauts}). We conclude by treating these in turn.\footnote{We thank Ching Hung Lam and Hiroki Shimakura for useful exchanges related to these exceptions.} \\

\noindent\textbf{Exception 1.} The first exception is when $\mathcal{A} = \mathbf{S}(\D[6,5]\A[1,1]\A[1,1]')$. In this case, there is no automorphism which maps $\A[1,1]$ into $\A[1,1]'$ and so there is a chance that the cosets $\mathbf{S}(\D[6,5]\A[1,1]\A[1,1]')\big/\A[1,1]$ and $\mathbf{S}(\D[6,5]\A[1,1]\A[1,1]')\big/\A[1,1]'$ are inequivalent. We can verify that this is indeed true. We note that, using the results of \cite{schellekens1993meromorphicc}, the first coset decomposes into $\D[6,5]\A[1,1]$-modules as 
\begin{align}
\begin{split}
    &\mathbf{S}(\D[6,5]\A[1,1]\A[1,1]')\big/\A[1,1]\cong \\ 
    & \hspace{.5in}\big[(0,0,0,0,0,0)\oplus (0,1,0,0,0,2)\oplus (0,1,0,0,2,0) \\
    &\hspace{.65in}\oplus(1,0,0,1,1,1)\oplus (0,0,2,0,0,0)\oplus(2,0,0,1,0,0)\big]\otimes \A[1,1] \\
&\hspace{.5in}\oplus \big[(0,0,0,0,0,5)\oplus(0,0,0,1,0,1)\oplus (2,0,0,1,0,1) \\
&\hspace{.65in}\oplus (1,1,0,0,1,0)+(0,0,2,0,0,1)+(0,1,0,0,2,1)\big]\otimes J
\end{split}
\end{align}
while the second coset decomposes as 
\begin{align}
    \begin{split}
    &\mathbf{S}(\D[6,5]\A[1,1]\A[1,1]')\big/\A[1,1]'\cong \\ 
        & \hspace{.5in}\big[(0,0,0,0,0,0)\oplus(0,1,0,0,0,2)\oplus(0,1,0,0,2,0) \\
        &\hspace{.65in}\oplus(1,0,0,1,1,1)\oplus(0,0,2,0,0,0)\oplus(2,0,0,1,0,0)\big]\otimes \A[1,1] \\
&\hspace{.5in}\oplus \big[(0,0,0,0,5,0)\oplus(2,0,0,1,1,0)\oplus(0,0,0,1,1,0)\\
&\hspace{.65in}\oplus(1,1,0,0,0,1)\oplus(0,0,2,0,1,0)\oplus(0,1,0,0,1,2)\big]\otimes J
    \end{split}
\end{align}
where here, $J$ is the unique non-vacuum module of $\A[1,1]$, and $(\lambda_1,\cdots,\lambda_6)$ denotes the $\D[6,5]$-module with Dynkin labels given by $\lambda_i$. As can be seen by inspection of the module structure, an isomorphism which relates these two cosets would need to restrict to an automorphism of $Z:=\mathbf{S}(\D[6,5]\A[1,1]\A[1,1]')\big/\A[1,1]\A[1,1]'$ which is a lift of the diagram outer-automorphism of $\D[6,5]$. But since $\mathcal{A}$ is a simple-current extension of $Z\otimes \A[1,1]\A[1,1]'$, if $Z$ had such an automorphism, then it could be lifted to an outer automorphism of $\mathcal{A}$, using Theorem 2.1 of \cite{shimakura2007lifts}. By consulting Table 1 of \cite{betsumiya2022automorphism}, one sees that the outer automorphism group of $\mathcal{A}$ is trivial, and hence no such outer automorphism exists. Thus, there cannot be an isomorphism which relates the two cosets, and they should be thought of as distinct theories. \\

\noindent\textbf{Exception 2.} The second exception is when $\mathcal{A} = \mathbf{S}(\A[7,4]\A[1,1]^3)$. In this case, there is a $\mathbb{Z}_2$ automorphism which swaps two of the $\A[1,1]$ factors, but does not connect them to the third $\A[1,1]$ factor. Thus, there are potentially two inequivalent cosets, which we might write as 
\begin{align}
    \mathbf{S}(\A[7,4]\A[1,1]'\A[1,1]^2)\big/\A[1,1], \ \ \ \ \mathbf{S}(\A[7,4]\A[1,1]'\A[1,1]^2)\big/\A[1,1]'.
\end{align}
One can verify that they are indeed inequivalent by using the results of \cite{schellekens1993meromorphicc} to decompose the two cosets into $\A[7,4]\A[1,1]^2$-modules, and confirming that the modules which appear cannot be related by any outer automorphism of $\A[7,4]\A[1,1]^2$. \\

\noindent The full list of 123 two-primary theories with $c<25$ is reported in Appendix \ref{app:theories}.

\clearpage

\section{Future Directions}\label{sec:future}

In this work, we have classified all unitary, rational \CFT{2}s with central charge $c<25$. Every such theory is either one of four MMS theories --- $\A[1,1]$, $\G[2,1]$, $\F[4,1]$, $\E[7,1]$ --- or a coset of a chiral algebra with one primary by such an MMS theory. Such theories can be explicitly enumerated, and their basic properties (characters, Kac--Moody subalgebras, etc.) computed. The complete list appears in Appendix \ref{app:theories}.  There are a number of questions which are ripe for future research.

The most immediate question one can ask is whether or not it is possible to push the classification of \CFT{2}s to higher central charge. The next case where one expects to find theories is $c=25$. The logic of the gluing principle says that classifying such theories is equivalent to classifying \CFT{1}s with $c=32$ and an $\E[7,1]$ subalgebra. Some examples here are known \cite{chandra2019curiosities}, but the list is far from complete. One motivation for studying this question comes from penumbral moonshine \cite{harvey2015traces,duncan2021overview,duncan2022modular,duncan2022two} (see also Appendix \ref{app:penumbral}). To explain this connection, note that the most general quasi-character with central charge $c=25$ is
\begin{align}\label{eqn:penumbralquasicharacter}
    \chi(\tau)=X_{25,+}^{(0)}(\tau)+\beta  V X_{19,+}^{(0)}(\tau)+\gamma X_{1,+}^{(0)}(\tau).
\end{align}
In \cite{duncan2021overview}, it was found that if one sets $\beta=\gamma=0$, then the function
\begin{align}
\begin{split}
    F^{(-4,1)}(\tau)&= 2\eta(\tau)\chi(\tau) \\
    &= \begin{cases}
    2q^{-1}-492+2\cdot 142884q+2\cdot 18473000q^2+\cdots \\
    2\cdot 565760q^{\frac54}+2\cdot 51179520q^{\frac94}+2\cdot 1912896000q^{\frac{13}4}+\cdots
    \end{cases}
\end{split}
\end{align}
enjoys a tight relationship with the representation theory of the exceptional finite group $G^{(-4,1)}=2.F_4(2).2$.\footnote{In the Gap \cite{GAP4} database, this group appears as ``\texttt{Isoclinic(2.F4(2).2)}''.} For example, $142884$ and $565760$ are precisely the dimensions of certain irreducible representations of $G^{(-4,1)}$. It would be remarkable if this connection could find some physical explanation in terms of VOA, or VOA-adjacent structures. There are of course obstacles to a unitary RCFT interpretation of this function as written: the vacuum is doubly degenerate and there is a negativity in the spectrum. Note however that the negativity can be cured by taking $\gamma \geq 246$ in Eq.\ \eqref{eqn:penumbralquasicharacter}, though it is not obvious that one can do this without sacrificing the appearance of the group $G^{(-4,1)}$. Similarly, it might be found in the final analysis that the overall factor of 2 which leads to a degenerate vacuum is inessential, or alternatively, that the degenerate vacuum has a natural interpretation in terms of intermediate vertex operator algebras \cite{kawasetsu2014intermediate}. Regardless, having complete knowledge of the genus $(c,\varrho)=(25,\varrho_{\A})$ may offer valuable clues for teasing out the physical origin of penumbral moonshine.

Another question is how effectively one can classify \CFT{p}s with $p>2$ and $c\lesssim24$ using the strategy we have employed in this paper.\footnote{Note added: This problem is addressed in \cite{brayhaun}.} One of the simplifying features of the $p=2$ analysis was that the only cosets we were required to take involved level 1 affine Kac--Moody algebras based on simple Lie algebras. Thus, the enumeration of these cosets  more or less reduced to the enumeration of index-1 embeddings of ordinary simple Lie algebras into simple Lie algebras, a problem for which there are an abundance of tools available. In the $p>2$ setting, the cosets required will be more involved. For example, the Ising model is the unique \CFT{3} at $c=\sfrac12$, and therefore the classification of theories at $c=\sfrac{47}2$ with conjugate modular data will involve having a detailed understanding of equivalence classes of Ising subalgebras in meromorphic conformal field theories. In the monster CFT $V^\natural$, it is understood \cite{miyamoto1996griess} that Ising subalgebras are in one-to-one correspondence with involutions in a particular conjugacy class of the monster group $\mathbb{M}$, which in the Atlas \cite{conway1985atlas} conventions is called the 2A conjugacy class. All of these Ising subalgebras are conjugate to one another, and so the coset $V^\natural\big/ \mathsf{L}_{\sfrac12}$ is always isomorphic to the Baby Monster VOA $\textsl{V}\mathbb{B}^\natural$ \cite{hohn1996selbstduale,hohn2012mckay} (see also \cite{Hampapura:2016mmz,lin2021duality,bae2021conformal}).\footnote{For us, the Baby Monster VOA $\textsl{V}\mathbb{B}^\natural$ refers to the bosonic subalgebra of the vertex operator ``super'' algebra defined in \cite{hohn1996selbstduale}, i.e.\ the subalgebra consisting just of the operators with integral conformal dimension.} Beyond the monster CFT, there are many beautiful results \cite{dong1996associative,dong1998framed,miyamoto1996griess,lam2015classification,lam2011constructions,lam2012quadratic,lam2005mckay,miyamoto2004new} which contend with the structure of Ising subalgebras in meromorphic CFTs; it would be interesting if everything that is known at present could be synthesized to produce a complete classification of \CFT{3}s with $c=\sfrac{47}{2}$.

Our classification naturally pairs $p=2$ chiral algebras $\mathcal{V}$, $\widetilde{\mathcal{V}}$ which can be glued together to produce a $p=1$ theory. Recently it has been appreciated that a new class of Hecke operators \cite{harvey2018hecke,harvey2020galois}, which act on vector-valued modular forms rather than scalar modular forms, often relate the characters of $\widetilde{\mathcal{V}}$ to those of $\mathcal{V}$ when the central charge of $\widetilde{\mathcal{V}}$ is an integer multiple of the central charge of $\mathcal{V}$ (see also \cite{bae2021conformal,Duan:2022ltz}). It would be interesting to explore how this Hecke correspondence plays out in the context of the \CFT{2}s classified in this paper. Even more interesting would be if one could define a generalized class of Hecke operators which relate the characters of $\widetilde{\mathcal{V}}$ to those of $\mathcal{V}$ even when the central charge of $\widetilde{\mathcal{V}}$ is only a fractional multiple of the central charge of $\mathcal{V}$. 

Finally, it would be interesting if one could use the techniques we have described in this paper to classify CFTs, but relaxing some of the assumptions of unitarity, regularity, etc. Among the many reasons to study non-unitary theories, the prospect of leveraging them to learn about 4d $\mathcal{N}=2$ SCFTs \cite{beem2015infinite,kaidi2022needles,beem2018vertex} is a comparatively recent and exciting one.

\section*{Acknowledgements}
We are happy to thank Nathan Benjamin, Thomas Creutzig, Arpit Das, Chethan N. Gowdigere, Greg Moore, Ching-Hung Lam, Daniel Ranard, and Hiroki Shimakura for insightful exchanges. We especially thank Hiroki Shimakura and Ching-Hung Lam for explaining their work to us, and Nathan Benjamin for comments which greatly simplified the presentation of our results. BR gratefully acknowledges support from NSF grant PHY 1720397. SM acknowledges the kind hospitality of Shamit Kachru and the Department of Physics, Stanford University, where this work was initiated. He is also grateful for support from a grant by Precision Wires India Ltd.\ for String Theory and Quantum Gravity research at IISER Pune. \\

\noindent The authors have no relevant financial or non-financial interests to disclose.

\clearpage

\appendix

\section{Theories}\label{app:theories}

This appendix contains tables which enumerate all theories with $p=2$ primaries and $c<25$, ordered by increasing central charge. In each table, the first column specifies the number of the theory. The second column gives a description of the theory. The notation $\mathbf{S}(\mathsf{X})$ denotes the unique $c=24$ meromorphic CFT with Kac--Moody subalgebra $\mathsf{X}$.  The notation $\mathbf{S}(\cdots \mathsf{X}_{r,1} \cdots )\big/(\mathsf{Y}_{r',1}\hookrightarrow \mathsf{X}_{r,1})$ stands for the coset of the Schellekens theory $\mathsf{S}(\cdots \mathsf{X}_{r,1}\cdots)$ by $\mathsf{Y}_{r',1}$ (which is always an MMS theory) embedded into the factor $\mathsf{X}_{r,1}$. When there is only a single kind of factor at level 1 in the Kac--Moody subalgebra of the Schellekens theory, we abbreviate this coset to $\mathbf{S}(\cdots \mathsf{X}_{r,1}\cdots)\big/ \mathsf{Y}_{r',1}$. The third column is the genus, i.e.\ a tuple consisting of the central charge and the modular tensor category which characterizes the representation theory of the chiral algebra. Sem stands for the ``Semion MTC'', Fib stands for the ``Fibonacci MTC'', and $\overline{\mathrm{Sem}}$, $\overline{\mathrm{Fib}}$ are their complex conjugates (see \cite{rowell2009classification} for details). The fourth column is the conformal dimension of the non-identity primary and the fifth column is the Wronskian index. The sixth column provides a subalgebra of the theory with the same central charge (i.e.\ a conformal subalgebra). This subalgebra captures the full Kac--Moody symmetry of the theory; Abelian factors are represented by $\mathsf{U}_1$. In the cases that the Kac--Moody subalgebra does not saturate the central charge of the full theory, the specified subalgebra includes a tensor product of minimal model chiral algebras, $\mathsf{L}_c$ with $c<1$, which make up the difference. The seventh column indicates the degeneracy $d$ of the non-identity primary. Finally, chiral algebras which are isomorphic to lattice VOAs are decorated by a $\star$.

\clearpage

\begin{table}[]
\begin{footnotesize}
    \begin{center}
    \begin{tabular}{r|c|c |c |c|c|c}
        No. & Theory & $(c,\mathscr{C})$ & $h$ & $\ell$ & Subalgebra & $d$ \\
        \midrule
$\star$ \rownumber & $\A[1,1]$ & $(1,\mathrm{Sem})$ & $\sfrac{1}{4}$ & $0$ & $\A[1,1]$ & $2 $ \\
\rownumber & $\G[2,1]$ & $(\sfrac{14}{5},\mathrm{Fib})$ & $\sfrac{2}{5}$ & $0$ & $\G[2,1]$ & $7 $ \\
\rownumber & $\F[4,1]\cong \E[8,1]\big/ \G[2,1]$ & $(\sfrac{26}{5},\overline{\mathrm{Fib}})$ & $\sfrac{3}{5}$ & $0$ & $\F[4,1]$ & $26 $ \\
$\star$ \rownumber & $\E[7,1]\cong \E[8,1]\big/\A[1,1]$ & $(7,\overline{\mathrm{Sem}})$ & $\sfrac{3}{4}$ & $0$ & $\E[7,1]$ & $56 $ \\
$\star$ \rownumber & $\E[8,1]\A[1,1]\cong \E[8,1]^2\big/\E[7,1]$ & $(9,\mathrm{Sem})$ & $\sfrac{1}{4}$ & $4$ & $\A[1,1]\E[8,1]$ & $2 $ \\
\rownumber & $\E[8,1]\G[2,1]\cong \E[8,1]^2\big/\F[4,1]$ & $(\sfrac{54}{5},\mathrm{Fib})$ & $\sfrac{2}{5}$ & $4$ & $\G[2,1]\E[8,1]$ & $7 $ \\
\rownumber & $\E[8,1]\F[4,1]\cong \E[8,1]^2\big/\G[2,1]$ & $(\sfrac{66}{5},\overline{\mathrm{Fib}})$ & $\sfrac{3}{5}$ & $4$ & $\F[4,1]\E[8,1]$ & $26 $ \\
\rownumber & $\D[16,1]^+\big/\G[2,1]$ & $(\sfrac{66}{5},\overline{\mathrm{Fib}})$ & $\sfrac{3}{5}$ & $4$ & $\B[12,1]\mathsf{L}_{\sfrac{7}{10}}$ & $26 $ \\
$\star$ \rownumber & $\E[8,1]\E[7,1]\cong\E[8,1]^2\big/\A[1,1]$ & $(15,\overline{\mathrm{Sem}})$ & $\sfrac{3}{4}$ & $4$ & $\E[7,1]\E[8,1]$ & $56 $ \\
$\star$ \rownumber & $\D[16,1]^+\big/\A[1,1]$ & $(15,\overline{\mathrm{Sem}})$ & $\sfrac{3}{4}$ & $4$ & $\D[14,1]\A[1,1]$ & $56 $ \\
$\star$ \rownumber & $\SC[\D[10,1]\E[7,1]^{2}][(\E[7,1]\hookrightarrow\E[7,1])]$ & $(17,\mathrm{Sem})$ & $\sfrac{5}{4}$ & $2$ & $\D[10,1]\E[7,1]$ & $1632 $ \\
$\star$ \rownumber & $\SC[\A[17,1]\E[7,1]][(\E[7,1]\hookrightarrow\E[7,1])]$ & $(17,\mathrm{Sem})$ & $\sfrac{5}{4}$ & $2$ & $\A[17,1]$ & $1632 $ \\
$\star$ \rownumber & $\E[8,1]^{2}\A[1,1]\cong \E[8,1]^3\big/\E[7,1]$ & $(17,\mathrm{Sem})$ & $\sfrac{1}{4}$ & $8$ & $\E[8,1]^{2}\A[1,1]$ & $2 $ \\
$\star$ \rownumber & $\SC[\D[16,1]\E[8,1]][(\E[7,1]\hookrightarrow\E[8,1])]$ & $(17,\mathrm{Sem})$ & $\sfrac{1}{4}$ & $8$ & $\D[16,1]\A[1,1]$ & $2 $ \\
\rownumber & $\SC[\C[8,1]\F[4,1]^{2}][(\F[4,1]\hookrightarrow\F[4,1])]$ & $(\sfrac{94}{5},\mathrm{Fib})$ & $\sfrac{7}{5}$ & $2$ & $\C[8,1]\F[4,1]$ & $4794 $ \\
\rownumber & $\SC[\E[7,2]\B[5,1]\F[4,1]][(\F[4,1]\hookrightarrow\F[4,1])]$ & $(\sfrac{94}{5},\mathrm{Fib})$ & $\sfrac{7}{5}$ & $2$ & $\E[7,2]\B[5,1]$ & $4794 $ \\
\rownumber & $\SC[\E[6,1]^{4}][\F[4,1]]$ & $(\sfrac{94}{5},\mathrm{Fib})$ & $\sfrac{2}{5}$ & $8$ & $\E[6,1]^{3}\mathsf{L}_{\sfrac{4}{5}}$ & $1 $ \\
\rownumber & $\SC[\A[11,1]\D[7,1]\E[6,1]][(\F[4,1]\hookrightarrow\E[6,1])]$ & $(\sfrac{94}{5},\mathrm{Fib})$ & $\sfrac{2}{5}$ & $8$ & $\A[11,1]\D[7,1]\mathsf{L}_{\sfrac{4}{5}}$ & $1 $ \\
\rownumber & $\SC[\D[10,1]\E[7,1]^{2}][(\F[4,1]\hookrightarrow\E[7,1])]$ & $(\sfrac{94}{5},\mathrm{Fib})$ & $\sfrac{2}{5}$ & $8$ & $\D[10,1]\E[7,1]\A[1,3]$ & $3 $ \\
\rownumber & $\SC[\A[17,1]\E[7,1]][(\F[4,1]\hookrightarrow\E[7,1])]$ & $(\sfrac{94}{5},\mathrm{Fib})$ & $\sfrac{2}{5}$ & $8$ & $\A[17,1]\A[1,3]$ & $3 $ \\
\rownumber & $\E[8,1]^2\G[2,1]\cong \E[8,1]^{3}\big/\F[4,1]$ & $(\sfrac{94}{5},\mathrm{Fib})$ & $\sfrac{2}{5}$ & $8$ & $\E[8,1]^{2}\G[2,1]$ & $7 $ \\
\rownumber & $\SC[\D[16,1]\E[8,1]][(\F[4,1]\hookrightarrow\E[8,1])]$ & $(\sfrac{94}{5},\mathrm{Fib})$ & $\sfrac{2}{5}$ & $8$ & $\D[16,1]\G[2,1]$ & $7 $ \\
\rownumber & $\SC[\E[6,3]\G[2,1]^{3}][\G[2,1]]$ & $(\sfrac{106}{5},\overline{\mathrm{Fib}})$ & $\sfrac{8}{5}$ & $2$ & $\E[6,3]\G[2,1]^{2}$ & $15847 $ \\
\rownumber & $\SC[\D[7,3]\A[3,1]\G[2,1]][(\G[2,1]\hookrightarrow\G[2,1])]$ & $(\sfrac{106}{5},\overline{\mathrm{Fib}})$ & $\sfrac{8}{5}$ & $2$ & $\D[7,3]\A[3,1]$ & $15847 $ \\
\rownumber & $\SC[\D[6,2]\C[4,1]\B[3,1]^{2}][(\G[2,1]\hookrightarrow\B[3,1])]$ & $(\sfrac{106}{5},\overline{\mathrm{Fib}})$ & $\sfrac{3}{5}$ & $8$ & $\D[6,2]\C[4,1]\B[3,1]\mathsf{L}_{\sfrac{7}{10}}$ & $1 $ \\
\rownumber & $\SC[\A[9,2]\A[4,1]\B[3,1]][(\G[2,1]\hookrightarrow\B[3,1])]$ & $(\sfrac{106}{5},\overline{\mathrm{Fib}})$ & $\sfrac{3}{5}$ & $8$ & $\A[9,2]\A[4,1]\mathsf{L}_{\sfrac{7}{10}}$ & $1 $ \\
\rownumber & $\SC[\D[4,1]^{6}][\G[2,1]]$ & $(\sfrac{106}{5},\overline{\mathrm{Fib}})$ & $\sfrac{3}{5}$ & $8$ & $\D[4,1]^{5}\mathsf{L}_{\sfrac{1}{2}}\mathsf{L}_{\sfrac{7}{10}}$ & $2 $ \\
\rownumber & $\SC[\A[5,1]^{4}\D[4,1]][(\G[2,1]\hookrightarrow\D[4,1])]$ & $(\sfrac{106}{5},\overline{\mathrm{Fib}})$ & $\sfrac{3}{5}$ & $8$ & $\A[5,1]^{4}\mathsf{L}_{\sfrac{1}{2}}\mathsf{L}_{\sfrac{7}{10}}$ & $2 $ \\
\rownumber & $\SC[\D[8,2]\B[4,1]^{2}][\G[2,1]]$ & $(\sfrac{106}{5},\overline{\mathrm{Fib}})$ & $\sfrac{3}{5}$ & $8$ & $\D[8,2]\B[4,1]\mathsf{U}_1\mathsf{L}_{\sfrac{7}{10}}$ & $3 $ \\
\rownumber & $\SC[\C[6,1]^{2}\B[4,1]][(\G[2,1]\hookrightarrow\B[4,1])]$ & $(\sfrac{106}{5},\overline{\mathrm{Fib}})$ & $\sfrac{3}{5}$ & $8$ & $\C[6,1]^{2}\mathsf{U}_1\mathsf{L}_{\sfrac{7}{10}}$ & $3 $ \\
\rownumber & $\SC[\A[7,1]^{2}\D[5,1]^{2}][(\G[2,1]\hookrightarrow\D[5,1])]$ & $(\sfrac{106}{5},\overline{\mathrm{Fib}})$ & $\sfrac{3}{5}$ & $8$ & $\A[7,1]^{2}\D[5,1]\A[1,2]\mathsf{L}_{\sfrac{7}{10}}$ & $4 $ \\
\rownumber & $\SC[\C[8,1]\F[4,1]^{2}][(\G[2,1]\hookrightarrow\F[4,1])]$ & $(\sfrac{106}{5},\overline{\mathrm{Fib}})$ & $\sfrac{3}{5}$ & $8$ & $\C[8,1]\F[4,1]\A[1,8]$ & $5 $ \\
\rownumber & $\SC[\E[7,2]\B[5,1]\F[4,1]][(\G[2,1]\hookrightarrow\B[5,1])]$ & $(\sfrac{106}{5},\overline{\mathrm{Fib}})$ & $\sfrac{3}{5}$ & $8$ & $\E[7,2]\A[1,1]^{2}\F[4,1]\mathsf{L}_{\sfrac{7}{10}}$ & $5 $ \\
\rownumber & $\SC[\E[7,2]\B[5,1]\F[4,1]][(\G[2,1]\hookrightarrow\F[4,1])]$ & $(\sfrac{106}{5},\overline{\mathrm{Fib}})$ & $\sfrac{3}{5}$ & $8$ & $\E[7,2]\B[5,1]\A[1,8]$ & $5 $ \\
\rownumber & $\SC[\D[6,1]^{4}][\G[2,1]]$ & $(\sfrac{106}{5},\overline{\mathrm{Fib}})$ & $\sfrac{3}{5}$ & $8$ & $\D[6,1]^{3}\B[2,1]\mathsf{L}_{\sfrac{7}{10}}$ & $6 $ \\
\rownumber & $\SC[\A[9,1]^{2}\D[6,1]][(\G[2,1]\hookrightarrow\D[6,1])]$ & $(\sfrac{106}{5},\overline{\mathrm{Fib}})$ & $\sfrac{3}{5}$ & $8$ & $\A[9,1]^{2}\B[2,1]\mathsf{L}_{\sfrac{7}{10}}$ & $6 $ \\
\rownumber & $\SC[\C[10,1]\B[6,1]][(\G[2,1]\hookrightarrow\B[6,1])]$ & $(\sfrac{106}{5},\overline{\mathrm{Fib}})$ & $\sfrac{3}{5}$ & $8$ & $\C[10,1]\A[3,1]\mathsf{L}_{\sfrac{7}{10}}$ & $7 $ \\
\rownumber & $\SC[\E[6,1]^{4}][\G[2,1]]$ & $(\sfrac{106}{5},\overline{\mathrm{Fib}})$ & $\sfrac{3}{5}$ & $8$ & $\E[6,1]^{3}\A[2,2]$ & $8 $ \\
\rownumber & $\SC[\A[11,1]\D[7,1]\E[6,1]][(\G[2,1]\hookrightarrow\D[7,1])]$ & $(\sfrac{106}{5},\overline{\mathrm{Fib}})$ & $\sfrac{3}{5}$ & $8$ & $\A[11,1]\B[3,1]\E[6,1]\mathsf{L}_{\sfrac{7}{10}}$ & $8 $ \\
\rownumber & $\SC[\A[11,1]\D[7,1]\E[6,1]][(\G[2,1]\hookrightarrow\E[6,1])]$ & $(\sfrac{106}{5},\overline{\mathrm{Fib}})$ & $\sfrac{3}{5}$ & $8$ & $\A[11,1]\D[7,1]\A[2,2]$ & $8 $ 
    \end{tabular}
    \end{center}
\end{footnotesize}
\end{table}

\clearpage

   \begin{table}[]
\begin{footnotesize}
    \begin{center}
    \begin{tabular}{r|c|c |c |c|c|c}
        No. & Theory & $(c,\mathscr{C})$ & $h$ & $\ell$ & Subalgebra & $d$ \\
        \midrule
       \rownumber & $\SC[\D[8,1]^{3}][\G[2,1]]$ & $(\sfrac{106}{5},\overline{\mathrm{Fib}})$ & $\sfrac{3}{5}$ & $8$ & $\D[8,1]^{2}\B[4,1]\mathsf{L}_{\sfrac{7}{10}}$ & $10 $ \\
\rownumber & $\SC[\E[8,2]\B[8,1]][\G[2,1]]$ & $(\sfrac{106}{5},\overline{\mathrm{Fib}})$ & $\sfrac{3}{5}$ & $8$ & $\E[8,2]\D[5,1]\mathsf{L}_{\sfrac{7}{10}}$ & $11 $ \\
\rownumber & $\SC[\A[15,1]\D[9,1]][(\G[2,1]\hookrightarrow\D[9,1])]$ & $(\sfrac{106}{5},\overline{\mathrm{Fib}})$ & $\sfrac{3}{5}$ & $8$ & $\A[15,1]\B[5,1]\mathsf{L}_{\sfrac{7}{10}}$ & $12 $ \\
\rownumber & $\SC[\D[10,1]\E[7,1]^{2}][(\G[2,1]\hookrightarrow\D[10,1])]$ & $(\sfrac{106}{5},\overline{\mathrm{Fib}})$ & $\sfrac{3}{5}$ & $8$ & $\B[6,1]\E[7,1]^{2}\mathsf{L}_{\sfrac{7}{10}}$ & $14 $ \\
\rownumber & $\SC[\D[10,1]\E[7,1]^{2}][(\G[2,1]\hookrightarrow\E[7,1])]$ & $(\sfrac{106}{5},\overline{\mathrm{Fib}})$ & $\sfrac{3}{5}$ & $8$ & $\D[10,1]\E[7,1]\C[3,1]$ & $14 $ \\
\rownumber & $\SC[\A[17,1]\E[7,1]][(\G[2,1]\hookrightarrow\E[7,1])]$ & $(\sfrac{106}{5},\overline{\mathrm{Fib}})$ & $\sfrac{3}{5}$ & $8$ & $\A[17,1]\C[3,1]$ & $14 $ \\
\rownumber & $\SC[\D[12,1]^{2}][\G[2,1]]$ & $(\sfrac{106}{5},\overline{\mathrm{Fib}})$ & $\sfrac{3}{5}$ & $8$ & $\D[12,1]\B[8,1]\mathsf{L}_{\sfrac{7}{10}}$ & $18 $ \\
\rownumber & $\E[8,1]^2\F[4,1]\cong\E[8,1]^{3}\big/\G[2,1]$ & $(\sfrac{106}{5},\overline{\mathrm{Fib}})$ & $\sfrac{3}{5}$ & $8$ & $\E[8,1]^{2}\F[4,1]$ & $26 $ \\
\rownumber & $\SC[\D[16,1]\E[8,1]][(\G[2,1]\hookrightarrow\D[16,1])]$ & $(\sfrac{106}{5},\overline{\mathrm{Fib}})$ & $\sfrac{3}{5}$ & $8$ & $\B[12,1]\E[8,1]\mathsf{L}_{\sfrac{7}{10}}$ & $26 $ \\
\rownumber & $\D[16,1]^+\F[4,1]\cong\SC[\D[16,1]\E[8,1]][(\G[2,1]\hookrightarrow\E[8,1])]$ & $(\sfrac{106}{5},\overline{\mathrm{Fib}})$ & $\sfrac{3}{5}$ & $8$ & $\D[16,1]\F[4,1]$ & $26 $ \\
\rownumber & $\SC[\D[24,1]][\G[2,1]]$ & $(\sfrac{106}{5},\overline{\mathrm{Fib}})$ & $\sfrac{3}{5}$ & $8$ & $\B[20,1]\mathsf{L}_{\sfrac{7}{10}}$ & $42 $ \\
$\star$ \rownumber & $\SC[\A[1,1]^{24}][\A[1,1]]$ & $(23,\overline{\mathrm{Sem}})$ & $\sfrac{7}{4}$ & $2$ & $\A[1,1]^{23}$ & $32384 $ \\
\rownumber & $\SC[\A[3,2]^{4}\A[1,1]^{4}][\A[1,1]]$ & $(23,\overline{\mathrm{Sem}})$ & $\sfrac{7}{4}$ & $2$ & $\A[3,2]^{4}\A[1,1]^{3}$ & $32384 $ \\
\rownumber & $\SC[\A[5,3]\D[4,3]\A[1,1]^{3}][\A[1,1]]$ & $(23,\overline{\mathrm{Sem}})$ & $\sfrac{7}{4}$ & $2$ & $\A[5,3]\D[4,3]\A[1,1]^{2}$ & $32384 $ \\
\rownumber & $\SC[\A[7,4]\A[1,1]'\A[1,1]^{2}][\A[1,1]]$ & $(23,\overline{\mathrm{Sem}})$ & $\sfrac{7}{4}$ & $2$ & $\A[7,4]\A[1,1]^{2}$ & $32384 $ \\
\rownumber & $\SC[\A[7,4]\A[1,1]'\A[1,1]^{2}][\A[1,1]']$ & $(23,\overline{\mathrm{Sem}})$ & $\sfrac{7}{4}$ & $2$ & $\A[7,4]\A[1,1]^{2}$ & $32384 $ \\
\rownumber & $\SC[\D[5,4]\C[3,2]\A[1,1]^{2}][\A[1,1]]$ & $(23,\overline{\mathrm{Sem}})$ & $\sfrac{7}{4}$ & $2$ & $\D[5,4]\C[3,2]\A[1,1]$ & $32384 $ \\
\rownumber & $\SC[\D[6,5]\A[1,1]\A[1,1]'][\A[1,1]]$ & $(23,\overline{\mathrm{Sem}})$ & $\sfrac{7}{4}$ & $2$ & $\D[6,5]\A[1,1]$ & $32384$ \\
\rownumber & $\SC[\D[6,5]\A[1,1]\A[1,1]'][\A[1,1]']$ & $(23,\overline{\mathrm{Sem}})$ & $\sfrac{7}{4}$ & $2$ & $\D[6,5]\A[1,1]$ & $32384$ \\
\rownumber & $\SC[\C[5,3]\G[2,2]\A[1,1]][\A[1,1]]$ & $(23,\overline{\mathrm{Sem}})$ & $\sfrac{7}{4}$ & $2$ & $\C[5,3]\G[2,2]$ & $32384 $ \\
$\star$ \rownumber & $\SC[\A[2,1]^{12}][\A[1,1]]$ & $(23,\overline{\mathrm{Sem}})$ & $\sfrac{3}{4}$ & $8$ & $\A[2,1]^{11}\mathsf{U}_1$ & $2 $ \\
\rownumber & $\SC[\D[4,2]^{2}\B[2,1]^{4}][\A[1,1]]$ & $(23,\overline{\mathrm{Sem}})$ & $\sfrac{3}{4}$ & $8$ & $\D[4,2]^{2}\B[2,1]^{3}\A[1,1]\mathsf{L}_{\sfrac12}$ & $2 $ \\
\rownumber & $\SC[\A[5,2]^{2}\B[2,1]\A[2,1]^{2}][(\A[1,1]\hookrightarrow\A[2,1])]$ & $(23,\overline{\mathrm{Sem}})$ & $\sfrac{3}{4}$ & $8$ & $\A[5,2]^{2}\C[2,1]\A[2,1]\mathsf{U}_1$ & $2 $ \\
\rownumber & $\SC[\A[5,2]^{2}\B[2,1]\A[2,1]^{2}][(\A[1,1]\hookrightarrow\B[2,1])]$ & $(23,\overline{\mathrm{Sem}})$ & $\sfrac{3}{4}$ & $8$ & $\A[5,2]^{2}\A[1,1]\A[2,1]^{2}\mathsf{L}_{\sfrac12}$ & $2 $ \\
\rownumber & $\SC[\A[8,3]\A[2,1]^{2}][\A[1,1]]$ & $(23,\overline{\mathrm{Sem}})$ & $\sfrac{3}{4}$ & $8$ & $\A[8,3]\A[2,1]\mathsf{U}_1$ & $2 $ \\
\rownumber & $\SC[\E[6,4]\B[2,1]\A[2,1]][(\A[1,1]\hookrightarrow\B[2,1])]$ & $(23,\overline{\mathrm{Sem}})$ & $\sfrac{3}{4}$ & $8$ & $\E[6,4]\A[1,1]\A[2,1]\mathsf{L}_{\sfrac12}$ & $2 $ \\
\rownumber & $\SC[\E[6,4]\B[2,1]\A[2,1]][(\A[1,1]\hookrightarrow\A[2,1])]$ & $(23,\overline{\mathrm{Sem}})$ & $\sfrac{3}{4}$ & $8$ & $\E[6,4]\C[2,1]\mathsf{U}_1$ & $2 $ \\
$\star$ \rownumber & $\SC[\A[3,1]^{8}][\A[1,1]]$ & $(23,\overline{\mathrm{Sem}})$ & $\sfrac{3}{4}$ & $8$ & $\A[3,1]^{7}\A[1,1]\mathsf{U}_1$ & $4 $ \\
\rownumber & $\SC[\D[5,2]^{2}\A[3,1]^{2}][\A[1,1]]$ & $(23,\overline{\mathrm{Sem}})$ & $\sfrac{3}{4}$ & $8$ & $\D[5,2]^2\A[3,1]\A[1,1]\mathsf{U}_1$ & $4 $ \\
\rownumber & $\SC[\E[6,3]\G[2,1]^{3}][\A[1,1]]$ & $(23,\overline{\mathrm{Sem}})$ & $\sfrac{3}{4}$ & $8$ & $\E[6,3]\G[2,1]^{2}\A[1,3]$ & $4 $ \\
\rownumber & $\SC[\A[7,2]\C[3,1]^{2}\A[3,1]][(\A[1,1]\hookrightarrow\A[3,1])]$ & $(23,\overline{\mathrm{Sem}})$ & $\sfrac{3}{4}$ & $8$ & $\A[7,2]\C[3,1]^{2}\A[1,1]\mathsf{U}_1$ & $4 $ \\
\rownumber & $\SC[\A[7,2]\C[3,1]^{2}\A[3,1]][(\A[1,1]\hookrightarrow\C[3,1])]$ & $(23,\overline{\mathrm{Sem}})$ & $\sfrac{3}{4}$ & $8$ & $\A[7,2]\C[3,1]\B[2,1]\A[3,1]\mathsf{L}_{\sfrac{7}{10}}$ & $4 $ \\
\rownumber & $\SC[\D[7,3]\A[3,1]\G[2,1]][(\A[1,1]\hookrightarrow\G[2,1])]$ & $(23,\overline{\mathrm{Sem}})$ & $\sfrac{3}{4}$ & $8$ & $\D[7,3]\A[3,1]\A[1,3]$ & $4 $ \\
\rownumber & $\SC[\D[7,3]\A[3,1]\G[2,1]][(\A[1,1]\hookrightarrow\A[3,1])]$ & $(23,\overline{\mathrm{Sem}})$ & $\sfrac{3}{4}$ & $8$ & $\D[7,3]\G[2,1]\A[1,1]\mathsf{U}_1$ & $4 $ \\
\rownumber & $\SC[\C[7,2]\A[3,1]][\A[1,1]]$ & $(23,\overline{\mathrm{Sem}})$ & $\sfrac{3}{4}$ & $8$ & $\C[7,2]\A[1,1]\mathsf{U}_1$ & $4 $ \\
$\star$ \rownumber & $\SC[\A[4,1]^{6}][\A[1,1]]$ & $(23,\overline{\mathrm{Sem}})$ & $\sfrac{3}{4}$ & $8$ & $\A[4,1]^{5}\A[2,1]\mathsf{U}_1$ & $6 $ \\
\rownumber & $\SC[\C[4,1]^{4}][\A[1,1]]$ & $(23,\overline{\mathrm{Sem}})$ & $\sfrac{3}{4}$ & $8$ & $\C[4,1]^{3}\C[3,1]\mathsf{L}_{\sfrac45}$ & $6 $ \\
\rownumber & $\SC[\D[6,2]\C[4,1]\B[3,1]^{2}][(\A[1,1]\hookrightarrow\C[4,1])]$ & $(23,\overline{\mathrm{Sem}})$ & $\sfrac{3}{4}$ & $8$ & $\D[6,2]\C[3,1]\B[3,1]^{2}\mathsf{L}_{\sfrac45}$ & $6 $ \\
\rownumber & $\SC[\D[6,2]\C[4,1]\B[3,1]^{2}][(\A[1,1]\hookrightarrow\B[3,1])]$ & $(23,\overline{\mathrm{Sem}})$ & $\sfrac{3}{4}$ & $8$ & $\D[6,2]\C[4,1]\B[3,1]\A[1,2]\A[1,1]$ & $6 $ \\
\rownumber & $\SC[\A[9,2]\A[4,1]\B[3,1]][(\A[1,1]\hookrightarrow\A[4,1])]$ & $(23,\overline{\mathrm{Sem}})$ & $\sfrac{3}{4}$ & $8$ & $\A[9,2]\A[2,1]\B[3,1]\mathsf{U}_1$ & $6 $ \\
\rownumber & $\SC[\A[9,2]\A[4,1]\B[3,1]][(\A[1,1]\hookrightarrow\B[3,1])]$ & $(23,\overline{\mathrm{Sem}})$ & $\sfrac{3}{4}$ & $8$ & $\A[9,2]\A[4,1]\A[1,2]\A[1,1]$ & $6 $ \\
$\star$ \rownumber & $\SC[\D[4,1]^{6}][\A[1,1]]$ & $(23,\overline{\mathrm{Sem}})$ & $\sfrac{3}{4}$ & $8$ & $\D[4,1]^{5}\A[1,1]\A[1,1]\A[1,1]$ & $8 $ 
    \end{tabular}
    \end{center}
\end{footnotesize}
\end{table}

\clearpage

   \begin{table}[]
\begin{footnotesize}
    \begin{center}
    \begin{tabular}{r|c|c |c |c|c|c}
        No. & Theory & $(c,\mathscr{C})$ & $h$ & $\ell$ & Subalgebra & $d$ \\
        \midrule
        $\star$ \rownumber & $\SC[\A[5,1]^{4}\D[4,1]][(\A[1,1]\hookrightarrow\A[5,1])]$ & $(23,\overline{\mathrm{Sem}})$ & $\sfrac{3}{4}$ & $8$ & $\A[5,1]^{3}\A[3,1]\D[4,1]\mathsf{U}_1$ & $8 $ \\
        $\star$ \rownumber & $\SC[\A[5,1]^{4}\D[4,1]][(\A[1,1]\hookrightarrow\D[4,1])]$ & $(23,\overline{\mathrm{Sem}})$ & $\sfrac{3}{4}$ & $8$ & $\A[5,1]^{4}\A[1,1]^3$ & $8 $ \\
\rownumber & $\SC[\E[6,2]\C[5,1]\A[5,1]][(\A[1,1]\hookrightarrow\C[5,1])]$ & $(23,\overline{\mathrm{Sem}})$ & $\sfrac{3}{4}$ & $8$ & $\E[6,2]\C[4,1]\A[5,1]\mathsf{L}_{\sfrac67}$ & $8 $ \\
\rownumber & $\SC[\E[6,2]\C[5,1]\A[5,1]][(\A[1,1]\hookrightarrow\A[5,1])]$ & $(23,\overline{\mathrm{Sem}})$ & $\sfrac{3}{4}$ & $8$ & $\E[6,2]\C[5,1]\A[3,1]\mathsf{U}_1$ & $8 $ \\
\rownumber & $\SC[\E[7,3]\A[5,1]][\A[1,1]]$ & $(23,\overline{\mathrm{Sem}})$ & $\sfrac{3}{4}$ & $8$ & $\E[7,3]\A[3,1]\mathsf{U}_1$ & $8 $ \\
$\star$ \rownumber & $\SC[\A[6,1]^{4}][\A[1,1]]$ & $(23,\overline{\mathrm{Sem}})$ & $\sfrac{3}{4}$ & $8$ & $\A[6,1]^{3}\A[4,1]\mathsf{U}_1$ & $10 $ \\
\rownumber & $\SC[\D[8,2]\B[4,1]^{2}][\A[1,1]]$ & $(23,\overline{\mathrm{Sem}})$ & $\sfrac{3}{4}$ & $8$ & $\D[8,2]\B[4,1]\B[2,1]\A[1,1]$ & $10 $ \\
\rownumber & $\SC[\C[6,1]^{2}\B[4,1]][(\A[1,1]\hookrightarrow\C[6,1])]$ & $(23,\overline{\mathrm{Sem}})$ & $\sfrac{3}{4}$ & $8$ & $\C[6,1]\C[5,1]\B[4,1]\mathsf{L}_{\sfrac{25}{28}}$ & $10 $ \\
\rownumber & $\SC[\C[6,1]^{2}\B[4,1]][(\A[1,1]\hookrightarrow\B[4,1])]$ & $(23,\overline{\mathrm{Sem}})$ & $\sfrac{3}{4}$ & $8$ & $\C[6,1]^{2}\B[2,1]\A[1,1]$ & $10 $ \\
$\star$ \rownumber & $\SC[\A[7,1]^{2}\D[5,1]^{2}][(\A[1,1]\hookrightarrow\A[7,1])]$ & $(23,\overline{\mathrm{Sem}})$ & $\sfrac{3}{4}$ & $8$ & $\A[7,1]\A[5,1]\D[5,1]^{2}\mathsf{U}_1$ & $12 $ \\
$\star$ \rownumber & $\SC[\A[7,1]^{2}\D[5,1]^{2}][(\A[1,1]\hookrightarrow\D[5,1])]$ & $(23,\overline{\mathrm{Sem}})$ & $\sfrac{3}{4}$ & $8$ & $\A[7,1]^{2}\A[3,1]\A[1,1]\D[5,1]$ & $12 $ \\
\rownumber & $\SC[\D[9,2]\A[7,1]][\A[1,1]]$ & $(23,\overline{\mathrm{Sem}})$ & $\sfrac{3}{4}$ & $8$ & $\D[9,2]\A[5,1]\mathsf{U}_1$ & $12 $ \\
$\star$ \rownumber & $\SC[\A[8,1]^{3}][\A[1,1]]$ & $(23,\overline{\mathrm{Sem}})$ & $\sfrac{3}{4}$ & $8$ & $\A[8,1]^{2}\A[6,1]\mathsf{U}_1$ & $14 $ \\
\rownumber & $\SC[\C[8,1]\F[4,1]^{2}][(\A[1,1]\hookrightarrow\C[8,1])]$ & $(23,\overline{\mathrm{Sem}})$ & $\sfrac{3}{4}$ & $8$ & $\C[7,1]\F[4,1]^{2}L_{\sfrac{14}{15}}$ & $14 $ \\
\rownumber & $\SC[\C[8,1]\F[4,1]^{2}][(\A[1,1]\hookrightarrow\F[4,1])]$ & $(23,\overline{\mathrm{Sem}})$ & $\sfrac{3}{4}$ & $8$ & $\C[8,1]\F[4,1]\C[3,1]$ & $14 $ \\
\rownumber & $\SC[\E[7,2]\B[5,1]\F[4,1]][(\A[1,1]\hookrightarrow\B[5,1])]$ & $(23,\overline{\mathrm{Sem}})$ & $\sfrac{3}{4}$ & $8$ & $\E[7,2]\B[3,1]\A[1,1]\F[4,1]$ & $14 $ \\
\rownumber & $\SC[\E[7,2]\B[5,1]\F[4,1]][(\A[1,1]\hookrightarrow\F[4,1])]$ & $(23,\overline{\mathrm{Sem}})$ & $\sfrac{3}{4}$ & $8$ & $\E[7,2]\B[5,1]\C[3,1]$ & $14 $ \\
$\star$ \rownumber & $\SC[\D[6,1]^{4}][\A[1,1]]$ & $(23,\overline{\mathrm{Sem}})$ & $\sfrac{3}{4}$ & $8$ & $\D[6,1]^{3}\D[4,1]\A[1,1]$ & $16 $ \\
$\star$ \rownumber & $\SC[\A[9,1]^{2}\D[6,1]][(\A[1,1]\hookrightarrow\A[9,1])]$ & $(23,\overline{\mathrm{Sem}})$ & $\sfrac{3}{4}$ & $8$ & $\A[9,1]\A[7,1]\D[6,1]\mathsf{U}_1$ & $16 $ \\
$\star$ \rownumber & $\SC[\A[9,1]^{2}\D[6,1]][(\A[1,1]\hookrightarrow\D[6,1])]$ & $(23,\overline{\mathrm{Sem}})$ & $\sfrac{3}{4}$ & $8$ & $\A[9,1]^{2}\D[4,1]\A[1,1]$ & $16 $ \\
\rownumber & $\SC[\C[10,1]\B[6,1]][(\A[1,1]\hookrightarrow\C[10,1])]$ & $(23,\overline{\mathrm{Sem}})$ & $\sfrac{3}{4}$ & $8$ & $\C[9,1]\B[6,1]L_{\sfrac{21}{22}}$ & $18 $ \\
\rownumber & $\SC[\C[10,1]\B[6,1]][(\A[1,1]\hookrightarrow\B[6,1])]$ & $(23,\overline{\mathrm{Sem}})$ & $\sfrac{3}{4}$ & $8$ & $\C[10,1]\B[4,1]\A[1,1]$ & $18 $ \\
$\star$ \rownumber & $\SC[\E[6,1]^{4}][\A[1,1]]$ & $(23,\overline{\mathrm{Sem}})$ & $\sfrac{3}{4}$ & $8$ & $\E[6,1]^{3}\A[5,1]$ & $20 $ \\
$\star$ \rownumber & $\SC[\A[11,1]\D[7,1]\E[6,1]][(\A[1,1]\hookrightarrow\A[11,1])]$ & $(23,\overline{\mathrm{Sem}})$ & $\sfrac{3}{4}$ & $8$ & $\A[9,1]\D[7,1]\E[6,1]\mathsf{U}_1$ & $20 $ \\
$\star$ \rownumber & $\SC[\A[11,1]\D[7,1]\E[6,1]][(\A[1,1]\hookrightarrow\D[7,1])]$ & $(23,\overline{\mathrm{Sem}})$ & $\sfrac{3}{4}$ & $8$ & $\A[11,1]\D[5,1]\A[1,1]\E[6,1]$ & $20 $ \\
$\star$ \rownumber & $\SC[\A[11,1]\D[7,1]\E[6,1]][(\A[1,1]\hookrightarrow\E[6,1])]$ & $(23,\overline{\mathrm{Sem}})$ & $\sfrac{3}{4}$ & $8$ & $\A[11,1]\D[7,1]\A[5,1]$ & $20 $ \\
$\star$ \rownumber & $\SC[\A[12,1]^{2}][\A[1,1]]$ & $(23,\overline{\mathrm{Sem}})$ & $\sfrac{3}{4}$ & $8$ & $\A[12,1]\A[10,1]\mathsf{U}_1$ & $22 $ \\
$\star$ \rownumber & $\SC[\D[8,1]^{3}][\A[1,1]]$ & $(23,\overline{\mathrm{Sem}})$ & $\sfrac{3}{4}$ & $8$ & $\D[8,1]^{2}\D[6,1]\A[1,1]$ & $24 $ \\
\rownumber & $\SC[\E[8,2]\B[8,1]][\A[1,1]]$ & $(23,\overline{\mathrm{Sem}})$ & $\sfrac{3}{4}$ & $8$ & $\E[8,2]\B[6,1]\A[1,1]$ & $26 $ \\
$\star$ \rownumber & $\SC[\A[15,1]\D[9,1]][(\A[1,1]\hookrightarrow\A[15,1])]$ & $(23,\overline{\mathrm{Sem}})$ & $\sfrac{3}{4}$ & $8$ & $\A[13,1]\D[9,1]\mathsf{U}_1$ & $28 $ \\
$\star$ \rownumber & $\SC[\A[15,1]\D[9,1]][(\A[1,1]\hookrightarrow\D[9,1])]$ & $(23,\overline{\mathrm{Sem}})$ & $\sfrac{3}{4}$ & $8$ & $\A[15,1]\D[7,1]\A[1,1]$ & $28 $ \\
$\star$ \rownumber & $\SC[\D[10,1]\E[7,1]^{2}][(\A[1,1]\hookrightarrow\D[10,1])]$ & $(23,\overline{\mathrm{Sem}})$ & $\sfrac{3}{4}$ & $8$ & $\D[8,1]\A[1,1]\E[7,1]^{2}$ & $32 $ \\
$\star$ \rownumber & $\SC[\D[10,1]\E[7,1]^{2}][(\A[1,1]\hookrightarrow\E[7,1])]$ & $(23,\overline{\mathrm{Sem}})$ & $\sfrac{3}{4}$ & $8$ & $\D[10,1]\E[7,1]\D[6,1]$ & $32 $ \\
$\star$ \rownumber & $\SC[\A[17,1]\E[7,1]][(\A[1,1]\hookrightarrow\A[17,1])]$ & $(23,\overline{\mathrm{Sem}})$ & $\sfrac{3}{4}$ & $8$ & $\A[15,1]\E[7,1]\mathsf{U}_1$ & $32 $ \\
$\star$ \rownumber & $\SC[\A[17,1]\E[7,1]][(\A[1,1]\hookrightarrow\E[7,1])]$ & $(23,\overline{\mathrm{Sem}})$ & $\sfrac{3}{4}$ & $8$ & $\A[17,1]\D[6,1]$ & $32 $ \\
$\star$ \rownumber & $\SC[\D[12,1]^{2}][\A[1,1]]$ & $(23,\overline{\mathrm{Sem}})$ & $\sfrac{3}{4}$ & $8$ & $\D[12,1]\D[10,1]\A[1,1]$ & $40 $ \\
$\star$ \rownumber & $\SC[\A[24,1]][\A[1,1]]$ & $(23,\overline{\mathrm{Sem}})$ & $\sfrac{3}{4}$ & $8$ & $\A[22,1]\mathsf{U}_1$ & $46 $ \\
$\star$ \rownumber & $\E[8,1]^2\E[7,1]\cong\E[8,1]^{3}\big/\A[1,1]$ & $(23,\overline{\mathrm{Sem}})$ & $\sfrac{3}{4}$ & $8$ & $\E[8,1]^{2}\E[7,1]$ & $56 $ \\
$\star$ \rownumber & $\SC[\D[16,1]\E[8,1]][(\A[1,1]\hookrightarrow\D[16,1])]$ & $(23,\overline{\mathrm{Sem}})$ & $\sfrac{3}{4}$ & $8$ & $\D[14,1]\A[1,1]\E[8,1]$ & $56 $ \\
$\star$ \rownumber & $\D[16,1]^+ \E[7,1]\cong\SC[\D[16,1]\E[8,1]][(\A[1,1]\hookrightarrow\E[8,1])]$ & $(23,\overline{\mathrm{Sem}})$ & $\sfrac{3}{4}$ & $8$ & $\D[16,1]\E[7,1]$ & $56 $ \\
$\star$ \rownumber & $\SC[\D[24,1]][\A[1,1]]$ & $(23,\overline{\mathrm{Sem}})$ & $\sfrac{3}{4}$ & $8$ & $\D[22,1]\A[1,1]$ & $88 $ 
    \end{tabular}
    \end{center}
\end{footnotesize}
\end{table}

\clearpage

\section{Two-Dimensional Representations of the Modular Group}\label{app:reps}

Table \ref{tab:oddweight} contains all equivalence classes of two-dimensional irreducible representations $\varrho$ with finite image and $\varrho(S)^2=-1$. The first column is the label/number of the representation, and the second column contains the pair $(m_0,m_1)$ which determines $\varrho(S)$ and $\varrho(T)$ through Eq.\ \eqref{eqn:SMatrix} and Eq.\ \eqref{eqn:TMatrix}.

Table \ref{tab:evenweight} contains all equivalence classes of two-dimensional irreducible representations $\varrho$ with finite image and $\varrho(S)^2=+1$. The first column is the number $I$ which labels the equivalence class of the representation. The second column provides the associated rational numbers $(m_0,m_1)$ which can be used to reconstruct $\varrho(S)$ and $\varrho(T)$ through Eq.\ \eqref{eqn:SMatrix} and Eq.\ \eqref{eqn:TMatrix}. The third column computes $\hat\ell$, which is the Wronskian index modulo 6 of any vector-valued modular form which transforms covariantly with respect to the representation $\varrho$ (or $\varrho_{\mathrm{U}}$, $\varrho_{\mathrm{V}}$, or $\varrho_{\mathrm{W}}$). The fourth column contains $(\hat{c},\hat{h})$ which are respectively the value of the central charge modulo 24 and the conformal dimension of $\Phi$ modulo 1 for any unitary RCFT (if any exists) with characters which transform covariantly under either $\varrho$ or $\varrho_{\mathrm{U}}$. Likewise, $(\hat{c}_{\mathrm{V}},\hat{h}_{\mathrm{V}})$ are respectively the value of the central charge modulo 24 and the conformal dimension of $\Phi$ modulo 1 for any unitary RCFT (if any exists) with characters which transform covariantly under either $\varrho_{\mathrm{V}}$ or $\varrho_{\mathrm{W}}$. An entry is colored green if the modular representation is admissible, in the sense of \S\ref{subsec:RCFTbasics}. Finally, the last three columns identify the label of $\varrho^\ast$, $\omega\varrho^\ast$, and $\omega^2\varrho^\ast$, where $\varrho^\ast$ is the representation obtained by complex conjugating the entries of $\varrho$, and $\omega:\SL_2(\IZ)\to\mathbb{C}^\ast$ is the character of the modular group which assigns $\omega(T) = e^{-2\pi i/3}$ and $\omega(S) = 1$.

\begin{table}[h]
    \centering
    \begin{tabular}{c c c}
    \begin{tabular}[t]{c|c}
    No. & $(m_0,m_1)$ \\\midrule
     1 & (\sfrac{3}{4},\sfrac{1}{4})\\
3 & (\sfrac{11}{12},\sfrac{5}{12})\\
5 & (\sfrac{7}{12},\sfrac{1}{12})\\
7 & (\sfrac{2}{3},\sfrac{1}{3})\\
9 & (\sfrac{5}{6},\sfrac{1}{2})\\
11 & (\sfrac{2}{3},0)\\
13 & (\sfrac{5}{6},\sfrac{1}{6})\\
15 & (\sfrac{1}{3},0)\\
17 & (\sfrac{1}{2},\sfrac{1}{6})
    \end{tabular}
    &
        \begin{tabular}[t]{c|c}
    No. & $(m_0,m_1)$ \\\midrule
   19 & (\sfrac{7}{8},\sfrac{1}{8})\\
21 & (\sfrac{7}{24},\sfrac{1}{24})\\
23 & (\sfrac{11}{24},\sfrac{5}{24})\\
25 & (\sfrac{5}{8},\sfrac{3}{8})\\
27 & (\sfrac{19}{24},\sfrac{13}{24})\\
29 & (\sfrac{23}{24},\sfrac{17}{24})\\
31 & (\sfrac{4}{5},\sfrac{1}{5})\\
33 & (\sfrac{29}{30},\sfrac{11}{30})\\
35 & (\sfrac{8}{15},\sfrac{2}{15})
    \end{tabular}
    &
        \begin{tabular}[t]{c|c}
    No. & $(m_0,m_1)$ \\\midrule
    37 & (\sfrac{7}{10},\sfrac{3}{10})\\
39 & (\sfrac{13}{15},\sfrac{7}{15})\\
41 & (\sfrac{19}{30},\sfrac{1}{30})\\
43 & (\sfrac{3}{5},\sfrac{2}{5})\\
45 & (\sfrac{23}{30},\sfrac{17}{30})\\
47 & (\sfrac{14}{15},\sfrac{11}{15})\\
49 & (\sfrac{9}{10},\sfrac{1}{10})\\
51 & (\sfrac{4}{15},\sfrac{1}{15})\\
53 & (\sfrac{13}{30},\sfrac{7}{30})
    \end{tabular}
\end{tabular}
    \caption{Equivalence classes of two-dimensional irreducible representations $\varrho$ of $\SL_2(\IZ)$ with finite image and  $\varrho(S)^2=-1$.}
    \label{tab:oddweight}
\end{table}

\begin{table}[]
\begin{footnotesize}
    \centering
    \begin{tabular}{r|c  | c | c c c | c c c | c c c}
        No.&$(m_0,m_1)$ & $\hat\ell$ & $(\hat c,\hat h)$ & $\varrho$ & $\varrho_{\mathrm{U}}$ & $(\hat{c}_{\mathrm{V}},\hat{h}_{\mathrm{V}})$ & $\varrho_{\mathrm{V}}$ & $\varrho_{\mathrm{W}}$ & $\varrho^\ast$ & $\omega\varrho^\ast$ & $\omega^2\varrho^\ast$ \\ \midrule
2&(\sfrac{5}{6},\sfrac{1}{3}) & $0$ & $(4,\sfrac{1}{2})$ & 2& 2U & $(16,\sfrac{1}{2})$ & 2V& 2W & 6W & 2 & 4\\
4&(\sfrac{1}{2},0) & $4$ & $(12,\sfrac{1}{2})$ & 4& 4U & $(24,\sfrac{1}{2})$ & 4V& 4W & 4 & 6W & 2\\
6&(\sfrac{2}{3},\sfrac{1}{6}) & $2$ & $(8,\sfrac{1}{2})$ & 6& 6U & $(20,\sfrac{1}{2})$ & 6V& 6W & 2W & 4W & 6\\
8&(\sfrac{3}{4},\sfrac{5}{12}) & $0$ & $(6,\sfrac{2}{3})$ & 8& 8U & $(14,\sfrac{1}{3})$ & 8V& 8W & 18W & 14 & 10W\\
10&(\sfrac{11}{12},\sfrac{7}{12}) & $4$ & $(2,\sfrac{2}{3})$ & 10& 10U & $(10,\sfrac{1}{3})$ & 10V& 10W & 16W & 12 & 8W\\
12&(\sfrac{3}{4},\sfrac{1}{12}) & $2$ & $(6,\sfrac{1}{3})$ & 12& 12U & $(22,\sfrac{2}{3})$ & 12V& 12W & 14W & 10 & 18\\
14&(\sfrac{11}{12},\sfrac{1}{4}) & $0$ & $(2,\sfrac{1}{3})$ & 14& 14U & $(18,\sfrac{2}{3})$ & 14V& 14W & 12W & 8 & 16\\
16&(\sfrac{5}{12},\sfrac{1}{12}) & $4$ & $(14,\sfrac{2}{3})$ & 16& 16U & $(22,\sfrac{1}{3})$ & 16V& 16W & 10W & 18W & 14\\
18&(\sfrac{7}{12},\sfrac{1}{4}) & $2$ & $(10,\sfrac{2}{3})$ & 18& 18U & $(18,\sfrac{1}{3})$ & 18V& 18W & 8W & 16W & 12\\
20&(\sfrac{23}{24},\sfrac{5}{24}) & $0$ & $(1,\sfrac{1}{4})$ & \cellcolor{green!60}20& 20U & $(19,\sfrac{3}{4})$ & 20V& 20W & 30W & 26 & 22\\
22&(\sfrac{3}{8},\sfrac{1}{8}) & $4$ & $(15,\sfrac{3}{4})$ & \cellcolor{green!60}22& 22U & $(21,\sfrac{1}{4})$ & 22V& 22W & 28W & 24W & 20\\
24&(\sfrac{13}{24},\sfrac{7}{24}) & $2$ & $(11,\sfrac{3}{4})$ & 24& 24U & $(17,\sfrac{1}{4})$ & 24V& \cellcolor{green!60}24W & 26W & 22W & 30\\
26&(\sfrac{17}{24},\sfrac{11}{24}) & $0$ & $(7,\sfrac{3}{4})$ & \cellcolor{green!60}26& 26U & $(13,\sfrac{1}{4})$ & 26V& 26W & 24W & 20 & 28W\\
28&(\sfrac{7}{8},\sfrac{5}{8}) & $4$ & $(3,\sfrac{3}{4})$ & 28& 28U & $(9,\sfrac{1}{4})$ & 28V& \cellcolor{green!60}28W & 22W & 30 & 26W\\
30&(\sfrac{19}{24},\sfrac{1}{24}) & $2$ & $(5,\sfrac{1}{4})$ & 30& 30U & $(23,\sfrac{3}{4})$ & 30V& \cellcolor{green!60}30W & 20W & 28 & 24\\
32&(\sfrac{53}{60},\sfrac{17}{60}) & $0$ & $(\sfrac{14}{5},\sfrac{2}{5})$ & \cellcolor{green!60}32& 32U & $(\sfrac{86}{5},\sfrac{3}{5})$ & 32V& 32W & 42W & 38 & 34\\
34&(\sfrac{9}{20},\sfrac{1}{20}) & $4$ & $(\sfrac{66}{5},\sfrac{3}{5})$ & \cellcolor{green!60}34& 34U & $(\sfrac{114}{5},\sfrac{2}{5})$ & 34V& 34W & 40W & 36W & 32\\
36&(\sfrac{37}{60},\sfrac{13}{60}) & $2$ & $(\sfrac{46}{5},\sfrac{3}{5})$ & 36& 36U & $(\sfrac{94}{5},\sfrac{2}{5})$ & 36V& \cellcolor{green!60}36W & 38W & 34W & 42\\
38&(\sfrac{47}{60},\sfrac{23}{60}) & $0$ & $(\sfrac{26}{5},\sfrac{3}{5})$ & \cellcolor{green!60}38& 38U & $(\sfrac{74}{5},\sfrac{2}{5})$ & 38V& 38W & 36W & 32 & 40W\\
40&(\sfrac{19}{20},\sfrac{11}{20}) & $4$ & $(\sfrac{6}{5},\sfrac{3}{5})$ & 40& 40U & $(\sfrac{54}{5},\sfrac{2}{5})$ & 40V& \cellcolor{green!60}40W & 34W & 42 & 38W\\
42&(\sfrac{43}{60},\sfrac{7}{60}) & $2$ & $(\sfrac{34}{5},\sfrac{2}{5})$ & 42& 42U & $(\sfrac{106}{5},\sfrac{3}{5})$ & 42V& \cellcolor{green!60}42W & 32W & 40 & 36\\
44&(\sfrac{41}{60},\sfrac{29}{60}) & $0$ & $(\sfrac{38}{5},\sfrac{4}{5})$ & 44& 44U & $(\sfrac{62}{5},\sfrac{1}{5})$ & 44V& 44W & 54W & 50 & 46W\\
46&(\sfrac{17}{20},\sfrac{13}{20}) & $4$ & $(\sfrac{18}{5},\sfrac{4}{5})$ & 46& 46U & $(\sfrac{42}{5},\sfrac{1}{5})$ & 46V& 46W & 52W & 48 & 44W\\
48&(\sfrac{49}{60},\sfrac{1}{60}) & $2$ & $(\sfrac{22}{5},\sfrac{1}{5})$ & 48& 48U & $(\sfrac{118}{5},\sfrac{4}{5})$ & 48V& 48W & 50W & 46 & 54\\
50&(\sfrac{59}{60},\sfrac{11}{60}) & $0$ & $(\sfrac{2}{5},\sfrac{1}{5})$ & 50& 50U & $(\sfrac{98}{5},\sfrac{4}{5})$ & 50V& 50W & 48W & 44 & 52\\
52&(\sfrac{7}{20},\sfrac{3}{20}) & $4$ & $(\sfrac{78}{5},\sfrac{4}{5})$ & 52& 52U & $(\sfrac{102}{5},\sfrac{1}{5})$ & 52V& 52W & 46W & 54W & 50\\
54&(\sfrac{31}{60},\sfrac{19}{60}) & $2$ & $(\sfrac{58}{5},\sfrac{4}{5})$ & 54& 54U & $(\sfrac{82}{5},\sfrac{1}{5})$ & 54V& 54W & 44W & 52W & 48\\\bottomrule
    \end{tabular}
    \caption{Two-dimensional irreducible representations $\varrho$ of $\SL_2(\IZ)$ with finite image and $\varrho(S)^2=+1$. The representations highlighted green are admissible.}
    \label{tab:evenweight}
\end{footnotesize}
\end{table}

\clearpage

\section{Lie Algebras and Current Algebras}\label{app:liecurrentalgebras}

In this section, we summarize a number of facts about Lie algebras and affine Kac--Moody algebras that are relevant for our analysis.

Our classification requires us to enumerate cosets of the form $\mathcal{A}\big/\mathfrak{h}_k$, where $\mathcal{A}$ is a chiral algebra and $\mathfrak{h}_k$ is the level $k$ affine Kac--Moody algebra based on a simple Lie algebra $\mathfrak{h}$. This in turn requires us to understand all subalgebras $\mathfrak{h}_k\subset \mathcal{A}$ up to equivalence. Recall (cf.\ \S\ref{subsec:coset}) that we consider two subalgebras of a chiral algebra $\mathcal{A}$ to be equivalent if there is an automorphism/symmetry of $\mathcal{A}$ which maps one to the other. 

To this end, let 
\begin{align}
   (\mathfrak{g}_{1})_{k_1}\otimes\cdots\otimes (\mathfrak{g}_n)_{k_n}\subset \mathcal{A}
\end{align}
be the Kac--Moody subalgebra of $\mathcal{A}$, excluding any Abelian factors. Here, $(\mathfrak{g}_i)_{k_i}$ is the level $k_i$ affine Kac--Moody algebra based on the simple Lie algebra $\mathfrak{g}_i$. Unitarity forces the levels $k_i$ of a Kac--Moody subalgebra inside of any chiral algebra to be non-negative integers. Because $\mathfrak{h}_k$ is generated by its dimension 1 currents, an embedding $\mathfrak{h}_k\hookrightarrow\mathcal{A}$ is completely determined by the induced embedding of ordinary Lie algebras 
\begin{align}\label{eqn:dimension1embedding}
    \mathfrak{h} \hookrightarrow \mathfrak{g}_1\oplus\cdots\oplus \mathfrak{g}_n
\end{align}
which describes how the space of dimension 1 operators of $\mathfrak{h}_k$ sits inside the space of dimension 1 operators of $\mathcal{A}$. 

In order for a Lie algebra embedding of the form  \eqref{eqn:dimension1embedding} to extend to an embedding 
\begin{align}\label{eqn:affineembedding}
    \mathfrak{h}_k\hookrightarrow (\mathfrak{g}_1)_{k_1}\otimes \cdots\otimes (\mathfrak{g}_{n})_{k_n} \hookrightarrow\mathcal{A}
\end{align}
of chiral algebras, it must satisfy an additional property: namely, its \emph{embedding indices} must be compatible with the levels $k,k_1,\dots,k_n$. There are two useful definitions of the index $x(\mathfrak{h}\hookrightarrow\mathfrak{g})$ associated to an embedding $\mathfrak{h}\hookrightarrow \mathfrak{g}$ of one simple Lie algebra into another. \\

\noindent\textbf{Projection matrix:} There is a projection matrix $\cal P$, unique up to Weyl reflections, which projects any weight $\lambda$ of $\mfg$ onto a weight $\mu$ of $\mfh$. Consider  $\theta_{\mfg},\theta_{\mfh}$, the highest roots of $\mfg,\mfh$ respectively. In terms of these, the {\em embedding index} $x(\mfh\hookrightarrow\mfg)$ is defined as
\be
x(\mfh\hookrightarrow\mfg)\equiv \frac{|{\cal P}\theta_{\mfg}|^2}{|\theta_{\mfh}|^2}.
\label{embex}
\ee
This quantity is a positive integer if the embedding is non-trivial, and 0 if it is trivial.\\

\noindent\textbf{Branching rules:} Any irreducible representation $\mathcal{R}_\lambda$ of $\mathfrak{g}$ can be decomposed into irreducible representations $\mathcal{R}_\mu$ of $\mathfrak{h}$ as 
\begin{align}
    \mathcal{R}_\lambda = \bigoplus_\mu b_{\lambda\mu} \mathcal{R}_\mu.
\end{align}
The coefficients $b_{\lambda\mu}$ are called {\em branching coefficients}, non-negative integers that label the multiplicity of $\mathcal{R}_\mu$ inside $\mathcal{R}_\lambda$. Given the coefficients $b_{\lambda\mu}$ for any single $\lambda$, the embedding index is determined as
\be
x(\mfh\hookrightarrow\mfg)=\sum_\mu b_{\lambda\mu}\frac{x_\mu}{x_\lambda}.
\ee
In this formula, $x_\lambda$ and $x_\mu$ are the {\em Dynkin indices} (or simply {\em indices}) associated to the irreducible representations $\mathcal{R}_\lambda$ and $\mathcal{R}_\mu$, respectively. For example, 
\begin{align}
    x_\lambda = \frac{\dim \mathcal{R}_\lambda}{2\dim \mathfrak{g}}(\lambda,\lambda+2\rho)
\end{align}
where $\rho$ is the Weyl vector, 
\begin{align}
    \rho = \frac{1}{2} \sum_{\alpha\in \Delta_+} \alpha ,
\end{align}
defined as half the sum of the positive roots;
 $x_\mu$ is defined similarly. \\

\noindent We can now describe when the Lie algebra embedding \eqref{eqn:dimension1embedding} extends to the affine embedding \eqref{eqn:affineembedding}. It is precisely when the levels $k,k_1,\dots,k_n$ are related as 
\begin{align}
    k = \sum_i x(\mathfrak{h}\hookrightarrow \mathfrak{g}_i) k_i
\end{align}
where $x(\mathfrak{h}\hookrightarrow\mathfrak{g}_i)$ is the index of the induced embedding $\mathfrak{h}\hookrightarrow\mathfrak{g}_i$. One consequence of this is that when $k=1$, the affine Kac--Moody algebra $\mathfrak{h}_1$ must sit entirely inside a single simple factor $(\mathfrak{g}_{i})_{k_i}$ which has $k_i=1$ and $x(\mathfrak{h}\hookrightarrow\mathfrak{g}_i)=1$, and must embed trivially into every $(\mathfrak{g}_j)_{k_j}$ with $j\neq i$.
 
Another useful fact is that the coset $\mathcal{A}\big/ \mathfrak{h}_k$ inherits a continuous global symmetry algebra from $\mathcal{A}$ (which is necessarily a reductive Lie algebra) given by the centralizer
\begin{align}\label{eqn:centralizer}
    C(\mathfrak{h})=\mathfrak{p}_1\oplus \cdots\oplus\mathfrak{p}_m\oplus \mathsf{U}_1^r
\end{align}
where the $\mathfrak{p}_j$ are simple Lie algebras, $\mathsf{U}_1^r$ is the Abelian part, and the centralizer is taken inside of the continuous global symmetry algebra of $\mathcal{A}$. In fact, one can say more: each $\mathfrak{p}_{j}$ is the dimension 1 subspace of an affine Kac--Moody algebra 
\begin{align}
 (\mathfrak{p}_j)_{k'_j}\hookrightarrow  \mathcal{A}\big/\mathfrak{h}_k   \hookrightarrow\mathcal{A}
\end{align} 
which embeds into the coset. The level $k'_j$ of this affine Kac--Moody algebra is determined as before through the formula 
\begin{align}
    k'_j = \sum_i   x(\mathfrak{p}_j\hookrightarrow \mathfrak{g}_i )k_i.
\end{align}
The currents associated to the Abelian factors in Eq.\ \eqref{eqn:centralizer} similarly sit inside their own Kac--Moody subalgebras.

Now, let us consider two isomorphic Kac--Moody subalgebras $\mathfrak{h}_1,\mathfrak{h}_1'\subset \mathcal{A}$ at level $k=1$. We are interested in whether these define distinct cosets $\mathcal{A}\big/\mathfrak{h}_1$, $\mathcal{A}\big/\mathfrak{h}'_1$, or whether the two cosets are isomorphic. As we explained in \S\ref{subsec:coset}, if there is an automorphism $X:\mathcal{A}\to\mathcal{A}$ which maps $X(\mathfrak{h}_1)=\mathfrak{h}'_1$, then the cosets will be isomorphic, so we are interested in determining when such an automorphism $X$ exists. 

First, note that because $k=1$, both $\mathfrak{h}_1$ and $\mathfrak{h}_1'$ must each be embedded into a single level 1 factor of the Kac--Moody subalgebra of $\mathcal{A}$, i.e.\ $\mathfrak{h}_1\hookrightarrow (\mathfrak{g}_i)_{1}$ and $\mathfrak{h}_1'\hookrightarrow (\mathfrak{g}_j)_{1}$, respectively. Let us start by considering the special case that $i=j$, i.e.\ the situation that they are embedded into the same factor $\mathfrak{g}_1:=(\mathfrak{g}_i)_{1}=(\mathfrak{g}_j)_{1}$. Exponentials of the zero-modes of the Noether currents in $\mathfrak{g}_{1}$ lift to Lie group symmetries of $\mathcal{A}$ which we could use to attempt to relate $\mathfrak{h}_1$ and $\mathfrak{h}_1'$. At the level of ordinary Lie algebras, this translates to the question of whether $\mathfrak{h},\mathfrak{h}'\subset \mathfrak{g}$ are equivalent subalgebras, in the sense of being taken into one another by the adjoint Lie group action of the exponential of $\mathfrak{g}$. If the answer is yes, then the cosets $\mathcal{A}\big/\mathfrak{h}_1$ and $\mathcal{A}\big/\mathfrak{h}_1'$ are isomorphic. To this end, we have the following proposition.

\begin{proposition}\label{prop:subalgebras}
Let $\mathfrak{g}$ be any simple Lie algebra such that $\mathfrak{g}_1\subset\mathcal{A}$ for some $c=24$ chiral algebra $\mathcal{A}$ with one primary operator (see Table \ref{tab:centralizers} for a superset of $\mathfrak{g}$ which arise in this way). Furthermore, let $\mathfrak{h},\mathfrak{h}'$ be Lie subalgebras of $\mathfrak{g}$ with embedding indices $x(\mathfrak{h}\hookrightarrow\mathfrak{g})=x(\mathfrak{h}'\hookrightarrow\mathfrak{g})=1$ such that $\mathfrak{h}$ and $\mathfrak{h}'$ are both isomorphic to one of $\A[1]$, $\G[2]$, $\F[4]$, or $\E[7]$. Then $\mathfrak{h}$ and $\mathfrak{h}'$ are equivalent subalgebras inside of $\mathfrak{g}$. 
\end{proposition}

\noindent Before substantiating this proposition, we comment on its implications. Proposition \ref{prop:subalgebras} shows that if $\mathcal{A}$ is a $c=24$ theory with one primary, then any two cosets $\mathcal{A}\big/\mathfrak{h}_1$ and $\mathcal{A}\big/\mathfrak{h}'_1$ are isomorphic if $\mathfrak{h}_1, \mathfrak{h}_1'$ are both isomorphic to $\A[1,1]$, $\G[2,1]$, $\F[4,1]$, or $\E[7,1]$, and are both embedded into the same Kac--Moody factor $\mathfrak{g}_1\subset \mathcal{A}$. 

This leaves just the case that $\mathfrak{h}_1$ lies inside one factor $\mathfrak{g}_1\subset\mathcal{A}$, and $\mathfrak{h}_1'$ lies inside another factor $\mathfrak{g}'_1\subset\mathcal{A}$. In this case,  the cosets $\mathcal{A}\big/\mathfrak{h}_1$ and $\mathcal{A}\big/\mathfrak{h}_1'$ are isomorphic if there is an automorphism $X:\mathcal{A}\to\mathcal{A}$ which maps $\mathfrak{g}_1$ into $\mathfrak{g}_1'$. Indeed, $X$ maps $\mathfrak{h}_1$ to a subalgebra $X(\mathfrak{h}_1)\subset \mathfrak{g}_1'$, and then one could use the proposition above to produce an automorphism $Y$ which rotates $X(\mathfrak{h}_1)$ into $\mathfrak{h}_1'$. Whether such an $X$ exists can be determined using the (outer) automorphism groups computed in \cite{betsumiya2022automorphism,lamshimakura}. See Proposition \ref{prop:outerauts} in the main text.

\begin{proof}
First, we note that there is a different useful notion of equivalence of Lie subalgebras, called \emph{linear equivalence}. Two subalgebras $\mathfrak{h},\mathfrak{h}'\subset\mathfrak{g}$ are said to be linearly equivalent in $\mathfrak{g}$ if, for every representation $\rho:\mathfrak{g}\to \mathfrak{gl}(V)$, the subalgebras $\rho(\mathfrak{h}),\rho(\mathfrak{h}')\subset  \mathfrak{gl}(V)$ are conjugate under $\textsl{GL}(V)$. Now, let $\mathfrak{g}$ be a simple Lie algebra which appears in Table \ref{tab:centralizers}, and let $N$ be the number of linear equivalence classes of subalgebras $\mathfrak{h}$ such that $x(\mathfrak{h}\hookrightarrow\mathfrak{g})=1$ and $\mathfrak{h}$ is isomorphic to $\A[1], \G[2], \F[4]$, or $\E[7]$. Then, one can use the computer algorithms of  \cite{de2011constructing,GAP4} to conclude that $N=0$ or $N=1$.

Now, two subalgebras which are equivalent are necessarily linearly equivalent, but the converse is not necessarily true in general. According to Theorem 3 of \cite{minchenko2006semisimple}, if $\mathfrak{g}\cong \A[n],\B[n],\C[n],\G[2],\F[4]$ or if $\mathfrak{h}\cong \A[1]$, then the notions of linear equivalence and ordinary equivalence coincide, and so the considerations of the previous paragraph show that the proposition holds in these cases. The remaining cases which we must check are the following, 
\begin{align}\label{eqn:exceptioncases}
    (\mathfrak{h},\mathfrak{g}) \in  \{(\G[2],\D[n]),(\G[2],\E[n]), (\F[4],\E[n]), (\E[7],\E[8])\}.
\end{align}
Theorem 4 of \cite{minchenko2006semisimple} explains that if the fundamental representation of $\D[n]$, when decomposed into representations of $\mathfrak{h}\subset \D[n]$, contains an odd-dimensional orthogonal subrepresentation, then the notions of linear equivalence and equivalence coincide. This is true in the case that $\mathfrak{h}\cong \G[2]$, because the fundamental representation of $\D[n]$ decomposes as $\mathbf{7}\oplus \mathbf{1}\oplus\cdots\oplus\mathbf{1}$, and $\mathbf{7}$ is an orthogonal representation according to Table 1 of \cite{bourbaki2008lie}. One similarly concludes that there is a unique equivalence class of subalgebras with embedding index 1 in the remaining cases in Eq.\ \eqref{eqn:exceptioncases} by consulting the various tables in \cite{minchenko2006semisimple}.
\end{proof}

Table \ref{tab:centralizers} lists centralizers of each of the subalgebras appearing in Proposition \ref{prop:subalgebras}, as well as the embedding indices of each of their simple factors. These were computed using Gap \cite{GAP4}.

\begin{table}[]
    \centering
    \begin{tabular}{cc}
    \begin{tabular}[t]{c|c|c|c|c}
        $\mathfrak{g}\backslash\mathfrak{h}^{(1)}$ & $\A[1]^{(1)}$ & $\G[2]^{(1)}$ & $\F[4]^{(1)}$ & $\E[7]^{(1)}$ \\\midrule
        $\A[1]$ & 0 & -- & -- & --\\
$\A[2]$ & $\mathsf{U}_1$ & -- & -- & --\\
$\A[3]$ & $\A[1]^{(1)}\mathsf{U}_1$ & -- & -- & --\\
$\A[4]$ & $\A[2]^{(1)}\mathsf{U}_1$ & -- & -- & --\\
$\A[5]$ & $\A[3]^{(1)}\mathsf{U}_1$ & -- & -- & --\\
$\A[6]$ & $\A[4]^{(1)}\mathsf{U}_1$ & -- & -- & --\\
$\A[7]$ & $\A[5]^{(1)}\mathsf{U}_1$ & -- & -- & --\\
$\A[8]$ & $\A[6]^{(1)}\mathsf{U}_1$ & -- & -- & --\\
$\A[9]$ & $\A[7]^{(1)}\mathsf{U}_1$ & -- & -- & --\\
$\A[10]$ & $\A[8]^{(1)}\mathsf{U}_1$ & -- & -- & --\\
$\A[11]$ & $\A[9]^{(1)}\mathsf{U}_1$ & -- & -- & --\\
$\A[12]$ & $\A[10]^{(1)}\mathsf{U}_1$ & -- & -- & --\\
$\A[15]$ & $\A[13]^{(1)}\mathsf{U}_1$ & -- & -- & --\\
$\A[17]$ & $\A[15]^{(1)}\mathsf{U}_1$ & -- & -- & --\\
$\A[24]$ & $\A[22]^{(1)}\mathsf{U}_1$ & -- & -- & --\\
$\B[2]$ & $\A[1]^{(1)}$ & -- & -- & --\\
$\B[3]$ & $\A[1]^{(2)}\A[1]^{(1)}$ & 0 & -- & --\\
$\B[4]$ & $\B[2]^{(1)}\A[1]^{(1)}$ & $\mathsf{U}_1$ & -- & --\\
$\B[5]$ & $\B[3]^{(1)}\A[1]^{(1)}$ & $\A[1]^{(1)}\A[1]^{(1)}$ & -- & --\\
$\B[6]$ & $\B[4]^{(1)}\A[1]^{(1)}$ & $\A[3]^{(1)}$ & -- & --\\
$\B[7]$ & $\B[5]^{(1)}\A[1]^{(1)}$ & $\D[4]^{(1)}$ & -- & --\\
$\B[8]$ & $\B[6]^{(1)}\A[1]^{(1)}$ & $\D[5]^{(1)}$ & -- & --
    \end{tabular}
    &
    \begin{tabular}[t]{c|c|c|c|c}
        $\mathfrak{g}\backslash\mathfrak{h}^{(1)}$ & $\A[1]^{(1)}$ & $\G[2]^{(1)}$ & $\F[4]^{(1)}$ & $\E[7]^{(1)}$ \\\midrule
        $\C[3]$ & $\B[2]^{(1)}$ & -- & -- & --\\
$\C[4]$ & $\C[3]^{(1)}$ & -- & -- & --\\
$\C[5]$ & $\C[4]^{(1)}$ & -- & -- & --\\
$\C[6]$ & $\C[5]^{(1)}$ & -- & -- & --\\
$\C[7]$ & $\C[6]^{(1)}$ & -- & -- & --\\
$\C[8]$ & $\C[7]^{(1)}$ & -- & -- & --\\
$\C[10]$ & $\C[9]^{(1)}$ & -- & -- & --\\
$\D[4]$ & $\A[1]^{(1)}\A[1]^{(1)}\A[1]^{(1)}$ & 0 & -- & --\\
$\D[5]$ & $\A[3]^{(1)}\A[1]^{(1)}$ & $\A[1]^{(2)}$ & -- & --\\
$\D[6]$ & $\D[4]^{(1)}\A[1]^{(1)}$ & $\B[2]^{(1)}$ & -- & --\\
$\D[7]$ & $\D[5]^{(1)}\A[1]^{(1)}$ & $\B[3]^{(1)}$ & -- & --\\
$\D[8]$ & $\D[6]^{(1)}\A[1]^{(1)}$ & $\B[4]^{(1)}$ & -- & --\\
$\D[9]$ & $\D[7]^{(1)}\A[1]^{(1)}$ & $\B[5]^{(1)}$ & -- & --\\
$\D[10]$ & $\D[8]^{(1)}\A[1]^{(1)}$ & $\B[6]^{(1)}$ & -- & --\\
$\D[12]$ & $\D[10]^{(1)}\A[1]^{(1)}$ & $\B[8]^{(1)}$ & -- & --\\
$\D[16]$ & $\D[14]^{(1)}\A[1]^{(1)}$ & $\B[12]^{(1)}$ & -- & --\\
$\D[24]$ & $\D[22]^{(1)}\A[1]^{(1)}$ & $\B[20]^{(1)}$ & -- & --\\
$\E[6]$ & $\A[5]^{(1)}$ & $\A[2]^{(2)}$ & 0 & --\\
$\E[7]$ & $\D[6]^{(1)}$ & $\C[3]^{(1)}$ & $\A[1]^{(3)}$ & 0\\
$\E[8]$ & $\E[7]^{(1)}$ & $\F[4]^{(1)}$ & $\G[2]^{(1)}$ & $\A[1]^{(1)}$\\
$\F[4]$ & $\C[3]^{(1)}$ & $\A[1]^{(8)}$ & 0 & --\\
$\G[2]$ & $\A[1]^{(3)}$ & 0 & -- & --
    \end{tabular}
    \tabularnewline
    \end{tabular}
    \caption{Centralizers of $\A[1]^{(1)}$, $\G[2]^{(1)}$, $\F[4]^{(1)}$, and $\E[7]^{(1)}$ in simple Lie algebras with small rank. Each entry is $\textsl{C}_{\mathfrak{g}}(\mathfrak{h})$ where $\mathfrak{h}\subset\mathfrak{g}$ is the unique subalgebra (up to equivalence) with embedding index $x(\mathfrak{h}\hookrightarrow\mathfrak{g})=1$. The notation $\mathsf{X}_r^{(x)}$ indicates that the rank $r$ simple Lie algebra $\mathsf{X}_r$ sits inside $\mathfrak{g}$ with embedding index $x$. For example, the centralizer of $\mathsf{G}_2$ inside $\mathsf{F}_4$ (embedded with index 1) is $\mathsf{A}_1$, which sits inside of $\F[4]$ with index $8$. A dash -- indicates that there is no embedding $\mathfrak{h}\hookrightarrow\mathfrak{g}$ with index 1, and $0$ indicates that the embedding has a trivial centralizer.}
    \label{tab:centralizers}
\end{table}

\clearpage

\section{Holomorphic and Skew-Holomorphic Jacobi Forms}\label{app:penumbral}

In \S\ref{subsec:modulardata}, we classified two-dimensional admissible representations of $\SL_2(\IZ)$, and in \S\ref{subsec:admissiblecharacters}, we computed vector-valued modular forms for these representations using modular linear differential equations. In this appendix, we telegraphically describe an alternative method for computing vector-valued modular forms for the $\A$-type representations, i.e.\ the representations 
\begin{align}
    \varrho_{\A}, ~\omega \varrho_{\A}, ~\omega^2\varrho_{\A}, ~\varrho_{\A}^\ast, ~\omega\varrho_{\A}^\ast, ~\omega^2\varrho_{\A}^\ast
\end{align}
where $\varrho_{\A}$ is defined in Eq.\ \eqref{eqn:modularrepAG}. The key will be to relate vector-valued modular forms for the representations $\omega^n \varrho_{\A}^\ast$ to holomorphic Jacobi forms of index $1$ \cite{eichler1985theory}, and vector-valued modular forms for the representations $\omega^n\varrho_{\A}$ to skew-holomorphic Jacobi forms of index $1$ \cite{skoruppa1989developments}, both of which have been more or less classified. Part of our motivation for highlighting this connection is to facilitate future comparisons of theories with two primaries to Umbral Moonshine \cite{Cheng:2012tq,cheng2014umbral} and Penumbral Moonshine \cite{harvey2015traces,duncan2021overview,duncan2022modular,duncan2022two} (though it is ``mock'' \cite{zwegers2002mock} Jacobi forms which appear in the former setting, and so one should expect a structure more subtle than ordinary unitary RCFT to categorify the functions there, see e.g.\ \cite{duncan2017umbral,anagiannis2019k3,Cheng:2022npj}).

We follow the exposition of \cite{duncan2021overview,cheng2016optimal} closely. First, we say that a function $\varphi:\mathbb{H}\times \IC\to\IC$ is an \emph{elliptic form of index} $m$ if $\varphi$ is holomorphic, and if 
\begin{align}\label{eqn:elliptictransformation}
    \varphi(\tau,z+\lambda\tau+\mu) = e^{-2\pi i (m\lambda^2\tau +2m\lambda z)}\varphi(\tau,z)
\end{align}
for any integers $\lambda,\mu\in\mathbb{Z}$.
Any elliptic form admits a \emph{theta-decomposition} of the form 
\begin{align}\label{eqn:thetadecomposition}
    \varphi(\tau,z) = \sum_{r~\mathrm{mod}~2m} h_r(\tau) \theta_{m,r}(\tau,z)
\end{align}
where $h_r(\tau)$ are arbitrary functions, and 
\begin{align}
    \theta_{m,r}(\tau,z) = \sum_{s\equiv r~\mathrm{mod}~2m} q^{s^2/4m}e^{2\pi i z s}.
\end{align}
To see this, note that taking $\lambda=0$ and $\mu=1$ in Eq.\ \eqref{eqn:elliptictransformation} implies that $\varphi$ is periodic under $z\mapsto  z+1$, and hence admits a Fourier expansion of the form 
\begin{align}
    \varphi(\tau,z) = \sum_{l\in\mathbb{Z}} b_l(\tau) e^{2\pi i l z}.
\end{align}
Imposing the  $\lambda=1$, $\mu=0$ transformation on this Fourier expansion shows that
\begin{align}
   \sum_{l\in\mathbb{Z}} (e^{2\pi i l \tau} b_l(\tau))e^{2\pi i l z} = e^{-2\pi i m\tau}\sum_{l\in\mathbb{Z}} b_l(\tau)e^{2\pi i (l-2m) z}
\end{align}
which implies that $h_r(\tau)\equiv b_r(\tau)q^{-r^2/4m}$ only depends on the value of $r$ modulo $2m$.

We further say that $\varphi$ is a holomorphic Jacobi form of weight $k$ and index $m$ if it is an elliptic form of index $m$ with holomorphic theta coefficients $h_r(\tau)$ which remain bounded as $\tau\to i\infty$, and  if it satisfies
\begin{align}
    \varphi\left( \frac{a\tau+b}{c\tau+d},\frac{z}{c\tau+d}  \right) = \exp\left( 2\pi i \frac{cmz^2}{c\tau+d}\right)(c\tau+d)^k \varphi(\tau,z).
\end{align}
Similarly, we say that $\varphi$ is a skew-holomorphic Jacobi form of weight $k$ and index $m$ if it is an elliptic form of index $m$ with  anti-holomorphic theta coefficients $h_r(\tau) = \overline{g_r(\tau)}$ which remain bounded as $\tau\to i\infty$, and if it satisfies 
\begin{align}
    \varphi\left( \frac{a\tau+b}{c\tau+d}, \frac{z}{c\tau+d}  \right)= \exp\left( 2\pi i \frac{cmz^2}{c\tau+d}\right)|c\tau+d|(c\bar\tau+d)^{k-1} \varphi(\tau,z).
\end{align}
It will be useful for us to note that the holomorphic and skew-holomorphic Jacobi forms of index 1 have been classified. For example, it is known that the space $J_k$ of holomorphic Jacobi forms of weight $k$ and index $1$ takes the form
\begin{align}\label{eqn:holspace}
    J_k = M_{k-4}(\SL_2(\IZ)) \cdot E_{4,1}(\tau,z) \oplus M_{k-6}(\SL_2(\IZ)) \cdot E_{6,1}(\tau,z)
\end{align}
where $M_k(\SL_2(\IZ))$ is the space of ordinary holomorphic modular forms of weight $k$, and $E_{k,1}(\tau,z)$ is a Jacobi Eisenstein series, which in general is defined as
\begin{align}
\begin{split}
    &E_{k,m}(\tau,z)= \\
    & \ \ \ \ \ \ \frac12\sum_{\substack{c,d\in\mathbb{Z}\\ (c,d)=1}}\sum_{\lambda\in\mathbb{Z}}(c\tau+d)^{-k}\exp 2\pi i m\left( \lambda^2\frac{a\tau+b}{c\tau+d}+2\lambda\frac{z}{c\tau+d}-\frac{cz^2}{c\tau+d}\right).
\end{split}
\end{align}
(See \cite{eichler1985theory} for the Fourier expansion.) Meanwhile, the space $J_k^{\mathrm{sk}}$ of skew-holomorphic Jacobi forms of weight $k$ and index $1$ is 
\begin{align}\label{eqn:skewspace}
    J_k^{\mathrm{sk}} = M_{k-1}(\SL_2(\IZ))\cdot T(\tau,z) \oplus M_{k-3}(\SL_2(\IZ))\cdot U(\tau,z) 
\end{align}
where 
\begin{align}
\begin{split}
    T(\tau,z) &= \sum_{r~\mathrm{mod}~2m}\overline{\theta_{1,r}(\tau,0)}\theta_{1,r}(\tau,z) \\
    U(\tau,z)&=\frac{12}{\pi i}\frac{\partial T}{\partial\bar\tau}(\tau,z)+\overline{E_2(\tau)}T(\tau,z).
\end{split}
\end{align}

Because of the modular transformation properties of the theta functions $\theta_{m,r}(\tau,z)$, it follows that the theta-coefficients $h_r(\tau)$ of a holomorphic Jacobi form constitute a holomorphic vector-valued modular form of weight $k-\frac12$ which transforms as
\begin{align}
    h_r\left( \frac{a\tau+b}{c\tau+d}  \right)= (c\tau+d)^{k-\frac12}\sum_{s~\mathrm{mod}~2m} \varrho^\ast_{\sqrt{2m}\mathbb{Z}}\left(\begin{smallmatrix} a & b \\ c & d \end{smallmatrix}\right)_{r,s} h_s(\tau).
\end{align}
Similarly, the complex-conjugates $g_r(\tau) =\overline{h_r(\tau)}$ of the theta-coefficients of a skew-holomorphic Jacobi form constitute a holomorphic vector-valued modular form of weight $k-\frac12$ which transforms as 
\begin{align}
    g_r\left( \frac{a\tau+b}{c\tau+d}  \right)= (c\tau+d)^{k-\frac12}\sum_{s~\mathrm{mod}~2m} \varrho_{\sqrt{2m}\mathbb{Z}}\left(\begin{smallmatrix} a & b \\ c & d \end{smallmatrix}\right)_{r,s} g_s(\tau).
\end{align}
Here, $\varrho_{\sqrt{2m}\mathbb{Z}}$ is a \emph{projective} representation of $\SL_2(\IZ)$ defined on the generators as 
\begin{align}
    \varrho_{\sqrt{2m}\mathbb{Z}}(T)_{r,s} = e^{\pi i r^2/2m}\delta_{r,s}, \ \ \ \ \varrho_{\sqrt{2m}\mathbb{Z}}(S)_{r,s} = \frac{e^{-\pi i/4 }}{\sqrt{2m}}e^{-\pi i rs/m}.
\end{align}
It is a special case (corresponding to $L=\sqrt{2m}\mathbb{Z}$) of a family of representations $\varrho_L$ known as the Weil representations, which can be associated to any even lattice $L$.

Thus, any skew-holomorphic/holomorphic Jacobi form defines a vector-valued modular form transforming under the Weil representation or its conjugate, and conversely every such vector-valued modular form determines a holomorphic/skew-holomorphic Jacobi form through Eq.\ \eqref{eqn:thetadecomposition}. 

In the case that $m=1$, the Weil representation is a two-dimensional projective representation, and we can ask how it is related to the two-dimensional admissible representations that we classified. The key fact is that 
\begin{align}
    \varrho_{\A} = \varrho_{\sqrt{2m}\mathbb{Z}}/\epsilon
\end{align}
where $\epsilon:\SL_2(\IZ)\to\mathbb{C}^\ast$ is the projective representation describing how the Dedekind eta function transforms. Specifically, it is defined through 
\begin{align}
    \epsilon(T) = e^{\pi i/12}, \ \ \ \ \epsilon(S) = \sqrt{-i}.
\end{align}
(Note that $\epsilon^2 = \zeta$ which was defined in Eq.\ \eqref{eqn:generatoronedimensional}.) Thus, this shows that the most general quasi-character transforming under the representation $\omega^n\varrho_{\A}$ for $n=0,1,2$ with central charge $c=8n+24q+1$ can be expressed as
\begin{align}
   \chi(\tau)= \frac{g(\tau)}{\eta(\tau)^{8n+24q+1}}
\end{align}
where $g(\tau)$ is the holomorphic vector-valued modular form associated to a skew-holomorphic Jacobi form of weight $4n+12q+1$, the most general of which can be determined from Eq.\ \eqref{eqn:skewspace}. Likewise, the most general quasi-character transforming under the representation $\omega^n\varrho_{\A}^\ast$ for $n=0,1,2$ with central charge $c=8n+24q-1$ is 
\begin{align}
    \chi(\tau) = \frac{h(\tau)}{\eta(\tau)^{8n+24q-1}}
\end{align}
where $h(\tau)$ is the vector-valued modular form associated to a holomorphic Jacobi form of weight $4n+12q$, the most general of which can be determined from Eq.\ \eqref{eqn:holspace}. 

As an illustrative example, we have claimed that the characters of $\E[7,1]$ transform covariantly with respect to the representation $\omega\varrho_{\A}^\ast$, and so should be expressible as 
\begin{align}\label{eqn:e7jacobi}
    \chi_{\E[7,1]}(\tau)=h(\tau)/\eta(\tau)^7
\end{align}
where $h(\tau)$ are the theta-coefficients of a holomorphic Jacobi form of weight $4$. The unique weight 4 Jacobi form is $E_{4,1}(\tau,z)$, and one can check that the function one obtains from Eq.\ \eqref{eqn:e7jacobi} agrees with the one reported in Table \ref{tab:healthycharacters} to low orders in the $q$-expansion.

\bibliographystyle{JHEP}

\bibliography{cl24}

\end{document}